\documentclass[11pt]{article}
\usepackage[utf8]{inputenc}
\usepackage[T5,T1]{fontenc}
\usepackage[letterpaper, margin=1in]{geometry}
\usepackage{xcolor}
\usepackage{enumitem}
\usepackage{subcaption}
\usepackage{caption}
\usepackage{algorithm2e}
\usepackage{array}
\usepackage{diagbox}
\usepackage[nodayofweek,level]{datetime}
\usepackage{graphicx}
\usepackage{amsmath}
\usepackage{amssymb}
\usepackage{amsthm}
\usepackage{mathtools}
\usepackage{bbm}
\usepackage{wasysym}

\DeclareSymbolFont{bbold}{U}{bbold}{m}{n}
\DeclareSymbolFontAlphabet{\mathbbold}{bbold}

\theoremstyle{plain}
\newtheorem{theorem}{Theorem}
\newtheorem{lemma}{Lemma}

\newtheorem{corollary}{Corollary}
\newtheorem{conjecture}{Conjecture}

\theoremstyle{definition}
\newtheorem{definition}{Definition}

\newtheorem{sketch}{Sketch}
\newtheorem{problem}{Problem}
\newtheorem{example}{Example}

\theoremstyle{remark}
\newtheorem{remark}{Remark}

\DeclarePairedDelimiter{\ceil}{\lceil}{\rceil}
\DeclarePairedDelimiter{\floor}{\lfloor}{\rfloor}

\DeclarePairedDelimiter{\norm}{\|}{\|}

\DeclarePairedDelimiter{\inner}{\langle}{\rangle}

\newcommand{\ind}[1]{\mathbbm{1}\!\left\{{#1}\right\}}
\newcommand{\var}{\mathbb{V}}
\newcommand{\pr}{\mathbb{P}}
\newcommand{\E}{\mathbb{E}}

\newcommand{\R}{\mathbb{R}}

\newcommand{\Z}{\mathbb{Z}}
\newcommand{\C}{\mathbb{C}}
\newcommand{\N}{\mathbb{N}}
\newcommand{\G}{\mathbb{G}}
\newcommand{\T}{\mathbb{T}}

\newcommand{\supp}{\mathrm{supp}}

\renewcommand{\Re}{\mathfrak{Re}}

\newcommand{\Exp}{\mathrm{Exp}}
\newcommand{\Poisson}{\mathrm{Poisson}}
\newcommand{\Uniform}{\mathrm{Uniform}}

\newcommand{\KMIN}{k\textsf{-Min}}
\newcommand{\ParetoSampler}{\textsf{ParetoSampler}}
\newcommand{\Pareto}{\textsf{Pareto}}

\newcommand{\argmin}{\operatorname{argmin}}

\newcommand{\SamplerWOR}{\textsf{Sampler-WOR}}
\newcommand{\ParetoSamplerWOR}{\textsf{ParetoSampler-WOR}}
\usepackage{hyperref}
\usepackage[capitalise]{cleveref}
\usepackage{multicol}

\crefname{problem}{Problem}{Problems}
\crefname{conjecture}{Conjecture}{Conjectures}

\newcommand{\rb}[2]{\raisebox{#1 mm}[0mm][0mm]{#2}}
\newcommand{\istrut}[2][0]{\rule[- #1 mm]{0mm}{#1 mm}\rule{0mm}{#2 mm}}

\newcommand{\hcm}[1]{\hspace*{#1 cm}}

\newcommand{\ignore}[1]{}
\newcommand{\bydef}{\stackrel{\operatorname{def}}{=}}
\newcommand{\polylog}{\operatorname{polylog}}
\newcommand{\poly}{\operatorname{poly}}

\newcommand{\LK}{L\'evy-Khintchine}
\newcommand{\Levy}{L\'evy}
\newcommand{\LevyTower}{\textsf{L\'evy-Tower}}
\newcommand{\LevyMinSampler}{\textsf{L\'evy-Min-Sampler}}
\newcommand{\PCSA}{\textsf{PCSA}}
\newcommand{\HyperLogLog}{\textsf{HyperLogLog}}

\newcommand{\LevyPCSA}{\textsf{L\'evyPCSA}}
\newcommand{\LevyHyperLogLog}{\textsf{L\'evyHyperLogLog}}
\newcommand{\StableHyperLogLog}{\textsf{Stable-HyperLogLog}}
\newcommand{\StableMinSampler}{\textsf{Stable-Min-Sampler}}
\newcommand{\LevyStable}{\textsf{L\'evy-Stable}}
\newcommand{\Fishmonger}{\textsf{Fishmonger}}
\newcommand{\LevyFishmonger}{\textsf{L\'evyFishmonger}}

\newcommand{\Update}{\mathsf{Update}}
\newcommand{\Geometric}{\operatorname{Geometric}}

\newcommand{\Blasiok}{B\l{}asiok}

\ignore{
\ccsdesc[500]{Mathematics of computing~Probabilistic algorithms}
\ccsdesc[500]{Theory of computation~Sketching and sampling}
\ccsdesc[500]{Theory of computation~Streaming models}
}

\title{A Unified Construction of Streaming Sketches via the \LK{} Representation Theorem\thanks{Supported by NSF Grants CCF-2221980 and CCF-2446604. This manuscript is based on two extended abstracts~\cite{PettieW25-a,PettieW25-b}, 
\emph{Universal Perfect Samplers for Incremental Streams},
presented at SODA 2025, and 
\emph{Sketching, Moment Estimation, and the L\'evy-Khintchine Representation Theorem}, presented at ITCS 2025.
Email: 
\texttt{\{pettie,wangdy\}@umich.edu.}}}
\author{Seth Pettie\\University of Michigan \and Dingyu Wang\\University of Michigan}
\date{}

\begin{document}
\maketitle
\begin{abstract}
In this work we uncover an intimate relation between \emph{\Levy{} processes}
and \emph{data sketches} for generalized moment estimation and weighted sampling.
Let $\mathbf{x}\in (\mathbb{R}^d)^n$ be a vector subject to element updates
and $f : \mathbb{R}^d \to \mathbb{R}_+$.  By definition, the \emph{$f$-moment} 
of $\mathbf{x}$ is
\[
f(\mathbf{x}) = \sum_{v\in [n]} f(\mathbf{x}(v)).
\]
The \emph{$f$-moment estimation} problem is to $(1\pm \epsilon)$-approximate 
$f(\mathbf{x})$ whereas the \emph{$f$-sampling} problem is to select an 
index $v_*\in [n]$ with probability $f(\mathbf{x}(v_*))/f(\mathbf{x})$.
Our primary conceptual contributions are 
\begin{itemize} 
    \item[I.] to draw a close connection between $f$-moment estimation sketches of $\mathbf{x}\in (\mathbb{R}^d)^n$ in the \emph{turnstile} model 
    and generic \Levy{} processes over $\mathbb{R}^d$, and
    \item[II.] to draw a close connection between $f$-samplers
    of $\mathbf{x} \in \mathbb{R}^n$ in the \emph{incremental} model 
    (positive updates only) and one-dimensional, non-negative \Levy{} processes, aka \emph{subordinators}.
\end{itemize}
Through these connections, 
we can apply the powerful \emph{\LK{} representation theorem} from the theory of L\'evy processes 
to design new sketches for $f$-moment estimation and sampling.  Out technical results are as follows.

\begin{itemize}
    \item We give a systematic method for transforming \emph{any} \Levy{} 
    process $X=(X_t)_{t\geq 0}$ in $\R^d$ into an $O(\epsilon^{-2}\log^2 n)$-bit sketch that estimates the $f_X$-moment, where $f_X(z) = -\log \E e^{i\inner{z,X_1}}$ is the characteristic exponent of $X$.
    This method handles essentially all known $f$-moments that can be estimated
    with $\poly(\epsilon^{-1},\log n)$-size sketches~\cite{alon1996space,indyk2006stable,cormode2003comparing,kane2010optimal,ganguly2012estimating,braverman2016streaming,Wang25} 
    in a uniform way, broadens the class of tractable functions, 
    and allows us to estimate 
    \emph{multivariate} functions, when $d>1$.
\item In the one-dimensional incremental setting, we transform any non-negative \Levy{} process (a \emph{subordinator}) $X$ into a $G_X$-sampler,
where $G_X(z) = -\log \E e^{-z X_1}$ is the Laplace exponent of $X$.
These samplers require essentially minimal space,
sample with precisely correct probabilities, and have zero probability of error.
They are distinguished from recent work on 
$G$-samplers~\cite{cohen2019sampling,jayaram2022truly},
which either introduce $(1\pm \epsilon)$-approximation in the probabilities,
a non-zero failure probability, or additional space.
\end{itemize}

\end{abstract}

\section{Introduction}

The study of \emph{sketches} and the \emph{streaming} model 
dates back to the late 1970s, whose early 
work included Morris's~\cite{Morris78} 
approximate counter, 
Munro and Patterson's~\cite{MunroP80} 
selection algorithms,
and
the Boyer-Moore \textsf{MJRTY} algorithm~\cite{BoyerM91}.
In 1983 Flajolet and Martin~\cite{flajolet1985probabilistic} 
designed and analyzed the first ``modern'' 
data sketch called \PCSA,\footnote{\emph{Probabilistic Counting with Stochastic Averaging}}
which estimates the cardinality of a set.
\PCSA{} works only in \emph{incremental} streams whereas 
all linear sketches, 
such as the celebrated \textsf{AMS} sketch
of Alon, Matias, and Szegedy~\cite{alon1996space} 
for estimating $F_2$ moments,
operate on a vector of elements subject 
to both increments and decrements.
After decades of intense investigation into 
streaming and sketching we now have a 
good---but still incomplete---understanding 
of which statistics are \emph{tractable}, 
meaning they are approximable with 
polylogarithmic-size sketches~\cite{braverman2010zero,BravermanC15,braverman2016streaming}.
The purpose of this paper is to 
approach this tractability question from a new direction
and find answers that are not merely \emph{true} 
but have significant \emph{explanatory power}.

\medskip
We shall begin by defining a single \emph{algebraic} 
streaming model that captures the incremental model, 
turnstile model, and others.

\begin{definition}[$M$-turnstile model \cite{wang2023probabilistic}]\label{def:M-turnstile}
Let $(M,+)$ be a commutative monoid with identity $0$. The \emph{$M$-turnstile model} is defined as follows. 
     Let $[n]=\{1,2,\ldots,n\}$ be the universe. 
    The state vector $\mathbf{x}=(\mathbf{x}(1),\ldots,\mathbf{x}(n))\in M^n$ is initialized as $0^n$ and gets updated by a stream of pairs in the form of $(v,y)$, 
    where $v\in [n]$ and $y\in M$.
    \begin{itemize}
        \item $\Update(v,y)$ : $\mathbf{x}(v) \gets \mathbf{x}(v) +y$.
    \end{itemize}
\end{definition}
The assumption that $M$ is commutative and associative allows the updates to be collected in a distributed system.
The point of the $M$-turnstile model is to 
identify \emph{algorithmic functionality}
with the \emph{mathematical structure} of the monoid $(M,+)$. 
For example, if $M$ is idempotent, then duplicated updates are ignored; if $M$ is a group, then ``deletions'' are allowed; if $M$ is a continuous space, then fractional updates are allowed; if $M$ is multi-dimensional, then attribute-wise updates are allowed.  
The usual integer turnstile model corresponds to $\Z$-turnstile while the incremental setting 
corresponds to $\N$-turnstile. 
In this paper we consider 1-dimensional turnstiles ($\mathbb{R}$ or $\mathbb{Z}$), multidimensional turnstiles ($\mathbb{R}^d$), and incremental turnstiles ($\mathbb{R}_+$ or $\mathbb{N}$).

\medskip 

Let us recall three generic streaming problems.  
Here $\mathbb{R}_+$ is the set of 
\emph{non-negative} reals.

\begin{problem}[$f$-moment estimation in the $\R^d$-turnstile model]
\label{prob:f-moment}
Fix $d\in \Z_+$.
Let $f:\R^d \to \C$ be a function with $f(0)=0$, and let $\mathbf{x}\in (\R^d)^n$ be 
    a vector subject to streaming updates to its coordinates.  
    The problem is to estimate the 
    \emph{$f$-moment} $f(\mathbf{x})$, 
    where
    \[
    f(\mathbf{x}) \bydef \sum_{v\in[n]}f(\mathbf{x}(v)).
    \]
\end{problem}

\cref{prob:G-moment} is 
a special case of \cref{prob:f-moment}, but it is nonetheless useful to highlight as a separate problem, due to its connection to \cref{prob:G-sampling} and a separate suite of techniques.

\begin{problem}[$G$-moment estimation in the $\R_+$-turnstile model]
\label{prob:G-moment}
    Let $G:\R_+ \to \R_+$ with $G(0)=0$. Let $\mathbf{x}\in \R_+^n$ be the state vector of the current stream. Estimate the \emph{$G$-moment} 
    $G(\mathbf{x})\bydef \sum_{v\in[n]}G(\mathbf{x}(v))$ over the stream.
\end{problem}

\begin{problem}[$G$-sampler in the $\R_+$-turnstile model]
\label{prob:G-sampling}
    Let $G:\R_+\to \R_+$ be a function 
    with $G(0)=0$ 
    and let $\mathbf{x}\in \R_+^n$ be the 
    vector subject to streaming updates. 
    Return a \emph{$G$-sample} $u\in[n]$ 
    with probability $G(\mathbf{x}(u))/G(\mathbf{x})$.    
\end{problem}
These three problems encompass a large body of 
research in the streaming/sketching literature. 
We give a brief overview of the prior research on these problems.

\begin{description}
    \item[$f$-moments in the $\R$- and $\Z$-turnstile.] Alon, Matias, and Szegedy~\cite{alon1996space} estimate
    the $F_2$-moment ($f(x)=x^2$) with a data structure now commonly known as the \textsf{AMS} sketch. 
    Indyk \cite{indyk2006stable} designed a class of sketches 
    for estimating the $F_p$ moments, with $p\in(0,2]$.
    When $p>2$, Bar-Yossef, Jayram, Kumar, and Sivakumar~\cite{bar2004information} 
    proved that estimating the $F_{p}$-moment requires $\Omega(\poly(n))$ space. 
    The $F_0$-moment ($f(x)=\ind{x\neq 0}$) 
    has been estimated 
    in \emph{three} distinct ways.\footnote{Here $\ind{P}$ is the indicator for $P$, i.e., 
    1 if the predicate $P$ is true, and 0 otherwise.}
    Cormode, Datar, Indyk, and Muthukrishnan~\cite{cormode2003comparing} approximate $F_0$ by $F_\alpha$ with very small $\alpha>0$. 
    Kane, Nelson, and Woodruff~\cite{kane2010optimal} project each element onto $\Z_p$, $p > \epsilon^{-1}\log n$ being a random prime, which effectively reduces $F_0$-estimation to cardinality estimation. 
    Very recently, Wang~\cite{Wang25} estimates $F_0$ by estimating the underlying harmonic components separately and then summing them up. 
    In the $\R^d$-turnstile model, the $F_{p,q}$ hybrid moment is defined by $f(x_1,\ldots,x_d)=(\sum_{j=1}^d |x_j|^p)^q$.
    Ganguly, Bansal, and Dube \cite{ganguly2012estimating} approximate $F_{p,q}$ for $p \in (0,2], q \in (0,1]$ in
    $\tilde{O}(\epsilon^{-2})$ space,
    and Jayram and Woodruff~\cite{JayramW09} give polynomial space bounds on the complexity of $F_{p,q}$ estimation
    for arbitrary values of $p,q$.

    Braverman and Ostrovsky~\cite{braverman2010zero} considered the problem of characterizing the class of \emph{tractable} functions $f$ for which $f$-moments could be approximated to within a $1\pm \epsilon$ factor in $\poly(\epsilon^{-1},\log n)$ space.  
    They managed to characterize all functions $f \colon \Z\to \R_+$ that are symmetric ($f(x)=f(-x)$) and increasing on $[0,\infty)$.
    Braverman, Chestnut, Woodruff, and Yang \cite{braverman2016streaming} extended the characterization to 
    almost all symmetric functions, 
    excluding a class 
    they termed ``nearly periodic functions.''
    \item[$G$-moments in $\R_+$-turnstile.] The incremental setting ($\N$-turnstile or $\R_+$-turnstile) leads to different sketching techniques and different space lower bounds. Flajolet and Martin's~\cite{flajolet1985probabilistic} \PCSA{} sketch estimates the cardinality 
    (number of distinct elements) $\|\mathbf{x}\|_0$, 
    which corresponds to the $G$-moment with $G(x)=\ind{x\neq 0}$. 
    Flajolet, Fusy, Gandouet, and Meunier's~\cite{flajolet2007hyperloglog}
    \HyperLogLog{} sketch is the most widely deployed sketch for 
    cardinality estimation.  
    The most efficient sketches in terms
    of space vs.~estimation error~\cite{pettie2021information,wang2023better,Lang17,DataSketches}
    are based on entropy-compressed versions of \PCSA{} with optimum estimators.  
    The cardinality sketches above are analyzed in the \emph{random oracle} model, where one can evaluate uniformly random hash functions.
    When the sketch stores its own hash functions, Kane, Nelson, and Woodruff~\cite{kane2010optimal} and \Blasiok~\cite{Blasiok20} 
    designed cardinality sketches with an $(\epsilon,\delta)$-guarantee 
    meeting the $\Omega(\log n + \epsilon^{-2}\log\delta^{-1})$ 
    lower bound~\cite{IndykW03,JayramW13,alon1996space}, up to a large constant factor.
    
    The problem of $F_p$-moment estimation ($G(x)=x^p$) can be solved by Indyk's sketches for $p\in(0,2]$. Nevertheless, 
    there are more efficient $F_p$-moment sketches in the incremental setting when $p\in(0,1)$;
    see Cohen \cite{cohen2017hyperloglog} and Jayaram and Woodruff \cite{jayaram2023towards}. Cohen \cite{cohen2017hyperloglog} estimates the 
    $G$-moment, where $G$ is in the class 
    of ``soft concave sublinear functions,'' 
    which are intended to approximate cap-statistics
    of the form $G(x) = \min\{x,C\}$.
    \item[$G$-sampling in $\R_+$-turnstile.] Vitter's \cite{vitter1985random}
    \emph{reservoir sampling} can be regarded as a $G$-sampler in the $\N$-turnstile model, where $G(x)=x$, i.e., elements are sampled proportional to their counts.  Cohen's~\cite{cohen1997size}
    \textsf{Min-Sampler} hashes elements and samples the index $v$, $\mathbf{x}(v)>0$, having the smallest hash value.  The \textsf{Min-Sampler} is insensitive to duplicates and thus solves the $G$-sampling problem with $G(x)=\ind{x>0}$ (disinct sampling). 
    
    For generic $G$-sampling, Cohen and Geri \cite{cohen2019sampling} convert the $G$-estimators in \cite{cohen2017hyperloglog} into \emph{approximate} $G$-samplers for soft concave sublinear functions. 
    Generic $G$-samplers with precisely correct sampling probabilities have only been studied recently, by Jayaram, Woodruff, and Zhou \cite{jayaram2022truly}, where they combine reservoir sampling and rejection sampling to ensure correct sampling probabilities, 
    conditioned on a sample being returned.
\end{description}

\subsection{A New Perspective}

The main take-away message of the present work is that 

\begin{center}
    \emph{$f$-moment estimation in the $\R^d$-turnstile model and $G$-moment estimation and $G$-sampling in the $\R_+$-turnstile model can all be done in a uniform manner by simulating \Levy{} processes.}
\end{center}

\emph{\Levy{} processes} have been studied since the early 20th century,
and have been used to model phenomena in various fields, e.g., 
physics (how does a gas particle move?) and finance (how does the stock price change?). 
We will show that all \Levy{} processes have algorithmic interpretations 
in the context of streaming sketches, and  
the fundamental \emph{\LK{} representation theorem} leads to a unified view of sketching for moment estimation 
and sampling.  To our knowledge, this
is the first application of the \LK{} theorem
in theoretical computer science.

\medskip 

We give a detailed technical synopsis of \Levy{} processes in \cref{sec:prelin_levy}. 
For the time being, a \Levy{} process $(X_t)_{t\geq 0}$, where $X_t \in \R^d$,
is defined by having independent, stationary increments.  I.e., for any $t_1,t_2\in \R_+$, 
$(X_{t_1+t_2} - X_{t_1}) \sim X_{t_2}$, 
and for any $t_1 < t_2 < \cdots < t_k$,
the increments
$X_{t_1},X_{t_2}-X_{t_1},\ldots,X_{t_k}-X_{t_{k-1}}$ 
are mutually independent.
Some common one-dimensional \Levy{} processes are
\begin{description}
    \item[Linear drift.] $X_t = \gamma t$ with drift rate $\gamma$.
    \item[Wiener process/Brownian motion.] $X_t \sim \mathcal{N}(0,t\sigma^2)$ with variance $\sigma^2$.
    \item[Poisson process.] $X_t \sim \Poisson(\lambda t)$ with rate $\lambda$.
    \item[Compound Poisson process.] $X_t = \sum_{k=1}^Z J_k$, where the number of \emph{jumps} $Z\sim \Poisson(\lambda t)$ follows a Poisson process, while each jump $J_k\sim \mathcal{J}$ is sampled independently from a common \emph{jump distribution}.
    
    \item[$\alpha$-stable process.] Defined by $X_t/t^{1/\alpha} \sim X_1\sim\text{$\alpha$-stable}$. (For $\alpha\in(0,2]$, the $\alpha$-stable random variable $X_1$ on $\R$ is defined by $\E e^{i z X_1}=e^{-|z|^\alpha}$, for any $z\in \R$.)
\end{description}
The \Levy{} processes defined above can all be generalized to higher dimensions.  
They may also be combined, as \Levy{} processes are closed under addition.

\medskip 

The characteristic function $\varphi_X(z) = \E e^{izX}$ 
of a random variable $X\in \mathbb{R}$ 
is equivalent to the Fourier transform of its pdf, 
when its pdf exists.
It uniquely determines the distribution of $X$.
For example, whenever it exists, 
the $k$th moment of $X$ can be recovered
from the $k$th derivative of $\varphi_X$ evaluated at zero.
When $X\in \mathbb{R}^d$, $\varphi_X(z) = \E e^{i\inner{X,z}}$,
where $z\in \mathbb{R}^d$ and $\inner{X,z}$ is the inner product.

The \LK{} representation theorem identifies every \Levy{}
process $X$ with its \emph{characteristic exponent} 
$f_X \colon \R^d \to \C$, 
where for any $z\in\R^d$, 
\[
\E \exp(i\inner{X_t,z}) = \exp(-tf_X(z)).
\]
When $X$ is a non-negative \Levy{} process in one dimension
(also known as a \emph{subordinator}), it is characterized by its 
\emph{Laplace exponent} 
$G_X \colon \R_+ \to \R_+$, 
where for any $z\in\R_+$,
\[
\E \exp(-zX_t) = \exp(-tG_X(z)).
\]
Note that the \emph{stationary increments} property 
implies that a non-negative \Levy{} process has only non-negative increments.

\medskip 

\paragraph{Roadmap.} In the remainder of the introduction we give a lightly technical introduction to how 
\Levy{} processes naturally arise in the study
of generalized moment estimation and sampling problems.
In \cref{sect:levy-processes-and-linear-sketching}
we show how linear sketches~\cite{alon1996space,indyk2006stable} 
explicitly or implicitly use infinitely divisible distributions and \Levy{} processes.  In
\cref{sec:f-moment-estimation-and-levy-processes}
we show a natural connection between the $f$-moment estimation problem, \Levy{} processes, 
and the \LK{} representation theorem, 
paving the way to answer \cref{prob:f-moment} in turnstile streams.
\cref{sec:g-sampler-g-moment-estimation} 
illustrates a connection between \emph{non-negative}
\Levy{} processes (subordinators) and incremental streams, which allows one to address both 
moment estimation (\cref{prob:G-moment}) 
and sampling (\cref{prob:G-sampling}) in incremental streams.  In \cref{sect:new-results} we give a formal statement of our new results, and in \cref{sect:organization} we 
lay out the organization of the remainder of the paper.

\subsection{L\'evy Processes and Linear Sketching}\label{sect:levy-processes-and-linear-sketching}

We briefly explain why L\'evy processes \emph{naturally} lie at the heart of linear sketching. Suppose $X$ is the random memory state of a \emph{linear} sketch\footnote{It suffices to consider linear sketches in the turnstile model \cite{li2014turnstile}.} of a particular 
input vector $\mathbf{x}\in \mathbb{R}^n$, 
with $\hat{f}(X)$ being an estimate of the 
$f$-moment $f(\mathbf{x})$.
Now consider the situation where the input is replicated $w$ times over disjoint domains, 
i.e., $\mathbf{x}_w \in \mathbb{R}^{nw}$,
with $\mathbf{x}_w(v) = \mathbf{x}(v \operatorname{mod} n)$.
By construction, $f(\mathbf{x}_w) = w\cdot f(\mathbf{x})$ and the estimate becomes $\hat{f}(X^{(1)}+\cdots + X^{(w)})$ where $X^{(j)}$ are i.i.d.~copies of $X$.
In other words, we should have
\begin{align*}
    \frac{\hat{f}(X^{(1)}+\cdots + X^{(w)})}{w} \approx f(\mathbf{x}),
\end{align*}
for \emph{any} $w\in \Z_+$ and also for the limiting case as $w\to \infty$. Therefore, no matter how complicated the distribution of the linear sketch may be, the sum $X^{(1)}+\cdots + X^{(w)}$ will converge to some well-behaved limiting distribution, with a proper normalization depending on $X$ and $w$. 
For the \textsf{AMS} sketch~\cite{alon1996space}, 
$X$ happens to be sub-Gaussian and therefore the normalized sum converges to a Gaussian. For Indyk's~\cite{indyk2006stable} stable sketch, $X$ is $\alpha$-stable, and therefore the normalized sum remains $\alpha$-stable. These are merely two special cases of L\'evy processes. In general, if $X=(X_t)_{t\geq 0}$ is a L\'evy process, then 
by the stationary and independent increments properties, $X^{(1)}+\cdots + X^{(w)} = (X^{(1)}_t+\cdots + X^{(w)}_t)_{t\geq 0}\sim (X_{wt})_{t\geq 0}$. 
In other words, \emph{summing} i.i.d.~L\'evy processes is equivalent to simply \emph{rescaling time}. After normalizing the time scale,\footnote{To see how this ``time normalization'' generalizes the typical scalar normalization of stable variables, note that for $\alpha$-stable variables, the normalization is $w^{-1/\alpha}(X^{(1)}+\cdots + X^{(w)})$, which is equivalent to scale the time down by $w$ since $w^{-1/\alpha}X_{wt}\sim X_t$ if $(X_t)_{t\geq 0}$ is $\alpha$-stable.} we see L\'evy processes are stable under i.i.d.~sums. L\'evy processes 
therefore form a mathematical closure of linear sketches in terms of their limiting distributions. 

In practice it usually suffices to construct some algorithmically simple random projection in the \emph{domain of attraction} of the limiting process. 
For example, the \textsf{AMS} sketch for estimating $F_2$ does not need to explicitly use Gaussians in its projection; Rademacher ($\{-1,1\}$) random variables suffice.
For a \emph{prototypical solution}, 
it is convenient to consider linear sketches with projection sampled directly from L\'evy processes.
The distribution of the sketch will always lie in the space of L\'evy processes and is therefore easier 
to track.

\medskip 

We now demonstrate how to convert L\'evy processes to streaming sketches. 
Throughout the paper we work in the \emph{random oracle} model, 
in which we can evaluate uniformly random hash functions $H \colon \Z\to [0,1]$.  
This is not a limiting assumption as it can be removed in a black-box way by using pseudorandom generators against space bounded computation, 
at a small loss in space efficiency; see~\cite{indyk2006stable,nisan1990pseudorandom}.

\subsection{$f$-Moment Estimation and \Levy{} Processes}\label{sec:f-moment-estimation-and-levy-processes}

Given any L\'evy process $X=(X_t)_{t\geq 0}$ on $\R^d$, by
the \LK{} representation theorem (\cref{thm:lk} in \cref{sec:prelin_levy}), there exists a function $f = f_X:\R^d\to \C$ such that for any $t\geq 0$ and $z\in \R^d$, $\E e^{i\inner{z,X_t}}=e^{-t f(z)}$. Suppose now we 
project the input vector to a single cell $C_t$ by sampling the \Levy{} processes at time $t$, as follows.
\begin{align*}
    C_t &= \sum_{v\in[n]}\inner*{X_t^{(v)},\mathbf{x}(v)}.
\end{align*}
Here the $X_t^{(v)}$ are i.i.d.~copies of the \Levy{} process $X_t$. 
Clearly $C_t$ is a linear sketch that can be maintained over a distributed stream. We thus have 
    \begin{align*}
        \E e^{iC_t} &= \E e^{i \sum_{v\in[n]}\inner*{X_{t}^{(v)},\mathbf{x}(v)}} & \text{(definition of $C_t$)}\\
        &= \prod_{v\in [n]} \E e^{i \inner*{X_{t}^{(v)},\mathbf{x}(v)}} & \text{(by independence)}\\
        &= \prod_{v\in [n]} e^{-tf(\mathbf{x}(v))} & \text{(by \LK)}\\
        &=e^{-t f(\mathbf{x})}, & \text{(definition of $f(\mathbf{x}) = \sum_{v\in[n]}f(\mathbf{x}(v))$)}
    \end{align*}
from which the $f$-moment $f(\mathbf{x})$
can be recovered by choosing a suitable time 
$t\approx \Theta(1/|f(\mathbf{x})|)$. 
Of course, we do not know $f(\mathbf{x})$ in advance
and therefore do not know the optimum $t$ in advance, but we achieve reasonable coverage by 
maintaining $(C_t)$ 
for many $t$, evenly 
spaced on a logarithmic scale.
This is the key observation that lets us estimate 
any $f$-moment, so long as 
$f$ is the characteristic 
exponent of some L\'evy process.

\subsection{$G$-Sampler, $G$-Moment Estimation, and Subordinators}\label{sec:g-sampler-g-moment-estimation}

Perhaps surprisingly, the connection between \Levy{} processes and  streaming sketches goes beyond the linear case. We now demonstrate that  the \emph{min-based} sketches/samplers are closely related to \emph{non-negative} \Levy{} processes, also known as subordinators.  
To index the fresh randomness associated with an update, 
we suppose updates occur at times $k=1,2,\ldots$. 
The min-based samplers are based on the following 
\emph{generic min sketch}.

\begin{definition}[generic min sketch]\label{def:generic_min}
A \emph{generic min sketch} is a pair $(v_*,h_*)\in ([n]\cup \{\perp\})\times(\R_+\cup\{\infty\})$ initialized as $(\perp,\infty)$.  
It is associated with a certain hash function 
$H:[n]\times \R_+ \times \Z_+ \to \R_+$ provided by the random oracle.
A vector update $\mathbf{x}(v) \gets \mathbf{x}(v)+\Delta$
is handled as follows.
\begin{description}
    \item[$\Update(v,\Delta)$ (issued at time $k$):]
    If $H(v,\Delta,k)<h_*$, then set $(v_*,h_*)\gets (v,H(v,\Delta,k))$.
\end{description}
\end{definition}

Assume for simplicity that $\Delta=1$ in all updates.
Suppose $\mathbf{x}(v) = w$, 
being incremented at 
times $k_1, \ldots,k_w$. 
The minimum hash value produced 
at index $v$ is thus 
\[
Z_v = \min(H(v,1,k_1),\ldots,H(v,1,k_w)).
\]
The observation is that, no matter what the weight function $G$ is, the probability that $v$ gets sampled should only depend on $w$, rather than the insertion timestamps $(k_1,\ldots,k_w)$. 
Thus, we want the random sequence
\begin{align}
   H(v,1,1), H(v,1,2), H(v,1,3),\ldots\label{eq:Hvk}
\end{align}
to be \emph{exchangeable}.\footnote{A finite random vector 
$(Y_1,\ldots,Y_k)$ is \emph{exchangeable} if its distribution 
remains unchanged after reordering the coordinates. An infinite 
random sequence is exchangeable if any finite prefix is 
exchangeable. See \cite[Definition 2.1]{bruck2023exchangeable} 
for a recent treatment.} De Finetti’s theorem\footnote{This is 
a classic result by De Finetti from 1931 on exchangeable random 
sequences. We use the version described in 
\cite[Theorem 2.2]{bruck2023exchangeable}.}  then implies that 
there is a non-negative random process $X=(X_t)_{t\geq 0}$ such 
that 
\[(H(v,1,k))_{k\in\Z_+} = (\{\inf \{t\geq 0 : X_t > E_k\}\})_{k\in \Z_+},\]
where the $E_k \sim \Exp(1)$ are i.i.d.~standard exponential 
random variables. In addition, in order to sample with the correct 
probability, we need the minimum hash value 
$\min_{j\in[w]}(H(v,1,k_j))$ to distribute as an exponential 
random variable (also observed in \cite{jayaram2021perfect}), 
which will be true if $X$ is a \Levy{} 
process.
Note that $X$ needs to be non-negative in De Finetti’s theorem. 
Now suppose $X$ is a subordinator (non-negative \Levy{} process) 
with Laplace exponent $G_X:\R_+\to \R_+$. Then for $z>0$,
\begin{align*}
    &\pr(\min(H(v,1,k_1),\ldots,H(v,1,k_w))> z)\\
    &=\pr(\min(H(v,1,1),\ldots,H(v,1,w))> z) &\text{((\ref{eq:Hvk}) is exchangeable)}\\
    &=\pr(\min(\inf \{t\geq 0 : X_t > E_1\},\ldots,\inf \{t\geq 0 : X_t > E_w\})> z) & \text{(De Finetti's theorem)}\\
    &=\pr(\inf \{t\geq 0 : X_t > \Exp(w)\}> z)\\
    &=\pr(X_z \leq \Exp(w)) &\text{(non-negativity)}\\
    &=\E(\pr(X_z \leq \Exp(w)) \mid X_z) & \text{(law of total probability)}\\
    &=\E e^{-w X_z}
    = e^{-zG_X(w)}. &\text{(\LK)}
\end{align*}
Note that $\pr(\Exp(\lambda)>z) = e^{-z\lambda}$. 
Therefore we know that
\begin{align*}
    \min(H(v,1,k_1),\ldots,H(v,1,k_w)) \sim \Exp(G_X(w)) = \Exp(G_X(\mathbf{x}(v))).
\end{align*}
I.e., the minimum hash value produced by an element $v$ is \emph{exactly} 
an exponential random variable with rate $G(\mathbf{x}(v))$. This is the underlying intuition for the \LevyMinSampler{} 
studied in \cref{sec:levy_min}. The precise simulation of $\Exp(G_X(w))$ can be used to emulate existing cardinality estimators with the \emph{cardinality} of $\mathbf{x}$ 
($\|\mathbf{x}\|_0 = |\supp(\mathbf{x})|$)
replaced by the \emph{$G$-moment} $G(\mathbf{x})$.
This is the underlying intuition for the \LevyPCSA{} 
and \LevyHyperLogLog{} sketches studied in \cref{sec:levy_pcsa}.

\subsection{New Results}\label{sect:new-results}

We prove two main theorems.  \cref{thm:main:LevyTower} connects the \LK{} representation theorem for generic L\'evy processes with $f$-moment estimation in the $\R^d$-turnstile model, while \cref{thm:main:LevyMinSampler} connects subordinators with $G$-sampling in the $\R_+$-turnstile. 
Refer to~\cref{tab:notation} for some standard notation. 
\begin{table}[]
    \centering
    \begin{tabular}{l|l|l}
        \textsf{Notation} & \textsf{Definition} & \textsf{Notes} \\
        \hline\hline
          $|x|$ &  $(\sum_{j=1}^d x_j^2)^{1/2}$ & $x\in \R^d$, Euclidean norm\\
          $|z|$ &  $(z\Bar{z})^{1/2}$ & $z\in \C$, modulus\\
          $\norm{\mathbf{x}}_0$ & $\sum_{v\in[n]}\ind{\mathbf{x}(v)\neq 0}$ & $\mathbf{x}\in (\R^d)^n$\\
          $\norm{\mathbf{x}}_\infty$ & $\max_{v\in[n]}|\mathbf{x}(v)|$ & $\mathbf{x}\in (\R^d)^n$\\
          $\T$ & complex unit circle & identified by $[0,2\pi)$\\
          $\mathbb{S}_{d-1}$ & $\{x\in\R^d:|x|=1\}$ &  unit sphere in $\R^d$\\
           $\var Z$ & $\E (Z-\E Z) \overline{(Z-\E Z)}$ &  
 variance of complex $Z$\\\hline
    \end{tabular}
    \caption{Notation}
    \label{tab:notation}
\end{table}

\begin{theorem}[\LevyTower{}, \cref{sec:levy_tower}]\label{thm:main:LevyTower}
Let $f:\R^d\to \C$ be any function of the form 
\[
f(z)=\frac{1}{2}\inner*{z,Az} - i\inner*{\gamma, z} +\int_{\R^d}\left(1+i\inner*{z,s}\ind{|s|<1}-e^{i\inner*{z,s}}\right)\,\nu(ds),
\]
where $A$ is a covariance matrix, $\gamma\in\R^d$ is a drift, and $\nu$ is a positive measure over $\R^d$ with $\int_{\R^d}\min\{|s|^2,1\}\,\nu(ds)<\infty$. 
The \LevyTower, parameterized by $f$,
is a mergeable sketch that 
occupies $O(\epsilon^{-2}\log n)$ words.
For any input stream
$\mathbf{x}\in(\R^d)^n$ with $|f(\mathbf{x})|\in[1,\poly(n)]$, 
the \LevyTower{} returns an estimate 
$\widehat{f(\mathbf{x})}$ 
that with probability 99/100 satisfies:
\begin{align*}
    \left|\widehat{f(\mathbf{x})}-f(\mathbf{x})\right|\leq O(\epsilon |f(\mathbf{x})|).
\end{align*}
\end{theorem}

The \LevyTower{} sketch improves our understanding of the tractability of one-dimension function moments. In particular, it implies the tractability of many nearly periodic functions that were not previously classified. 
See \cref{sec:tractability} for a discussion of the new result and the existing tractability results. 
It also implies the tractability of a large class of 
\emph{multidimensional moments} which have not been considered before.

\begin{theorem}[\LevyMinSampler{}, \cref{sec:levy_min_sampler}]\label{thm:main:LevyMinSampler}
Let $G:\R_+\to\R_+$ be any function of the form
    \begin{align*}
        G(z) = c\ind{z>0} +\gamma_0 z + \int_0^\infty (1-e^{-zs})\,\nu(ds),
    \end{align*}
    where $c,\gamma_0,\geq 0$, and $\nu$ is any positive measure such that $\int_0^\infty \min\{s,1\}\,\nu(ds)<\infty$. There is a min-based sketch storing only a single pair $(v_*,h_*) \in [n]\times \mathbb{R}_+$,
    where $v_*\in[n]$ is the 
    sampled index 
    and $h_*\in \R_+$ 
    is the minimum hash value. 
    For any given input vector $\mathbf{x}\in\R_+^n$ it is guaranteed that
\begin{itemize}
    \item For any $u\in[n]$, $\pr(v_*=u)=G(\mathbf{x}(u))/G(\mathbf{x})$.
    \item $h_*\sim \Exp(G(\mathbf{x}))$.
\end{itemize}    
The space required is therefore only 2 words.
\end{theorem}

The raw version of the \LevyMinSampler{} (\cref{sec:raw_levy_min}) is built on the structure theory of min-wise infinitely divisible, exchangeable sequences. See Br\"uck, Mai, and Scherer \cite{bruck2023exchangeable} for a recent treatment. 
Cohen \cite{cohen2019sampling} developed \emph{approximate} $G$-samplers for 
soft concave sublinear functions, which are a subset of Laplace exponents. 
Jayaram, Woodruff, and Zhou~\cite{jayaram2022truly} have developed $G$-samplers with precisely correct sampling probabilities for functions $G:\N\to\R$ such that $\max_{z\in\N} (G(z)-G(z-1))<\infty$ with $O(\frac{\sum_{v\in[n]}\mathbf{x}(v)}{G(\mathbf{x})}\log n)$ bits of space. For $G$ being a Laplace exponent, the factor $\frac{\sum_{v\in[n]}\mathbf{x}(v)}{G(\mathbf{x})}$ in \cite{jayaram2022truly}  can be $\Omega(\poly(n))$ when $\sum_{v\in[n]}\mathbf{x}(v)\gg G(\mathbf{x})$. The significance of \cref{thm:main:LevyMinSampler} is that the sampling probability is \emph{precisely correct} and the space usage is only \emph{two words}.

\medskip

One way to leverage previous sketching research is to \emph{reduce} complex sketching tasks to simpler ones,
which are both well studied and widely deployed.  
We effect these reductions through a number of powerful \emph{emulation} theorems, which create sketches
for complex estimation tasks whose distribution is \emph{identical} to an existing sketch.

\begin{theorem}[\textsf{$F_\alpha$-stable} emulation, \cref{sec:alpha-levy-stable}]\label{them:simulation:stable}
Let $f$ be the characteristic exponent of any $d$-dimensional symmetric $\alpha$-stable process  $X$.
When $\alpha=2$, $f(x)=e^{-\frac{1}{2}t \inner{x,Ax}}$, where
$A$ is a covariance matrix, and when $\alpha\in(0,2)$, 
$f(x)=\int_{\mathbb{S}_{n-1}}|\inner{x,\xi}|^{\alpha}\,\mu(d\xi)$, 
where $\mu$ is a symmetric, positive measure on $\mathbb{S}_{n-1}$.\footnote{For the case $\alpha\in(0,2)$, only symmetric processes are considered here for simplicity.  See~\cite{ken1999levy} for the full characterization of stable processes.}
Let the \LevyStable{} sketch be parameterized by \Levy{} process $X$ on $\R^d$.
Given any input vector $\mathbf{x}\in (\R^d)^n$ and $\mathbf{x}'\in \R^n$ such that $f(\mathbf{x})=\sum_{v\in[n]}|\mathbf{x}'(v)|^\alpha$,
 \LevyStable{} with input $\mathbf{x}$ and Indyk's \textsf{$F_\alpha$-stable} 
 sketch with input $\mathbf{x}'$ distribute identically.
\end{theorem}

\cref{them:simulation:stable} serves as another illuminating example showing how the L\'evy-Khintchine theorem helps to understand streaming sketching. Previously, only two classes of stable moments are considered: one dimensional $F_\alpha$-moments
by Indyk \cite{indyk2006stable}, 
and multidimensional $F_{p,q}$-moments 
by Ganguly et al.~\cite{ganguly2012estimating} (see also~\cite{JayramW09}), which are sketched by \emph{algorithmic tricks} of combining stable random variables. The \LevyStable{} sketch extends such tricks to \emph{all stable processes} in a systematic way. For example, we can now estimate, 
using any estimator of Indyk's sketch~\cite{indyk2006stable,li2008estimators}, 
the $f$-moment of an $\R^3$-turnstile stream, where $f$ is the 1-stable function
\[
f(x) = \int_{\mathbb{S}_2} \frac{|\inner{\xi,x}|}{|\xi_1|^2+|\xi_2|+|\xi_3|^{1/2}} \, d\xi,
\]
and $\mathbb{S}_2$ is the unit $\R^3$-sphere. 
(It is 1-stable due to the numerator $|\inner{\xi,x}|$.  
The denominator can be any symmetric 
function of $\xi$ 
that is bounded away from 0 on the sphere.)
We are not aware of any prior work that proved such stable $f$-moments could be sketched efficiently.

\medskip

For the following two emulation theorems, 
let 
$G(z) = c +\gamma_0 z + \int_0^\infty (1-e^{-zs})\,\nu(ds)$ be the Laplace exponent of any subordinator $X$ and parameterize the L\'evy-based sketches by $X$. A cardinality sketch is \emph{Poissonized} if  
each actual insertion is simulated by $\Poisson(1)$ insertions.

\begin{theorem}[\textsf{PCSA} emulation, \cref{sec:levy_pcsa}]\label{thm:simulation:PCSA}
Let vectors $\mathbf{x},\mathbf{x}'\in \R_+^n$ be 
such that $G(\mathbf{x})=\norm{\mathbf{x}'}_0$.
Then 
 \LevyPCSA{} (parameterized by $G$) on input $\mathbf{x}$ 
 and Poissonized \PCSA{} with input $\mathbf{x}'$ distribute identically.
\end{theorem}

\begin{theorem}[\textsf{HyperLogLog} emulation, \cref{sec:levy_pcsa}]\label{thm:simulation:HyperLogLog}
Let vectors $\mathbf{x},\mathbf{x}'\in \R_+^n$ be such that $G(\mathbf{x})=\norm{\mathbf{x}'}_0$.
Then \LevyHyperLogLog{} (parameterized by $G$) with input $\mathbf{x}$ and Poissonized \textsf{HyperLogLog} with input $\mathbf{x}'$ distribute identically.
\end{theorem}

The significance of the three emulation theorems is that, 
since the final distribution of the sketches are the same as the classic sketches, every analysis, estimator, or practical optimization 
developed over the years
can be applied to \LevyStable, \LevyPCSA, and \LevyHyperLogLog{} for free,
such as Li's estimators for \textsf{$F_\alpha$-stable} sketches~\cite{li2008estimators}, 
the near-optimal $\tau$-GRA estimators~\cite{wang2023better} 
for \PCSA{} and \HyperLogLog, 
and maximum likelihood estimation for \PCSA~\cite{pettie2021information,Lang17}.
In particular, the \Fishmonger{} sketch~\cite{pettie2021information} (and entropy-compressed version of \PCSA{} with maximum likelihood estimation) can now estimate $G(\mathbf{x})$ with relative variance $1/m$ using $(1+o(1))m(H_0/I_0)
\approx 1.98m$ bits, for $m=\Omega(\log^2\log n)$.\footnote{By definition $H_0 = \frac{1}{\log 2} + \sum_{k=1}^\infty\frac{1}{k}\log_2 \left(1+1/k\right)$ and $I_0=\pi^2/6$.}

\subsection{Organization}\label{sect:organization}

We first review \Levy{} processes and the \LK{} representation theorem in \cref{sec:prelin_levy}. 
We prove the two main theorems in \cref{sec:levy_tower} and \cref{sec:levy_min} that connect the \LK{} theorem to streaming sketches. 
Specifically, in \cref{sec:levy_tower} we construct the \LevyTower{} sketch, 
which transforms any generic L\'evy process $X$ into an $f_X$-moment sketch, where $f_X$ is the characteristic exponent of $X$. 
A specialization of \LevyTower{} is presented, call \LevyStable, which is more space efficient and applies whenever $X$ is a (multidimensional) stable process.
In \cref{sec:levy_min} we construct the \LevyMinSampler, 
which can be parameterized to sample elements according to any weight 
function $G$, where $G$ is the Laplace 
exponent of a subordinator. 

In \cref{sec:levy_pcsa}, we use the $G$-transformation technique to reduce the 
problem of $G$-moment estimation to cardinality estimation; this leads to 
the \LevyPCSA, \LevyHyperLogLog, and \StableHyperLogLog{} sketches. 
In \cref{sec:previous-sketches-as-Levy-processes}, we discuss the connection between 
previous sketches and L\'evy processes in greater detail. 
In \cref{sec:tractability}, 
we discuss the problem of characterizing the set of \emph{tractable} functions,
and introduce the powerful \emph{Fourier-Hahn-\Levy} method for expanding the 
range of the \LevyTower{} beyond \LK-representable functions.
We conclude in \cref{sec:conclusion} with some conjectures on the 
relationship between moment estimation, sampling, and \Levy{} processes.

\section{Preliminaries: \Levy{} Processes and the \LK{} Theorem}\label{sec:prelin_levy}

\subsection{\Levy{} Processes on $\R^d$}

We use Sato's text~\cite{ken1999levy} as our reference for the theory of \Levy{} processes. 

\begin{definition}[\Levy{} processes {\cite[page~3]{ken1999levy}}]
    A random process $X=(X_t)_{t\geq 0}$ on $\R^d$ is a \emph{\Levy{} process} if it has the following three properties.
    \begin{description}
        \item[(A) Stationary Increments.] $X_{t+s}-X_{t}\sim X_s$ for all $ t,s\in \R_+$;\label{item:stationary}
        \item[(B) Independent Increments.] for $0\leq t_1 < t_2\ldots <t_k$, $X_{t_1},X_{t_2}-X_{t_1},\ldots, X_{t_k}-X_{t_{k-1}}$ are mutually independent; \label{item:independent}
        \item[(C) Stochastic Continuity.] $X_0=0$ almost surely and $\lim_{t\searrow 0}\pr(|X_t|>\epsilon)=0$ for any $\epsilon>0$.
\end{description}
\end{definition}

A process that satisfies (A) and (B) 
is called \emph{memoryless}, 
i.e., conditioned on $X_{t^*}$, $(X_t)_{t>t^*}$ is independent of $(X_t)_{t<t^*}$.
The primary way to study \Levy{} processes is through their \emph{characteristic functions}. Refer to \cref{tab:notation} for notation.

\begin{theorem}[\Levy{}-Khintchine representation {\cite[page~37]{ken1999levy}}]\label{thm:lk}
Any \Levy{} process $X=(X_t)_{t\geq 0}$ on $\R^d$ can be identified by a triplet $(A,\nu,\gamma)$ where $A$ is a covariance matrix, $\gamma\in\R^d$, 
and $\nu$ is a measure on $\R$ such that
\begin{align}
    \nu(\{0\})=0\quad \text{and}\quad\int_{\R^d}\min\{|s|^2,1\}\,\nu(d{s})<\infty.\label{eq:measure_cond}
\end{align}
The identification is through the characteristic function. For any $t\geq 0$ and $z\in \R^d$,
\begin{align}
    \E e^{i \inner*{X_t,z}} &= \exp\left(-t \left(\frac{1}{2}\inner*{z,Az} - i\inner*{\gamma, z} +\int_{\R^d}(1+i\inner*{z,s}\ind{|s|<1}-e^{i\inner*{z,s}})\,\nu(ds)\right)\right). \label{eq:levy-khintchine}
\end{align}
Conversely, any triplet $(A,\nu,\gamma)$ where $A$ is a covariance matrix, $\nu$ satisfies (\ref{eq:measure_cond}), and $\gamma\in\R^d$ corresponds to a  \Levy{} process satisfying (\ref{eq:levy-khintchine}). 
\end{theorem}

We call the exponent in (\ref{eq:levy-khintchine}) the \emph{characteristic exponent}, denoted by $f_X(z)=-\log \E e^{i\inner*{X_1,z}}$. 

\begin{remark}\label{rem:LK-1d-real-special-case}
    In the one-dimensional case, the matrix $A$ can be identified by the usual variance $\sigma^2$. Then we have for $z\in \R$,
    \begin{align}
    \E e^{i z X_t} &= \exp\left(-t \left(\frac{1}{2}\sigma^2 z^2 - i\gamma z +\int_{\R}(1+iz s\ind{|s|<1}-e^{izs})\,\nu(ds)\right)\right). \label{eq:levy-khintchine-1d}
\intertext{Moreover, the characteristic exponent $f_X(z)$ is real \emph{if and only if} the \Levy{} measure $\nu$ is symmetric around 0~\cite{ken1999levy}, 
in which case the characteristic exponent can be expressed as}
    \E e^{i z X_t} &= \exp\left(-t\left(Ax^2 + 2\int_0^\infty (1-\cos(sz))\nu(ds)\right)\right)\label{eq:cosine-LK-exponent}
\end{align}
\noindent for some constant $A$ and \Levy{} measure $\nu$.
\end{remark}    
We note some properties of the characteristic exponent $f$.
\begin{lemma}\label{lem:ch_exp}
Let $X$ be any \Levy{} process on $\R^d$ and $f$ be its characteristic exponent.
\begin{itemize}
    \item $\Re f \geq 0$.  $\Re f =0$ if and only if $X$ is a deterministic drift, i.e., $X_t=\gamma t$ for all $t\geq 0$.
    \item For any $x\in \R^d$, $f(-x)=\overline{f(x)}$, the complex conjugate of $f(x)$.
\end{itemize}
\end{lemma}

We list some common one-dimensional L\'evy processes in \cref{tab:common_process} that will be frequently used later.
\begin{table}[!ht]
    \centering
    \begin{tabular}{l|l|l}
       \textsf{Process}  & \textsf{Characteristic Exponent} & \textsf{Notes} \\
       \hline\hline
       Linear drift & $i \gamma z$ & $\gamma\in \R_+$\istrut[1]{3}\\\hline
        $\alpha$-stable  &  $|z|^\alpha$ & $\alpha\in(0,2]$\istrut[1]{3}\\\hline
        Poisson & $1-e^{iz}$ & \istrut[1]{3}\\\hline
        \rb{-2.5}{Compound Poisson}  & \rb{-2.5}{$\int_{-\infty}^\infty(1-e^{izs})\,\nu(ds)$} & $\nu(\R)<\infty$ is the \emph{jump rate},\istrut[0]{3}\\
        && $\nu(\R)^{-1}\nu$ is the \emph{jump distribution}\istrut[1]{0}\\\hline
        Symmetric Compound Poisson & $2\int_0^\infty(1-\cos(zs))\,\nu(ds)$ & $\nu(\R)<\infty$\istrut[2]{4}\\\hline
    \end{tabular}
    \caption{Common one-dimensional L\'evy processes.}
    \label{tab:common_process}
\end{table}

\subsection{Subordinators: \Levy{} Processes on 
$\R_+\cup \{\infty\}$}\label{sec:prelin_sub}

One-dimensional, non-negative \Levy{} processes occupy a special place in the theory of \Levy{} processes~\cite{ken1999levy}.
They are also known as 
\emph{subordinators} for reasons 
that will be explained shortly.
As all subordinators are supported on $\R_+\cup\{\infty\}$ 
they can be characterized via their \emph{Laplace transforms}.
\begin{theorem}[\LK{} for subordinators {\cite[page~138]{ken1999levy}}]\label{thm:lk_sub}
    Any subordinator $X=(X_t)_{t\geq 0}$ on $\R_+\cup \{\infty\}$ can be identified by a triplet $(c,\nu,\gamma_0)$ where 
    $c,\gamma_0\geq 0$, and $\int_0^\infty \min\{s,1\}\,\nu(ds)<\infty$, such that for any $t,z\in\R_+$,
    \begin{align}
        \E e^{-z X_t} &= \exp\left(-t\left(c\ind{z>0}  +\gamma_0 z + \int_0^\infty (1-e^{-zs})\,\nu(ds)\right)\right).\label{eq:levy-khintchine-sub}
    \end{align}
     Conversely, given any $c,\gamma_0\geq 0$, 
     and measure $\nu$ such that $\int_0^\infty \min\{s,1\}\,\nu(ds)<\infty$, there is a corresponding subordinator satisfying (\ref{eq:levy-khintchine-sub}).
\end{theorem}

    The exponent in (\ref{eq:levy-khintchine-sub}) is called the \emph{Laplace exponent}, denoted as $G_X(z)=-\log \E e^{-zX_1}$.
    Subordinators are necessarily non-decreasing. The parameter $c$ is called the \emph{kill rate}, and only comes into play when it is possible for $X_t$ to reach $\infty$, effectively killing the process.
    Sato's text \cite{ken1999levy} only considers processes with $c=0$. 
    However, we do want to consider killed processes (those with $c>0$) 
    here because, as we will see later, popular sketches like \PCSA{} and \HyperLogLog{} are induced by killed processes. 
    There is a standard way to kill a process.
    
\begin{lemma}[How to kill a \Levy{} process; see~\cite{yakubovich2021simple}]
We are given a subordinator $X$ with Laplace exponent $G:\R_+\to \R_+$. 
Fix any $c>0$, let $Y$ be an independent 
$\Exp(c)$ random variable, 
and define a new process $(X'_t)_{t\geq 0}$ on $\R_+ \cup \{\infty\}$ 
as
\begin{align*}
    X'_t =\begin{cases}
        X_t, &t<Y,\\
        \infty, &t \geq Y.
    \end{cases}
\end{align*}
Then for $z\in \R$,
\begin{align*}
    \E e^{- z X'_t} =e^{-t(G(z)+c\ind{z>0})},
\end{align*}
where $e^{-z\infty}=\ind{z=0}$.
\end{lemma}

A non-negative \Levy{} process is non-decreasing and can therefore
be used to index 
\emph{time} 
within another \Levy{} process.  
This is called \emph{subordination}.  
For example, if $(X_t)_{t\geq 0}$ is a Wiener process/Brownian motion 
and $(Z_t)_{t\geq 0}$ is a Poisson process with rate 1, 
then $(X_{Z_t})_{t\geq 0}$ ($X$ subordinated by $Z$) is piecewise constant. 
It increments the ``clock'' in the Wiener process by 1 
at intervals that are independent $\Exp(1)$ random variables.

\begin{theorem}[Subordination {\cite[page~198]{ken1999levy}}]\label{thm:subordination}
    Let $X=(X_t)$ be a \Levy{} process on $\R^d$ and $Z=(Z_t)_{t\geq 0}$ be a 
    subordinator on $\R_+$, 
    with $c=0$. Then $(X_{Z_t})_{t\geq 0}$ is a \Levy{} process on $\R^d$ such that for any $t\geq 0$ and $z\in \R^d$,
    \begin{align*}
        \E e^{i\inner*{X_{Z_t},z}} &= \exp(-t G_Z(f_X(z))),
    \end{align*}
    where $f_X$ is the characteristic exponent of $X$ and $G_Z$ is the Laplace exponent of $Z$.
\end{theorem}

\section{\LevyTower{} and \LevyStable{}}\label{sec:levy_tower}

\subsection{\LevyTower{} Sketches}

We now present a sketch that is induced by a \emph{generic} \Levy{} process on $\R^d$, which consists of two parameters.
\begin{itemize}
    \item An accuracy parameter $m\in\Z_+$ which corresponds to the number of subsketches in classic settings.
    \item A \Levy{} process $X=(X_t)_{t\geq 0}$ on $\R^d$ with characteristic exponent $f(z)=-\log \E e^{i \inner*{ X_1,z}}$.
\end{itemize}

Recall that $\T$ is the complex unit circle.

\begin{definition}[$(f,m)$-\LevyTower]
Let $f$ be the characteristic exponent of a \Levy{} process $X$ on $\R^d$ and $m\in \Z_+$.
An abstract \emph{$(f,m)$-\LevyTower} is an infinite vector $S=(S_k^{(j)})_{k\in \Z,j\in [m]} \subset \T^{\Z}$, initialized as all zero. For any element $v\in[n]$ and $y\in \R^d$,
a vector update $\mathbf{x}(v)\gets \mathbf{x}(v)+y$ is effected by:
\begin{description}    
    \item[$\Update(v,y):$] For each $k\in \Z$ and $j\in[m]$, $S_k^{(j)}\gets S_k^{(j)} + \inner*{y,X_{2^{-k}}^{(v,j)}}\mod 2\pi$, where $X^{(v,j)}=(X^{(v,j)}_t)_{t\geq 0}$ is an i.i.d.~copy of the \Levy{} process $X$ with characteristic exponent $f$.
\end{description}
\end{definition}
\begin{remark}
If $|f(\mathbf{x})|$ is in the range $[1,\poly(n)]$, then
it suffices to store levels $k\in[0,O(\log n)]$. Therefore a \LevyTower{} takes $O(m\log n)$  words of space and each word stores a number in $[0,2\pi)$.
\end{remark}

Many sketches follow the three-step design template: 
\emph{subsample}, \emph{randomly project}, and \emph{sum}.  
The \LevyTower{} can be regarded as doing the 
subsample and projection steps at the same time, 
where the measurement time $2^{-k}$ is analogous 
to a subsampling probability.

We characterize the distribution of the sketch as follows.
\begin{figure}[!ht]
    \centering
    \begin{subfigure}{0.2\linewidth}
        \includegraphics[]{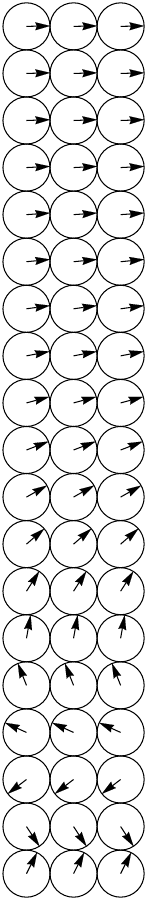}
    \end{subfigure}%
    \begin{subfigure}{0.2\linewidth}
        \includegraphics[]{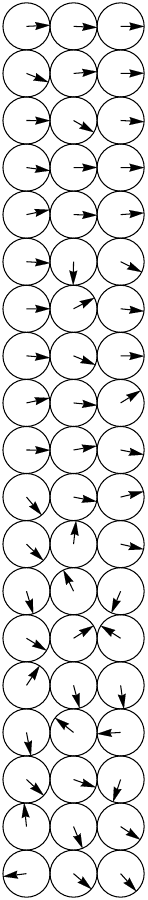}
    \end{subfigure}%
    \begin{subfigure}{0.2\linewidth}
        \includegraphics[]{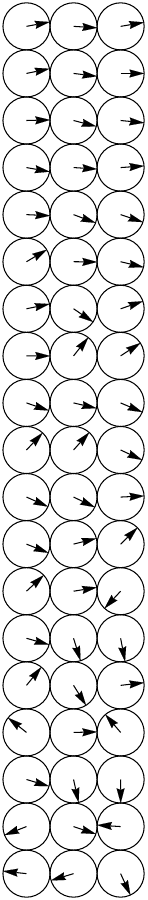}
    \end{subfigure}%
    \begin{subfigure}{0.2\linewidth}
        \includegraphics[]{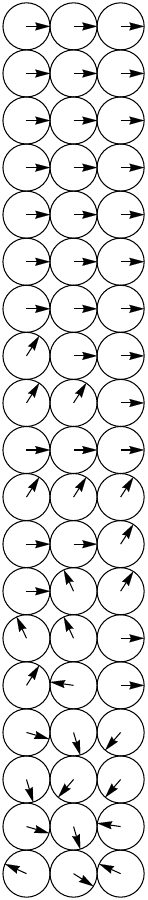}
    \end{subfigure}%
    \begin{subfigure}{0.2\linewidth}
        \includegraphics[]{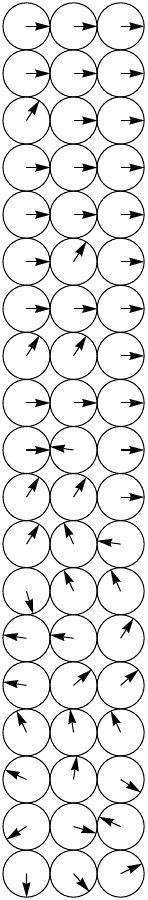}
    \end{subfigure}%
    \caption{\LevyTower{} with $m=3$ and $\mathbf{x}=(1,0,\ldots,0)$. From left to right: linear drift, the 1-stable Cauchy process, the 2-stable Wiener process/Brownian motion, the Poisson process with rate 1, and the Poisson process with rate 2. Different \Levy{} processes have different ``sensitivities'' for a target function-moment. For example, linear drift is only sensitive to the sum of the vector and insensitive to how the values are distributed. The Cauchy process is only sensitive to the $F_1$-moment, while the Wiener process/Brownian motion is only sensitive to the $F_2$-moment. Poisson processes are sensitive to the \emph{support size} of $\mathbf{x}$ and at the same time leak
    information about other $f$-moments.}
    \label{fig:levy_angular}
\end{figure}

\begin{lemma}\label{lem:Sk}
Fix a vector $\mathbf{x}=(\mathbf{x}(v))_{v\in [n]} \in (\mathbb{R}^d)^n$. For any $k\in \Z$, we have 
\begin{align*}
S_k &\sim \sum_{v\in[n]}\inner*{X_{2^{-k}}^{(v)},\mathbf{x}(v)}\\
\E e^{iS_k} &= e^{-2^{-k} f(\mathbf{x})}.
\end{align*}
\end{lemma}
\begin{proof}
    The first statement is trivially true since the sketch is linear. For the second, note that the $X^{(v)}$ are i.i.d., and we have
    \begin{align*}
        \E e^{iS_k} &= \E e^{i \sum_{v\in[n]}\inner*{X_{2^{-k}}^{(v)},\mathbf{x}(v)}} \\
        &= \prod_{v\in [n]} \E e^{i \inner*{X_{2^{-k}}^{(v)},\mathbf{x}(v)}} & \text{(by independence)}\\
        &= \prod_{v\in [n]} e^{-2^{-k}f(\mathbf{x}(v))} &\text{(by \LK)}\\
        &=e^{-2^{-k} f(\mathbf{x})}. &\text{(definition of $f(\mathbf{x})$)}
    \end{align*}
\end{proof}
Note that the $f$-moment lies exactly in the exponent of $\E e^{i S_k}$. We now formally present a method that can recover an accurate 
estimate of the $f$-moment from $S_k$, or more accurately, $m$ i.i.d.~copies $S_k^{(1)},\ldots,S_k^{(m)}$.

\subsection{Estimation of \LevyTower{}}
Recall from \cref{sec:f-moment-estimation-and-levy-processes} that
$C_t = \sum_{v\in[n]}\inner*{\mathbf{x}(v), X_t^{(v)}}$
is the linear projection measured at time $t$.
The estimator must pick a suitable $t$ and infer
an estimate of the $f$-moment from $C_t$.
In the \LevyTower{}, we store $S_k=C_{2^{-k}}$ for all 
$k\in \Z$. We first prove the concentration of the empirical mean at each level.

\begin{lemma}\label{lem:1}
For any $t\geq 0$ and $m$ i.i.d.~copies $C_t^{(1)},\ldots,C_t^{(m)}$,
\begin{align*}
    \pr\left(\left|\frac{1}{m}\sum_{j=1}^me^{iC_t^{(j)}}-e^{-tf(\mathbf{x})}\right|>\eta\right) \leq \frac{2t  |f(\mathbf{x})|}{m\eta^2}.
\end{align*}
\end{lemma}
\begin{proof}    
By Chebyshev's inequality, we have
\begin{align*}
    \pr\left(\left|\frac{1}{m}\sum_{j=1}^me^{iC_t^{(j)}}-e^{-tf(\mathbf{x})}\right|>\eta\right) \leq \frac{\var e^{iC_t}}{m\eta^2}.
\end{align*}
Furthermore, since $\E e^{i C_t}=e^{-tf(\mathbf{x})}$, 
$|e^{iC_t}|=1$, and 
$|\E e^{iC_t}|^2 = e^{-2t\cdot \Re f(\mathbf{x})}$,
we have
\begin{align*}
    \var e^{iC_t} &= \E\left|e^{iC_t}\right|^2-\left|\E e^{iC_t}\right|^2 = 1-e^{-2t\cdot \Re f(\mathbf{x})}\leq 2t \cdot \Re f(\mathbf{x})\leq 2t |f(\mathbf{x})|.
\end{align*}
\end{proof}

\begin{lemma}\label{lem:2}
With probability at least $99/100$, for any $k\geq \log_2 |f(\mathbf{x})|$, we have
\begin{align*}
    \left|\frac{1}{m}\sum_{j=1}^me^{iS_k^{(j)}}-e^{-2^{-k}f(\mathbf{x})}\right|\leq O(1/\sqrt{ m}).
\end{align*}
\end{lemma}
\begin{proof}
By the union bound and \cref{lem:1},
\begin{align*}
   &\pr\left(\exists k\geq \log_2|f(\mathbf{x})|.\; \left|\frac{1}{m}\sum_{j=1}^me^{iS_k^{(j)}}-e^{-2^{-k}f(\mathbf{x})}\right|> \eta \right)\\
   &\leq \sum_{k=\log_2|f(\mathbf{x})|}^\infty \frac{2\cdot2^{-k} |f(\mathbf{x})|}{m\eta^2}\\
   &= O\left(\frac{1}{m\eta^2}\right).
\end{align*}
Thus, it suffices to choose $\eta = O(1/\sqrt{m})$ for the probability above to be at most $1/100$.
\end{proof}

We consider the complex logarithm defined on $\C$ except for the negative real line. 
\begin{lemma}\label{lem:1-e-z}
Let $z\in\C$ with $|z|\leq 1$ and $\Re(z)\geq 0$, then $|1-e^{-z}|\in (1/2,1]\cdot |z|$. 
\end{lemma}
\begin{lemma}\label{lem:3}
Let $x,y\in \{z\in \C:|1-z|\leq 0.98\}$, then $|\log x - \log y | < 50|x-y|$.
\end{lemma}

\RestyleAlgo{ruled}
\begin{algorithm}[htbp]
  \SetAlgoLined\DontPrintSemicolon
  \textbf{Input:} A \LevyTower{} $(S_k^{(j)})_{k\in [K], j\in [m]}$ sketch of $\mathbf{x}\in (\R^d)^n$, 
  parameterized by \Levy{} process $X$ with characteristic exponent $f=f_X$.\istrut[2]{0}\\
  \textbf{Output:} An estimate $\widehat{f(\mathbf{x})}$ such that 
  if $\log_2|f(\mathbf{x})|<K$ and $m=\Omega(1)$,
  then with probability 99/100, $|\widehat{f(\mathbf{x})}-f(\mathbf{x})|=O(|f(\mathbf{x})|/\sqrt{m})$.\istrut[2]{0}\\
  \For{$W$ from $K$ down to 1}{
        $Y_W =\frac{1}{m}\sum_{j=1}^m e^{iS_W^{(j)}}$\\
        \If{$|1-Y_W|\geq 0.12$}{
            \Return{$\widehat{f(\mathbf{x})} = -2^W\log Y_W$}
        }
  }
  \Return{$-2\log Y_1$} \tcp*{\textsc{fail} outcome}
  \caption{$f$\textsf{-Moment-Estimation}\label{alg:f-moment-estimator}}
\end{algorithm}

\begin{theorem}
Consider a \LevyTower{} sketch $(S_k^{(j)})$ of a vector
$\mathbf{x}\in (\R^d)^n$,
where $k\in [K], j\in [m]$, and $K>\log_2|f(\mathbf{x})|$.
With probability at least 99/100, 
$f$\textsf{-Moment-Estimation} 
(\cref{alg:f-moment-estimator})
returns an estimate
$\widehat{f(\mathbf{x})}$ such that $|\widehat{f(\mathbf{x})}-f(\mathbf{x})|\leq O(|f(\mathbf{x})|/\sqrt{m})$. 
\end{theorem}

\begin{proof}
We analyze the behavior of the estimator conditioned on event $\mathcal{E}$ holding. \cref{lem:2} states that $\pr(\mathcal{E})\geq 99/100$.
\begin{align}\label{eqn:eventE}
\text{$\mathcal{E}$ : For all $k\in [\log_2|f(\mathbf{x})|,K]$,}\quad \left|Y_k-e^{-2^{-k}f(\mathbf{x})}\right|\leq O(1/\sqrt{m}).
\end{align}
For $k\in [\log_2|f(\mathbf{x})|,K]$, we have $|2^{-k}f(\mathbf{x})|\leq 1$ and by \cref{lem:ch_exp} we know $\Re(f(\mathbf{x}))\geq 0$. Thus we can
bound the distance from $Y_k$ to 1 both from below and above, as follows.
\begin{align}
    |1-Y_k| &\in \left|1-e^{-2^{-k}f(\mathbf{x})}\right| \pm \left|e^{-2^{-k}f(\mathbf{x})}-Y_k\right|  & \text{triangle inequality}\nonumber\\
    &\subset \left|1-e^{-2^{-k}f(\mathbf{x})}\right| \pm O(1/\sqrt{m}) & \text{\cref{lem:2}}\nonumber\\
    &\subset (1/2,1]\left|2^{-k}f(\mathbf{x})\right| \pm O(1/\sqrt{m}). & \text{\cref{lem:1-e-z}}\label{eq:1-y}
\end{align}
Let $m$ be large enough such that the 
$O(1/\sqrt{m})$ term is at most $0.1$.
When the estimate is returned, 
$W$ is the largest index such
that $|1-Y_W|\geq 0.12$.
It follows that
$W \geq \ceil{\log_2|f(\mathbf{x})|}$ 
since by \cref{eq:1-y}, 
\[
|1-Y_{\ceil{\log_2|f(\mathbf{x})|}}| 
\in (1/2,1]\left|2^{-{\ceil{\log_2|f(\mathbf{x})|}}}f(\mathbf{x})\right| \pm O(1/\sqrt{m})
> 1/4 - 0.1 = 0.15,
\]
implying the algorithm 
halted with $W\geq \ceil{\log_2|f(\mathbf{x})|}$. 
We further claim that 
$W - \ceil{\log_2|f(\mathbf{x})|} = O(1)$
and specifically that
\begin{align}
    2^{-W}|f(\mathbf{x})| &\in [0.02,0.88]. \label{eq:2-w}
\end{align}
If it were the case that $2^{-W}|f(\mathbf{x})| < 0.02$ then
\cref{eq:1-y} implies that 
\[
|1-Y_W|< 0.02+0.1 =  0.12, 
\]
which contradicts the choice of $W$. 
On the other hand, 
$2^{-W}|f(\mathbf{x})|$ cannot be bigger than 0.88 for otherwise one would 
have $2^{-(W+1)}|f(\mathbf{x})|\geq 0.44$ and \cref{eq:1-y} implies 
\[
|1-Y_{W+1}| > \frac{1}{2}\left|2^{-(W+1)}f(\mathbf{x})\right| > 0.22-0.1=0.12,
\]
also contradicting the choice of $W$. 
By \cref{lem:1-e-z} and \cref{eq:2-w} we have 
\[
\left|1-e^{-2^{-W}f(\mathbf{x})}\right|\leq \left|2^{-W}f(\mathbf{x})\right| \leq 0.88,
\]
and by \cref{eq:1-y} 
we have $|1-Y_W|\leq 0.88+0.1= 0.98$. 
Applying \cref{lem:2,lem:3}, we have 
\begin{align*}
    \left|\log Y_W - \left(-2^{-W}f(\mathbf{x})\right)\right| 
    &\leq 50\left|Y_W-e^{-2^{-W}f(\mathbf{x})}\right| & \text{\cref{lem:3}, $x=Y_W, y=e^{-2^{-W}f(\mathbf{x})}$,}\\
    &= O(1/\sqrt{m}) & \text{Event $\mathcal{E}$ (\cref{lem:2}).}
\end{align*}
Multiplying by 
$2^W = \Theta(|f(\mathbf{x})|)$ 
we conclude that
$\widehat{f(\mathbf{x})}$ achieves the desired approximation, conditioned on event $\mathcal{E}$.
\begin{align*}
    \left|\widehat{f(\mathbf{x})} - f(\mathbf{x})\right|
    =
    \left|-2^W\log Y_W - f(\mathbf{x})\right|
    \leq 
    O(2^W/\sqrt{m})
    =
    O(|f(\mathbf{x})|/\sqrt{m}).
\end{align*}
\end{proof}

\subsection{\LevyStable{} Sketches}\label{sec:alpha-levy-stable}

Recall that a $d$-dimensional L\'evy process $X$ is \emph{$\alpha$-stable} if for any 
$t\geq 0$, $X_t\sim t^{1/\alpha}X_1$. We call the characteristic exponents of stable processes \emph{stable moments}. 
Stable moments are of special interest in the context of streaming sketches since there is no need to store the whole tower; only $m$ 
i.i.d.~samples suffice to return an estimate with relative error $O(1/\sqrt{m})$.\footnote{The \LevyTower{}s induced by stable processes are \emph{self-similar}: all registers in the tower are identically distributed with a proper normalization.} 
 The only symmetric one-dimensional stable processes are $\alpha$-stable random processes, which correspond to Indyk's sketches \cite{indyk2006stable}. On the other hand, there is a rich class of higher dimensional stable processes, the one implicit in Ganguly et al.~\cite{ganguly2012estimating} for estimating
$F_{p,q}$ hybrid moments being just a single special case. 

\begin{theorem}[\LK{} for symmetric stable processes {\cite[page~86]{ken1999levy}}]
Let $X$ be a L\'evy process on $\R^d$ and $x\in\R^d$. 
$X$ is $2$-stable if and only if $\E e^{i\inner{x,X_t}}=e^{-\frac{1}{2}t \inner{x,Ax}}$ for some covariance matrix $A$. If $X$ is symmetric, then it is $\alpha$-stable for $\alpha\in(0,2)$ if and only if there is a finite, positive, symmetric measure $\mu$ on the sphere $\mathbb{S}_{d-1}=\{x\in\R^d:|x|=1\}$ such that
\begin{align*}
    \E e^{i\inner{x,X_t}} =\exp\left(-t\int_{\mathbb{S}_{d-1}}|\inner{x,\xi}|^{\alpha}\,\mu(d\xi)\right).
\end{align*}
\end{theorem}
\begin{remark}
For simplicity we only consider symmetric processes here when $\alpha\in(0,2)$ so that the characteristic exponent $f_X(x_1,\ldots,x_d)=f_X(|x_1|,\ldots,|x_d|)$ for any $x\in \R^d$. See \cite{ken1999levy} for the full characterization of (not necessarily symmetric) stable processes over $\R^d$.
\end{remark}

It turns out that \underline{any} multidimensional $\alpha$-stable moment is as simple to estimate as the one-dimensional $F_\alpha$-moment, since it is possible to maintain a random variable distributed as $f(\mathbf{x})^{1/\alpha}Y_\alpha$ where $Y_\alpha$ is a unit $\alpha$-stable random variable. We now give the formal definition of the \LevyStable{} sketch together with the proof of the emulation theorem.

\begin{definition}[$(f,m)$-\LevyStable]
Let $f$ be the characteristic exponent of an $\alpha$-stable L\'evy process $X$ on $\R^d$ and $m\in \Z_+$. An 
\emph{$(f,m)$-\LevyStable} sketch
is a vector of $m$ registers $T=(T_1,\ldots,T_m)\in \R^m$, initialized as all zero. For any element $v\in[n]$ and $y\in \R^d$, the vector update $\mathbf{x}(v)\gets \mathbf{x}(v)+y$ is effected by:
\begin{description}    
    \item[$\Update(v,y):$] For each $k\in [m]$, $T_k\gets T_k + \inner*{y,X_{k}^{(v)}-X_{k-1}^{(v)}}$, where $X^{(v)}=(X^{(v)}_t)_{t\geq 0}$ is an independent copy of the \Levy{} process $X$ with characteristic exponent $f$, indexed by $v$.
\end{description}
\end{definition}
\begin{proof}[Proof of \cref{them:simulation:stable}]
Since $X$ is L\'evy, $X_1,X_2-X_1,\ldots, X_{m}-X_{m-1}$ are i.i.d.~with the same distribution as $X_1$, 
so it suffices to look at the first register $T_1 =  \sum_{v\in[n]}\inner*{\mathbf{x}(v),X^{(v)}_1}$. Recall that for a stream with frequency vector $\mathbf{x}'$, Indyk \cite{indyk2006stable} stores $I_1=\sum_{v\in[n]}\mathbf{x}'(v)Y^{(v)}_\alpha$ where $Y_\alpha$ is a unit $\alpha$-stable random variable. By the properties of stable random variables, we have
\begin{align*}
    I_1 =\sum_{v\in[n]}\mathbf{x}'(v)Y^{(v)}_\alpha\sim \left( \sum_{v\in[n]}|\mathbf{x}'(v)|^\alpha \right)^{1/\alpha}Y_\alpha
\end{align*}
To show the simulation theorem, we need to prove
\begin{align*}
     T_1=\sum_{v\in[n]}\inner*{\mathbf{x}(v),X^{(v)}_1} \sim \left(\sum_{v\in[n]}f(\mathbf{x}(v))\right)^{1/\alpha}Y_\alpha,
\end{align*}
so that $T_1\sim I_1$ if the $f$-moment of $\mathbf{x}$ is equal to the $F_\alpha$-moment $\mathbf{x}'$. We compute the characteristic function
    \begin{align*}
\E e^{iz T_1} &=
    \E e^{i z \sum_{v\in[n]}\inner*{\mathbf{x}(v),X^{(v)}_1}} = \prod_{v\in[n]}\E e^{i \inner*{\mathbf{x}(v),z X^{(v)}_1}} & \text{(by independence)}\intertext{Since $X^{(v)}$ is $\alpha$-stable, we have $zX_1^{(v)}\sim X_{z^{\alpha}}^{(v)}$}
    &= \prod_{v\in[n]}\E e^{i \inner*{\mathbf{x}(v), X^{(v)}_{z^{\alpha}}}}\\
    &= \prod_{v\in[n]}e^{-z^{\alpha} f(\mathbf{x}(v))} & \text{(by \LK)}\\
    &= e^{-z^{\alpha} f(\mathbf{x})}=e^{-(z f(\mathbf{x})^{1/\alpha})^{\alpha}}\\
    &= \E e^{izf(\mathbf{x})^{1/\alpha} Y_\alpha}.&\text{(by \LK)}
\end{align*}
Thus $T_1$ has the same characteristic function with $f(\mathbf{x})^{1/\alpha}Y_\alpha$, which implies $T_1\sim f(\mathbf{x})^{1/\alpha}Y_\alpha$.
\end{proof}
\begin{remark}
We remark here on the differences between the \LevyTower{} sketch and the $\LevyStable{}$ sketch.
\begin{itemize}
    \item \LevyTower{} estimates a generic function moment $f_X$ where $X$ is the characteristic exponent of \emph{any} L\'evy process. \LevyTower{} stores projections at multiple times of the process and each projection can be taken modulo $2\pi$. 
    \item \LevyStable{} estimates a function moment $f_X$ where $X$ is the characteristic exponent of a \emph{stable} L\'evy process. \LevyStable{} stores projections at unit time of the process and each projection is in $\R$. \LevyStable{} generalizes the stable sketches of Indyk~\cite{indyk2006stable} and Ganguly, Bansal, and Dube~\cite{ganguly2012estimating}. 
\end{itemize}
    
\end{remark}

\section{The \LevyMinSampler}\label{sec:levy_min}

The main takeaway message from 
\cref{sec:levy_tower} is that in the $\R^d$-turnstile model, there is a uniform way to solve $f$-moment estimation whenever $f$ is the 
characteristic exponent of a \Levy{} process.
The thesis of this section is that 
non-negative one-dimensional \Levy{} processes (subordinators) play a similar role in $G$-sampling 
and $G$-moment estimation in \emph{incremental} streams.

\subsection{The Raw \LevyMinSampler}\label{sec:raw_levy_min}
Recall that a generic min sketch (\cref{def:generic_min}) is parameterized by a hash function $H:[n]\times \R_+ \times \Z_+ \to \R_+$. The sketch stores a pair $(v_*,h_*)$ initialized as $(\perp,\infty)$ and then reads the stream input at time $k=1,2,3,\ldots$
\begin{description}
    \item[$\Update(v,\Delta)$ (issued at time $k$):] if $H(v,\Delta,k)<h_*$, then set $(v_*,h_*)\gets (v,H(v,\Delta,k))$. 
\end{description}
Based on the discussion in \cref{sec:g-sampler-g-moment-estimation}, we now consider hash functions induced by subordinators.

\begin{theorem}[Raw \LevyMinSampler]\label{thm:raw_min}
For $u\in [n], \Delta\in\R_+$, and $k\in \Z_+$, define
\begin{align*}
    H(u,\Delta,k)=\inf\left\{t\geq 0 : X^{(u)}_t > Z_{k}/\Delta\right\},
\end{align*}
where  the $X^{(u)}=(X^{(u)}_t)_{t\geq 0}$ are i.i.d.~subordinators with Laplace exponent $G:\R_+\to \R_+$ and the $Z_{k}\sim \Exp(1)$ 
are i.i.d.~standard exponential random variables. Then the min sketch induced by $H$ is a $G$-sampler. I.e., the resulting state $(v_*,h_*)$ satisfies the following.
\begin{itemize}
    \item For any $u\in[n]$, $\pr(v_*=u)=G(\mathbf{x}(u))/G(\mathbf{x})$.
    \item $h_*\sim \Exp(G(\mathbf{x}))$.
\end{itemize}    
\end{theorem}

The proof of \cref{thm:raw_min} 
makes use of the following standard properties of 
exponential random variables.

\begin{lemma}[Properties of exponential random variables]\label{lem:min}
Let $Y_1,\ldots,Y_k$ be independent random variables where $Y_j\sim \Exp(\lambda_j)$ for $j\in[k]$. Then
\begin{enumerate}
\item $\min \{Y_1,\ldots,Y_k\} \sim \Exp(\lambda_1+\cdots+\lambda_k)$.

\item $\pr\left(Y_u<\min_{j\neq u}Y_j\right) = \frac{\lambda_u}{\sum_{j\in[k]}\lambda_j}$.
\end{enumerate}
\end{lemma}

\begin{proof}[Proof of \cref{thm:raw_min}]
Suppose element $u$ is updated at times $(k_j)$ with increments $(\Delta_j)$, where $\sum_j \Delta_j = \mathbf{x}(u)$.
Observe that for any $z\in \R_+$,
\begin{align*}
    \pr\left(\min_{j}\{H(u,\Delta_j,k_j)\}\geq z\right) &= \pr\left(\inf\left\{t\geq 0 : X^{(u)}_t > \min_{j}\{Z_{k_j}/\Delta_j\}\right\}\geq z\right)\\
    &= \pr\left(X^{(u)}_z \leq \min_{j}\{Z_{k_j}/\Delta_j\}\right).
    \intertext{Since the $Z_{k_j}$ are i.i.d.~$\Exp(1)$ random variables, $\min_{j}\left\{Z_{k_j}/\Delta\right\}\sim\Exp\left(\sum_j \Delta_j\right) = \Exp(\mathbf{x}(u))$.  Continuing,}    
    &= \pr(X^{(u)}_z \leq \Exp(\mathbf{x}(u)))\\
    &= \E \pr \left({X^{(u)}_z \leq \Exp(\mathbf{x}(u))} \;\middle|\; X^{(u)}_z\right)\\
    &= \E e^{- \mathbf{x}(u)X^{(u)}_z},
    \intertext{and since $X^{(u)}$ has Laplace exponent $G$, this is equal to}
    &= e^{- z G(\mathbf{x}(u))}.
\end{align*}
By the CDF of the exponential distribution,
\begin{align*}
    \min_{j}\{H(u,\Delta_j,k_j)\} &\sim \Exp(G(\mathbf{x}(u))).
\end{align*}
\end{proof}

\subsection{\StableMinSampler}
When $G(x)=x^{\alpha}$, $\alpha\in(0,1)$, we can obtain a much simpler emulation.
Recall that a \emph{standard} one-sided $\alpha$-stable random variable $W$ is one for which $\E e^{-zW}=e^{-z^\alpha}$.
\begin{theorem}[\StableMinSampler]\label{thm:stable-min}
Let $\alpha\in(0,1)$.
For $u\in [n]$ and $k\in \Z_+$, define
\begin{align*}
    H(u,\Delta,k)= \left(\frac{Z_k}{\Delta W^{(u)}}\right)^{\alpha},  
\end{align*}
where  $W^{(u)}$s are i.i.d.~standard 
one-sided $\alpha$-stable random variables and the $Z_{k}\sim \Exp(1)$ are i.i.d.~standard exponential random variables. Then the min sketch induced by $H$ is an $F_\alpha$-sampler. I.e., the resulting state $(v_*,h_*)$ satisfies the following.
\begin{itemize}
    \item For any $u\in[n]$, $\pr(v_*=u)=\mathbf{x}(u)^\alpha/\sum_{v\in [n]}\mathbf{x}(v)^\alpha$.
    \item $h_*\sim \Exp\left(\sum_{v\in [n]}\mathbf{x}(v)^\alpha\right)$.
\end{itemize}    
\end{theorem}
\begin{proof}Note that
    \begin{align*}
    \pr\left(\left(\frac{\Exp(\mathbf{x}(u))}{W^{(u)}}\right)^{\alpha}\geq z\right)
    &=\pr\left(\Exp(\mathbf{x}(u))\geq z^{1/\alpha} W^{(u)}\right)\\
    &=\E\left(\pr(\Exp(\mathbf{x}(u)) \geq z^{1/\alpha} W^{(u)}) \;\middle|\; W^{(u)}\right)\\
    &=\E e^{-z^{1/\alpha} \mathbf{x}(u) W^{(u)}}\\
    &=e^{-z \mathbf{x}(u)^\alpha}.
    \end{align*}
\end{proof}

\subsection{\LevyMinSampler}
\label{sec:levy_min_sampler}

For a generic weight function $G$, the raw \LevyMinSampler{} can be simplified since we only need $\min_{j}\{H(u,\Delta_j,k_j)\}$ to be distributed as 
$\Exp(G(\mathbf{x}(u)))$, 
which can be simulated through \emph{level functions}. 
Level functions are similar to the \emph{score} functions 
used by Cohen~\cite{cohen2018stream}, which map 
the indexed hash values and fresh randomness into one number. 
We now construct level functions from subordinators.

\begin{definition}[level function induced by $G$]\label{def:induced_level}
A \emph{level function} $\ell$ is a map from $(0,\infty)\times (0,1)$ to $[0,\infty]$ such that $\ell(x,y)$ is increasing in both $x$ and $y$. 
    Given a subordinator $X=(X_t)_{t\geq 0}$ with Laplace exponent $G:\R_+\to \R_+$, 
    the \emph{induced level function} 
    is defined to be
    \begin{align*}
        \ell_G (x,y) &= \inf\{z:\pr(X_z\geq x)\geq y\},
    \end{align*}
    for any $x,y\in \R_+$.
\end{definition}

The final \LevyMinSampler{} using level functions is presented 
in \cref{alg:G-sampler}
and \cref{thm:levy-min-sampler}.

\begin{theorem}[\LevyMinSampler]\label{thm:levy-min-sampler}
Let $H:[n]\to [0,1]$ be a uniformly random hash function. Let $G:\R_+\to \R_+$ be the Laplace exponent of any subordinator. The \LevyMinSampler{} is a pair $(v_*,h_*)$ initialized as $(\perp,\infty)$. 
A vector update $\mathbf{x}(v) \gets \mathbf{x}(v)+\Delta$ 
at time $k$ is effected by:
\begin{description}
    \item[$\Update(v,\Delta)$ (issued at time $k$):] If $\ell_G\left(\frac{Z_k}{\Delta},H(v)\right)<h_*$, 
    then set $(v_*,h_*)\gets \left(v,\ell_G\left(\frac{Z_k}{\Delta},H(v)\right)\right)$,
\end{description}
where the $Z_k\sim \Exp(1)$ are i.i.d.~standard 
exponential random variables. 
When the input vector is $\mathbf{x}\in \R_+^n$, 
the resulting sampler state satisfies the following.
\begin{itemize}
    \item For any $u\in[n]$, $\pr(v_*=u)=G(\mathbf{x}(u))/G(\mathbf{x})$.
    \item $h_*\sim \Exp(G(\mathbf{x}))$.
\end{itemize}    
\end{theorem}

\RestyleAlgo{ruled}
\begin{algorithm}[htbp]
  \SetAlgoLined\DontPrintSemicolon
  \SetKwFunction{algo}{algo}\SetKwFunction{proc}{proc}
  \SetKwFunction{activate}{Activate}
  \SetKwProg{update}{Update}{}{}
  \SetKwProg{sample}{Sample}{}{}
  \SetKwInOut{sketch}{Sketch}
  \SetKwInOut{hash}{Hash function}
  \sketch{$(v_*,h_*)$, initialized as $(\perp,\infty)$}
  \hash{$H:[n]\to \Uniform(0,1)$}
  \KwResult{Sample an element $u$ with prob.~$G(\mathbf{x}(u))/\sum_{v\in[n]}G(\mathbf{x}(v))$}
  
  \tcp{upon update $\mathbf{x}(v)\gets \mathbf{x}(v)+\Delta$}
  \update{$(v,\Delta)$}{
   $h\gets \ell_G(\Exp(\Delta), H(v))$ \tcp*{$\Exp(\Delta)$ is freshly sampled}
   \If{$h<h_*$}{
   $(v_*,h_*)\gets (v,h)$
   }
  }{}
  \tcp{upon sample}
  \sample{$(\,)$}{
  \Return{$v_*$}
  }
  \caption{\LevyMinSampler{}. $\ell_G$ is the \emph{level function} for the weight function $G$ defined in \cref{def:induced_level}.}\label{alg:G-sampler}
\end{algorithm} 
\begin{proof}
It suffices to prove that whenever 
$Z\sim \Exp(\lambda)$ and 
$Y\sim \Uniform(0,1)$, 
that
$\ell_G(Z,Y)\sim \Exp(G(\lambda))$. 
Indeed, let $X=(X_t)_{t\geq 0}$ be the subordinator with Laplace exponent $G$.  We have, for any $w>0$,
\begin{align*}
    \pr(\ell_G(Z,Y)\geq w) &= \pr(\inf\{z:\pr(X_z\geq Z)\geq Y\}\geq w) & \text{(by definition of $\ell_G$)}\\
        &=\pr(\pr(X_w\geq Z)< Y) & \text{($(X_t)$ non-decreasing)}\\
        &= \pr(X_w < Z) & \text{($Y\in \Uniform(0,1)$)}\\
        &= \E (\pr(X_w< Z) \mid X_w) & \text{(law of total probability)}\\
        &= \E (e^{-\lambda X_w} \mid X_w)\\
        &= \E e^{-\lambda X_w}
        \:=\: e^{-w G(\lambda)}.
\end{align*}
By the CDF of the exponential distribution, $\ell_G(Z,Y)\sim \Exp(G(\lambda))$.
\end{proof}

We demonstrate by examples how the \LevyMinSampler{} generalizes previous sketches and leads to new sketches. First, recall that by \LK{} (\cref{thm:lk_sub}), a function $G:\R_+\to\R_+$ is the Laplace exponent of some subordinator if and only if it is generated by a triplet $(c,\nu,\gamma_0)$ where $c,\gamma_0\geq 0$, 
and $\int_0^\infty \min(s,1)\,\nu(ds)<\infty$. 
\begin{align*}
    G(x) &= c\ind{x>0} +\gamma_0 x + \int_0^\infty (1-e^{-xs})\,\nu(ds).
\end{align*}
     
\begin{example}[$F_0$-sampler $\mapsto$ the 
$\textsf{Min}$ sketch~\cite{cohen1997size}]\label{exa:f0_sampler}
For $F_0$-sampling, we have $G(x)=\ind{x>0}$.
which is generated by $(c,\nu,\gamma_0)=(1,0,0)$.
This corresponds to a ``pure-killed'' process $X$, where, given a kill time $Y\sim\Exp(1)$, $X_t = 0$ if $t<Y$ and $X_t=\infty$ if $t\geq Y$. The induced level function (\cref{def:induced_level}) is,
\begin{align*}
    \ell(x,y)&=\inf\{z:\pr(X_z\geq x)\geq y\}
    \intertext{Since $x>0$, $\pr(X_z\geq x)=\pr(z\geq Y)=1-e^{-z}$. Regardless of $x$, this is equal to}
    &=\inf\{z:1-e^{-z}\geq y\}\\
    &=-\log (1-y).
\end{align*}
Thus the resulting sketch stores $-\log(1-\eta)$,
where $\eta$ is the minimum hash value, thereby essentially reproducing Cohen's~\cite{cohen1997size} $\textsf{Min}$ sketch. 

\end{example}\begin{example}[$F_1$-sampler $\mapsto$ min-based reservoir sampling~\cite{vitter1985random}]
For $F_1$-sampling, we have $G(x)=x$, which is generated by $(c,\nu,\gamma_0)=(0,0,1)$. This corresponds to a deterministic drift process $X$, where $X_t = t$ for $t\geq 0$. The induced level function is,
\begin{align*}
    \ell(x,y)&=\inf\{z:\pr(X_z\geq x)\geq y\}
    \intertext{and since $X_z=z$, $\pr(X_z\geq x)=\ind{z\geq x}$,}
    &=\inf\{z:\ind{z\geq x}\geq y\}\\
    &=x.  & \text{(note that $y\in(0,1)$)}
\end{align*}
Thus the resulting sketch stores the minimum random value sampled freshly at each insertion, reproducing the reservoir sampler~\cite{vitter1985random} 
with the choice of replacement simulated by taking the min.
\end{example}

Next, we demonstrate a new and 
``non-trivial'' application: 
the construction of an $F_{1/2}$-sampler.

\RestyleAlgo{ruled}
\begin{algorithm}[htbp]
  \SetAlgoLined\DontPrintSemicolon
  \SetKwFunction{algo}{algo}\SetKwFunction{proc}{proc}
  \SetKwFunction{activate}{Activate}
  \SetKwProg{update}{Update}{}{}
  \SetKwProg{sample}{Sample}{}{}
  \SetKwInOut{sketch}{Sketch}
  \SetKwInOut{hash}{Hash function}
  \sketch{$(v_*,h_*)$, initialized as $(\perp,\infty)$}
  \hash{$H:[n]\to \Uniform(0,1)$}
  \KwResult{Sample an element $u$ with prob.~$\sqrt{\mathbf{x}(u)}/\sum_{v\in[n]}\sqrt{\mathbf{x}(v)}$}
  
  \tcp{upon update $\mathbf{x}(v)\gets \mathbf{x}(v)+\Delta$}
  \update{$(v,\Delta)$}{
   $h\gets \sqrt{2 \Exp(\Delta)}\cdot \mathrm{erf}^{-1}(H(v))$ \tcp*{$\Exp(\Delta)$ is freshly sampled}
   \If{$h<h_*$}{
   $(v_*,h_*)\gets (v,h)$
   }
  }{}
  \tcp{upon sample}
  \sample{$(\,)$}{
  \Return{$v_*$}
  }
  \caption{$F_{1/2}$-sampler}
  \label{alg:1/2-stable}
\end{algorithm} 

\begin{example}[$F_{1/2}$-sampler]\label{example:1/2-stable}
For $F_{1/2}$, we have $G(x)=\sqrt{x}$, which corresponds to the $1/2$-stable process. The induced level function is
\begin{align*}
    \ell(x,y) &= \inf\{z:\pr(X_z\geq x)\geq y\}
    \intertext{and since $X$ is $1/2$-stable, we have $X_z\sim z^{2} X_{1}$}
    &= \inf\{z:\pr(z^{2} X_{1}\geq x)\geq y\}.
    \intertext{It is known that the standard 1/2-stable $X_1$ distributes identically with $1/Z^2$ where $Z$ is a standard Gaussian~\cite[page~29]{ken1999levy}. 
    Thus, we have $\pr(X_1\geq r)=\pr(|Z|\leq \sqrt{1/r})=\mathrm{erf}\left(\sqrt{\frac{1}{2r}}\right)$, 
    where $\mathrm{erf}(x)=\frac{2}{\sqrt{\pi}}\int_0^xe^{-x^2}\,dx$ is the \emph{Gauss error function}.  As $\pr(z^2X_1 \geq x) = \pr(X_1 \geq x/z^2) = \mathrm{erf}(z/\sqrt{2x})$, this is equal to}
    &= \sqrt{2 x} \cdot \mathrm{erf}^{-1}(y).
\end{align*}
Thus to sample with weight function $G(x)=\sqrt{x}$, upon an $\Update(v,\Delta)$, one just needs to compute $h=\sqrt{2 Y/\Delta} \cdot \mathrm{erf}^{-1}(H(v))$ where $Y\sim \Exp(1)$ is freshly sampled and $H(v)\sim \Uniform(0,1)$ is the hash value of $v$, and replace $(v_*,h_*)$ with $(v,h)$ if  $h<h_*$.\footnote{The inverse error function $\mathrm{erf}^{-1}$ is available, e.g., as $\mathtt{scipy.special.erfinv}$ in \texttt{Python}.}
See \cref{alg:1/2-stable}.
\end{example}

The $F_{1/2}$-moment is a special case where the $1/2$-stable distribution has a clean expression.
Penson and G\'orska \cite{penson2010exact} give explicit formulas for one-sided $k/l$-stable distributions for integers $k<l$. 
For a generic $\alpha\in(0,1)$ whose level function is difficult to compute, it is more convenient to use the \textsf{Stable-Min-Sampler} (\cref{thm:stable-min}).

\begin{example}[Common subordinators]
We now consider  samplers induced by common subordinators.
\begin{itemize}
    \item The $F_\alpha$-sampler with $G(z)=z^\alpha$ for $\alpha \in (0,1)$ corresponds to the non-negative $\alpha$-stable process, where $X_1\sim \alpha\text{-stable}$.
    \item The sampler with $G(z)=1-e^{-\lambda x}$ for $\lambda>0$  corresponds to the Poisson processes, where $X_1\sim \mathrm{Poisson}(\lambda)$.
    \item The sampler with $G(z)=\alpha \log(1+x/\beta)$ for $\alpha,\beta >0$ corresponds to Gamma processes, where $X_1\sim \mathrm{Gamma}(\alpha,\beta)$.
\end{itemize}
\end{example}
See \cref{fig:f0f1} and \cref{fig:demo} for contour plots of the level functions used in the examples.  For a generic $G$, in practice one may pre-compute the level function of $G$ on a geometrically spaced lattice and cache it as a read-only table. Such a table can be shared and read simultaneously by an unbounded number of $G$-samplers for different applications and therefore the amortized space overhead is typically small.

\begin{figure}[ht]
    \centering
    \begin{subfigure}[b]{\linewidth}
   % \hspace{1.4cm} 
   \includegraphics[width=0.9\linewidth]{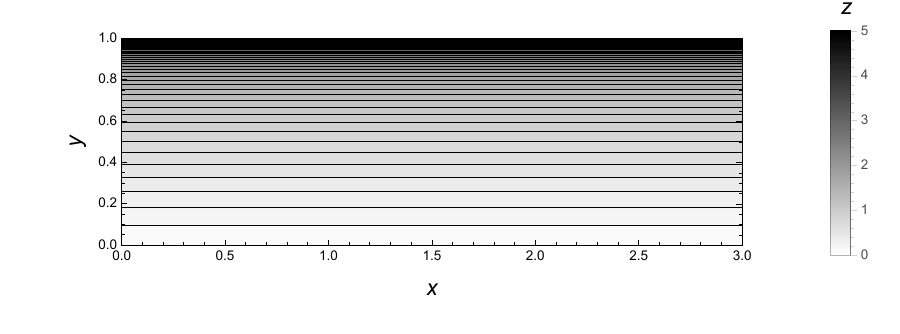}       
   \caption{Level function for $F_0$-sampler ($G$-sampler with $G(z)=\ind{z>0}$)}
    \end{subfigure}
    \begin{subfigure}[b]{\linewidth}
   % \hspace{1.4cm} 
    \includegraphics[width=0.9\linewidth]{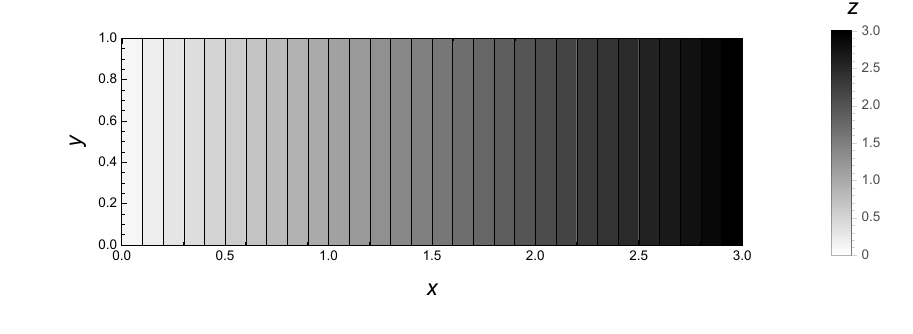}        
    \caption{Level function for $F_1$-sampler ($G$-sampler with $G(z)=z$)}
    \end{subfigure}
    \caption{Level functions for $F_0$-sampler and $F_1$-sampler}
    \label{fig:f0f1}
\end{figure}

\begin{figure}[!ht]
    \centering
    
    \begin{subfigure}[b]{\linewidth}
    \includegraphics[width=0.9\linewidth]{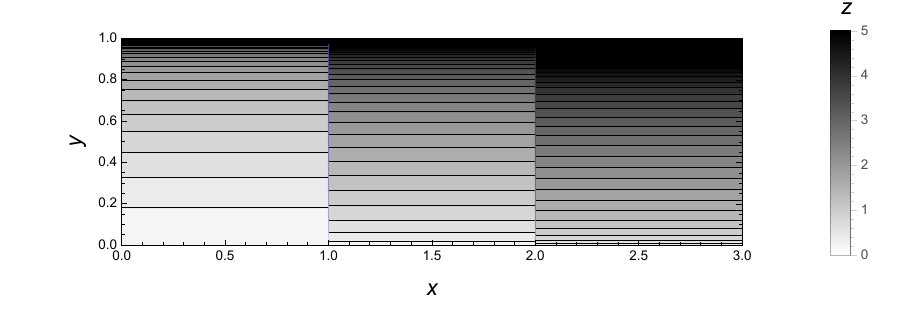}     
    \caption{Level function for $(1-e^{-z})$-sampler, induced by a Poisson counting process}
    \end{subfigure}
    
    \begin{subfigure}[b]{\linewidth}
    \includegraphics[width=0.9\linewidth]{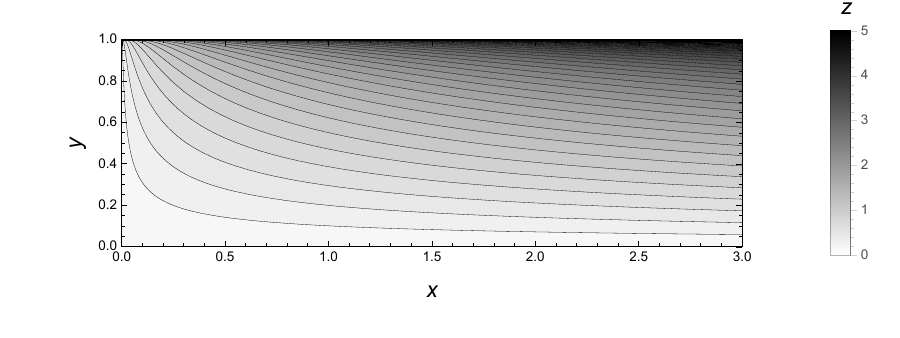}   
    \caption{Level function for $z^{1/2}$-sampler (a.k.a~$F_{1/2}$-sampler), induced by a $1/2$-stable process}
    \end{subfigure}
    
    \begin{subfigure}[b]{\linewidth}
    \includegraphics[width=0.9\linewidth]{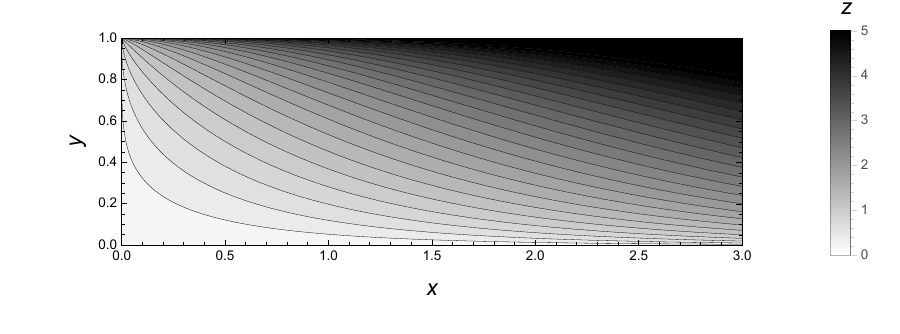}        
    \caption{Level function for $\log(1+z)$-sampler, induced by a 
    Gamma process}
    \end{subfigure}
    
    \caption{Level functions for samplers induced by well-known 
    \Levy{} processes}
    \label{fig:demo}
\end{figure}

\subsection{Universal Sampler}
Since $\ell_G(a,b)$ is increasing in both arguments,
among all updates $\{(v_i,\Delta_i)\}$ to the 
\LevyMinSampler{}, the stored sample must 
correspond to a point on the (minimum) \emph{Pareto frontier} of 
$\{(Y_i/\Delta_i,H(v_i))\}$.  Thus, it is possible to 
produce a $G$-sample for any $G\in\mathcal{G}$ 
simply by storing the Pareto frontier.
(This observation was also used by Cohen~\cite{cohen2018stream} 
in her approximate samplers.)
The size of the Pareto frontier is a random variable
that is less than $\ln n+1$ in expectation and 
$O(\log n)$ with high probability.

\begin{theorem}\label{thm:ParetoSampler}
Suppose \ParetoSampler{} processes a stream of
$\poly(n)$ updates to $\mathbf{x}$.
The maximum space used is $O(\log n)$ 
words with probability $1-1/\poly(n)$.
At any time, given a $G\in\mathcal{G}$, 
it can produce a $v_*\in[n]$ such that
$\pr(v_*=v)=G(\mathbf{x}(v))/G(\mathbf{x})$.
\end{theorem}

By instantiating $k$ independent copies of the 
\LevyMinSampler{} we can $G$-sample $k$ indices 
with replacement, but in some applications we would like to sample $k$ indices \emph{without replacement}.  In Appendix \ref{sec:wor}, we show that both \LevyMinSampler{} and \ParetoSampler{} can be modified to sample \emph{without replacement} as well.

\RestyleAlgo{ruled}
\begin{algorithm}[htbp]
  \SetAlgoLined\DontPrintSemicolon
  \SetKwFunction{algo}{algo}\SetKwFunction{proc}{proc}
  \SetKwFunction{activate}{Activate}
  \SetKwProg{update}{Update}{}{}
  \SetKwProg{sample}{Sample}{}{}
  \SetKwInOut{sketch}{Sketch}
  \SetKwInOut{hash}{Hash function}
  \sketch{$S\subset \R_+\times \R_+\times[n]$, initialized as $\emptyset$}
  \hash{$H:[n]\to \Uniform(0,1)$}
  \KwResult{Sample an element $u$ with prob.~$G(\mathbf{x}(u))/\sum_{v\in[n]}G(\mathbf{x}(v))$ upon a specified $G$}
  
  \tcp{upon update $\mathbf{x}(v)\gets \mathbf{x}(v)+\Delta$}
  \update{$(v,\Delta)$}{
   $ S \gets \Pareto(S \cup \{(\Exp(\Delta),H(v),v)\})$ \tcp*{$\Exp(\Delta)$ is freshly sampled}}
  \tcp{upon $G$-sample}
  \sample{$(G)$}{
  $(a_*,b_*,v_*) \gets \argmin_{(a,b,v)\in S} \{\ell_G(a,b)\}$\;
  \Return{$v_*$}
  }
  \caption{\ParetoSampler{}. The function $\Pareto(L)$ returns the (minimum) Pareto frontier of the tuples $L$ w.r.t.~their first two coordinates.}\label{alg:pareto-sampler}
\end{algorithm}

\section{The \LevyPCSA{} and \LevyHyperLogLog{} Sketches}\label{sec:levy_pcsa}
The $G$-transformation technique can be used to reduce \emph{$G$-moment} estimation in the $\R_+$-turnstile to the intensively studied 
\emph{cardinality estimation} problem.
\subsection{\Levy{}-Activation and $G$-Cells}

Ting~\cite{ting2014streamed} and Pettie and Wang~\cite{pettie2021information} analyze cardinality sketches
in a uniform fashion by viewing them 
as a set $\mathcal{C}$ of 
\emph{cells} which begin \emph{inactive} (0) and
can make a one-time transition to \emph{active} (1).
The probability that
an update activates a cell $c\in \mathcal{C}$ is proportional to its \emph{size} $s(c)$.  
Without loss of generality, 
one may assume the sizes are normalized to that
$\sum_{c\in \mathcal{C}} s(c)=1$.
The \emph{state} of the sketch is the
set of active cells. 
See \cite[\S 1.2]{pettie2021information} for more details. 

\begin{definition}[activation of cells]\label{def:cells}
Let $H:[n]\times \mathcal{C} \to [0,1]$ be a uniformly random hash function.
Let $s(c)$ be the size of cell $c\in \mathcal{C}$.
Initially $c\gets 0$ for all $c\in\mathcal{C}$.
Upon an update with $\Delta>0$,
\begin{description}
    \item[$\Update(v,\Delta)$ :] For each $c\in\mathcal{C}$, $c\gets c \vee \ind{H(v,c)<s(c)}$.
\end{description}
\end{definition}

Assume the total cardinality is $\lambda\in \N$. A cell with size $s$ will be activated with probability $1-(1-s)^\lambda\to 1-e^{-s\lambda}$ as $\lambda\to\infty$ and $s\lambda=O(1)$. All upper bounds for cardinality estimation in the random oracle model can be applied to the $G$-moment estimation if cells with activation probability $1-e^{-sG(\mathbf{x})}$ can be \emph{simulated}.
We call the simulated cells \emph{$G$-cells}.

\begin{definition}[\Levy-activation of $G$-cells]
Let $X$ be a subordinator and $G$ be the corresponding Laplace exponent. Upon an update with $\Delta>0$,
\begin{description}
    \item[$\Update(v,\Delta)$:] For each $c\in\mathcal{C}$, 
    $c\gets c \vee \ind{X^{(v,c)}_{s(c)} > Y/\Delta}$, where $X^{(v,c)} \sim X$ are i.i.d.~copies of $X$ and 
    $Y\sim \Exp(1)$ is a freshly sampled standard 
    exponential random variable.
\end{description}
\end{definition}

We now prove that the $G$-cells simulate the ordinary cells in \cref{def:cells} asymptotically with the cardinality $\lambda$ replaced by the $G$-moment.

\begin{lemma}
Given any input vector $\mathbf{x}\in \R_+^n$, the activation of $G$-cells are independent and for each $G$-cell $c$, 
\begin{align*}
    \pr(c = 1) &= 1- e^{-s(c) G(\mathbf{x})}.
\end{align*}
\end{lemma}
\begin{proof}
The independence is given by construction. 
\begin{align*}
    \pr(c=0) &=\pr\left(\forall v\in[n].\; X_{s(c)}^{(v,c)} 
    \leq \Exp(\mathbf{x}(v))\right)\\ 
    &= \prod_{v\in[n]} \pr\left(X_{s(c)}^{(v,c)} \leq \Exp(\mathbf{x}(v))\right) & \text{(by independence)}\\
    &=\prod_{v\in[n]} \E e^{-X_{s(c)}^{(v,c)}\mathbf{x}(v)}\\
    &= \prod_{v\in[n]}e^{-s(c)G(\mathbf{x}(v))}=e^{-s(c)G(\mathbf{x})}. & \text{(by \LK)}
\end{align*}
Hence $\pr(c=1)=1-e^{-s(c)G(\mathbf{x})}$.
\end{proof}

The \Levy{} activation lemma enables essentially all cardinality estimation techniques to be applied to $G$-moment estimation, 
by replacing cells with $G$-cells. 
We first consider the \emph{\LevyPCSA} sketch, i.e., the sketch that replaces every cell in \textsf{PCSA} with a $G$-cell with the same size.\footnote{(In brief, the \emph{basic} \textsf{PCSA} sketch has cells with sizes $1/2,1/4,1/8,\ldots$.  A \PCSA{} sketch composed of $m$ subsketches has $m$ cells of size $m^{-1}\cdot 2^{-i}$ for all $i\geq 1$.  See~\cite{flajolet1985probabilistic,pettie2021information,ting2014streamed} for more details.)}

\begin{theorem}[\textsf{PCSA} emulation]
Given any input vectors $\mathbf{x},\mathbf{x}'\in \R_+^n$ such that $G(\mathbf{x})=\norm{\mathbf{x}'}_0$,
 \LevyPCSA{} with input $\mathbf{x}$ and Poissonized \PCSA{} with input $\mathbf{x}'$ distribute identically.
\end{theorem}

The most efficient mergeable cardinality sketch is \Fishmonger~\cite{pettie2021information}, which is an entropy-compressed version of \PCSA{} with an asymptotically
optimum estimator.
It is optimal among all ``linearizable'' mergeable sketches~\cite{pettie2021information}.  (Lang's \emph{CPC} sketch~\cite{Lang17} 
included in Apache DataSketches~\cite{DataSketches} is similar, but uses an off-the-shelf compressor that does not meet the entropy bound and uses an 
estimator~\cite{Lang17} that is
worse than~\cite{wang2023better,pettie2021information}.)
Substituting cells with $G$-cells in \Fishmonger{} yields \emph{\LevyFishmonger}.

\begin{corollary}
The \LevyFishmonger{} sketch~\cite{pettie2021information}
provides an asymptotically unbiased
estimate of the $G$-moment with relative variance $1/m$.  It uses $(1+o(1))(H_0/I_0)m \approx 1.98m$ bits when $m=\Omega(\log^2\log n)$.  
Here $H_0 = \frac{1}{\log 2} + \sum_{k=1}^\infty\frac{1}{k}\log_2 \left(1+1/k\right)$, $I_0=\pi^2/6$, and $H_0/I_0\approx 1.98$.
\end{corollary}

Although (compressed) \PCSA{} has the best size-accuracy tradeoff~\cite{pettie2021information,Lang17,DataSketches}, 
\HyperLogLog{} is 
already widely deployed in industry.
For these reasons the \LevyHyperLogLog{} sketch will be
more attractive in some real world applications.
Let  $H:[n]\to \N$ be a hash function where $\pr(H(v)\geq k)=2^{-k}$ for $k\in \N$.  
One \textsf{HyperLogLog} subsketch \cite{flajolet2007hyperloglog} stores a single number $M$ such that upon insertion $v$, $M\gets \max(M,H(v))$. With $\Poisson(\lambda)$ distinct insertions, $\pr(M\geq k)=1-e^{-\lambda 2^{-k}}$. We now show how to emulate $M$ with $\lambda$ replaced by $G(\mathbf{x})$.

\begin{lemma}[\LevyHyperLogLog]
Let $G$ be the Laplace exponent and $X$ be the corresponding subordinator. Let $M$ be an $\N$-register initialized as $0$.
A vector update $\mathbf{x}(v)\gets \mathbf{x}(v)+\Delta$ is effected by:
\begin{description}
        \item[$\Update(v,\Delta)$ :] $M\gets\max\left\{M,\max\left\{k\in\N: X^{(u)}_{2^{-k}}>Y/\Delta\right\}\right\}$, where $X^{(v)}\sim X$ is hashed with key $v$ and $Y\sim \Exp(1)$ is a freshly sampled standard exponential random variable.
\end{description}
Then given any input vector $\mathbf{x}\in \R_+^n$, $\pr(M\geq k)=1-e^{-2^{-k}G(\mathbf{x})}$.
\end{lemma}

\begin{remark}
    \textsf{HyperLogLog} is restored by taking $G(x)=\ind{x>0}$, with $X$ being 
    the ``pure killed'' process.
\end{remark}
\begin{proof}
For any $u\in[n], \mathbf{x}(u)>0$, and $j\in \N$,
    \begin{align*}
        &\pr\left(\max\left\{k\in\N: X^{(u)}_{2^{-k}}>\Exp(\mathbf{x}(u))\right\} <  j\right)\\
        &= \pr\left(X^{(u)}_{2^{-j}}\leq \Exp(\mathbf{x}(u))\right)\\
        &= \E e^{-\mathbf{x}(u) X^{(u)}_{2^{-j}}}\\
        &= e^{-2^{-j}G(\mathbf{x}(u))}.    & \text{(by \LK)}
    \end{align*}
Thus
\begin{align*}
    \pr(M<j) &= \prod_{u\in[n]}\pr\left(\max\left\{k\in\N: X^{(u)}_{2^{-k}}>\Exp(\mathbf{x}(u))\right\} <  j\right)
    \:=\: e^{-2^{-j}\sum_{u\in[n]} G(\mathbf{x}(u))} = e^{-2^{-j}G(\mathbf{x})}.
\end{align*}
\end{proof}
\begin{theorem}[\textsf{HyperLogLog} emulation]
Given any input vector $\mathbf{x},\mathbf{x}'\in \R_+^n$ such that $G(\mathbf{x})=\norm{\mathbf{x}'}_0$,
 \LevyHyperLogLog{} with input $\mathbf{x}$ and Poissonized \textsf{HyperLogLog} with input $\mathbf{x}'$ distribute identically.
\end{theorem}

This allows us to easily adapt any existing implementation of 
\HyperLogLog{} to estimate $G$-moments,
using any estimator for
\HyperLogLog~\cite{flajolet2007hyperloglog,wang2023better,Ertl17b}.

\subsection{\StableHyperLogLog}\label{sec:stable-hyperloglog}

For $\alpha$-stable processes,
it is sufficient to hash every element to a single one-sided $\alpha$-stable random variable, instead of an $\alpha$-stable sample path.
\cref{lem:stable-hyperloglog} establishes
the correctness of \cref{alg:F-alpha},
whose \textbf{Query} algorithm 
includes both Flajolet et al.'s~\cite{flajolet2007hyperloglog} 
harmonic mean estimator and Wang and Pettie's~\cite{wang2023better} 
improved $\tau$-GRA-based estimator.

\RestyleAlgo{ruled}
\begin{algorithm}[htbp]
  \SetAlgoLined\DontPrintSemicolon
  \SetKwFunction{algo}{algo}\SetKwFunction{proc}{proc}
  \SetKwFunction{activate}{Activate}
  \SetKwProg{update}{Update}{}{}
  \SetKwProg{query}{Query}{}{}
  \SetKwInOut{sketch}{Sketch}
  \SetKwInOut{hash}{Hash function}
  \sketch{$M[1],\ldots,M[m]$, initialized as $-\infty$}
  \hash{  $H:[n]\times[m]\to \textrm{one sided $\alpha$-stable}$, $\alpha\in(0,1)$}
  \KwResult{Estimate $\sum_{v\in[n]}\mathbf{x}(v)^\alpha$}
  
  \tcp{upon update $\mathbf{x}(v)\gets \mathbf{x}(v)+\Delta$}
  \update{$(v,\Delta)$}{
  \For{$j\in [m]$}{
  $w\gets\floor{\alpha(\log_2 H(v,j)-\log_2\Exp(\Delta))}$\tcp*{$\Exp(\Delta)$ is freshly sampled}
  $M[j]\gets \max(M[j],w)$ 
   }
  }{}
  \tcp{upon query}
  \query{$(\,)$}{
  \Return{$C_m(\frac{1}{m}\sum_{j=1}^m 2^{-M[j]})^{-1}$} \tcp*{\HyperLogLog's estimator; $C_m \to 0.7213$}
  or \Return{$  (\frac{\log 2}{\Gamma(\tau_*)(1-2^{-\tau_*})}\frac{1}{m}\sum_{j=1}^m 2^{-\tau_* M[j]})^{-\tau_*^{-1}}$} \tcp*{$\tau_*$-\textsf{GRA} estimator; $ \tau_*\to 0.89$}
  }
  \caption{$F_\alpha$-moment estimation (\StableHyperLogLog)}
  \label{alg:F-alpha}
\end{algorithm}

\begin{lemma}[\StableHyperLogLog]\label{lem:stable-hyperloglog}
Fix $\alpha\in(0,1)$. Let $M$ be an $\N$-register initialized as $0$.  A vector update $\mathbf{x}(v)\gets \mathbf{x}(v)+\Delta$, $\Delta>0$, is effected by:
$$
M\gets\max\left\{M,\floor*{\alpha(\log_2 Z^{(v)}-\log_2 Y +\log_2 \Delta)}\right\},
$$ 
where the $Z^{(v)}$ are i.i.d.~standard 
one-sided $\alpha$-stable random variables and $Y\sim \Exp(1)$ is a freshly sampled standard exponential random variable.
For any input vector $\mathbf{x}\in \R_+^n$, $\pr(M\geq k)=1-e^{-2^{-k}\sum_{v\in[n]}\mathbf{x}(v)^\alpha}$.
\end{lemma}

\begin{remark}
    For comparison, note that at $\Update(v,\Delta)$, for any $\Delta>0$,
    \textsf{HyperLogLog} updates as
    \begin{align*}
        M\gets\max(M,\floor{-\log_2 W^{(v)}}),
    \end{align*}
    where $W^{(v)}\sim \Exp(1)$. \HyperLogLog{} can be considered as the limiting case of \StableHyperLogLog, since
    \begin{align*}
        \alpha(\log_2 Z^{(v)}-\log_2 Y +\log_2 \Delta) \to -\log_2 W^{(v)}
    \end{align*}
    in distribution as $\alpha\to 0$. It is known that $(\textrm{$\alpha$-stable})^{\alpha}\to 1/\Exp(1)$ in distribution as $\alpha\to 0$ \cite{cressie1975note}. 
\end{remark}

\begin{proof}    
For any $u\in[n], \mathbf{x}(u)>0$, and $j\in \N$,
    \begin{align*}
        &\pr\left(\floor{\alpha(\log_2 Z^{(u)}-\log_2 \Exp(\mathbf{x}(u)))}<j\right)\\
        &=\pr\left(\max\left\{k\in\N: Z^{(u)}>2^{k/\alpha}\Exp(\mathbf{x}(u))\right\} <  j\right)\\
        &= \pr\left(Z^{(u)}\leq2^{j/\alpha}\Exp(\mathbf{x}(u))\right)\\
        &= \E e^{-\mathbf{x}(u) 2^{-j/\alpha} Z^{(u)}}.
        \intertext{By \LK, $\E e^{- zZ^{(u)}}=e^{-z^\alpha}$ for $z\geq 0$ since $Z^{(u)}$ is $\alpha$-stable, so this is equal to}
        &= e^{-2^{-j}\mathbf{x}(u)^\alpha}.
    \end{align*}
    It follows that $\pr(M<k) 
    = \prod_{u\in[n]} \pr\left(\floor{\alpha(\log_2 Z^{(u)}-\log_2 \Exp(\mathbf{x}(u)))}<k\right)
    = 2^{-2^{-k}\sum_{u\in [n]} \mathbf{x}(u)^\alpha}$.
\end{proof}

\section{Previous Sketches and \Levy{} Processes}\label{sec:previous-sketches-as-Levy-processes}
In this section we will discuss in greater detail 
how previous sketches can be viewed
as being constructed from 
\Levy{} processes. 
The guiding question is: what will happen if the input vector 
is replicated a large number of 
times on disjoint supports? 

\subsection{Single-Level Aggregation and the Central Limit Theorem}\label{sec:single-level-Gaussian}
We start from the \textsf{AMS} sketch by Alon, Matias, and Szegedy \cite{alon1996space}, where a cell $Q$ stores
\begin{align*}
    Q &= \sum_{v\in[n]} \mathbf{x}(v)\xi_v.
\end{align*}
Here $\mathbf{x}\in \mathbb{R}^n$ 
is updated in an $\R$-turnstile and the 
$\xi_v \in\{-1,1\}$ are i.i.d.~Rademacher random variables.\footnote{Again, in Alon, Matias, and Szegedy's analysis~\cite{alon1996space}, 4-wise independent hashing suffices to guarantee a good enough estimate but we need to assume they are i.i.d.~here to talk about the exact distribution of the final state.} Now suppose the input vector is repeated $w$ times on disjoint supports, i.e., $\mathbf{x}_w\in \mathbb{R}^{nw}$ with
$\mathbf{x}_w(v) = \mathbf{x}(v\bmod n)$.
Then the final state of the sketch for $\mathbf{x}_w$ would be
\begin{align*}
    M_w=Q_1+Q_2+ \cdots + Q_w,
\end{align*}
where the $(Q_j)$ are i.i.d.~copies of $Q$. Note that $\E Q =0$ and 
$\var Q = |\mathbf{x}|^2 = \sum_{v\in[n]}\mathbf{x}(v)^2$. Thus as $w\to\infty$ the normalized final state $M_w/\sqrt{w}$ converges to a centered Gaussian random variable with variance $\sum_{v\in[n]}\mathbf{x}(v)^2$ in distribution, by the central limit theorem. Thus, in the regime as $w\to \infty$, the \textsf{AMS} sketch eventually becomes
the same as a cell in the \LevyTower{} $\sum_{v\in[n]}\inner*{X^{(v)}_1,\mathbf{x}(v)}$, where the $(X^{(v)})$ are i.i.d.~one-dimensional Wiener processes/Brownian motion and $X_1^{(v)}$ is the value of the process at time $t=1$. 
Of course, since a Wiener process/Brownian motion is self-similar, one does not need to sample the \Levy{} processes at exponentially spaced intervals; samples from $m$ independent processes at time $t=1$ suffice to approximate $\sum_{v\in[n]}\mathbf{x}(v)^2$ with $O(1/m)$ relative variance. 
The characteristic exponent of the Wiener process/Brownian motion is the target function for the $F_2$-moment: $f(x)=|x|^2$.

\medskip 

Another illuminating example is Ganguly's \cite{ganguly2004estimating} 
$F_k$-moment estimator for $k\geq 3$. 
Ganguly randomly projects the elements and stores 
$S=\sum_{v\in[n]}\mathbf{x}(v)Z^{(v)}$ where the $Z^{(v)}$ are i.i.d.~uniformly random roots of 
$x^k=1$ on the complex unit circle. 
The statistic $S^k$ \emph{seems} to be a good estimator for the $F_k$-moment since $\E S^k = \sum_{v\in[n]}\mathbf{x}(v)^k$. However, since this random projection has finite variance, the normalized sum $S/\sqrt{\sum_{v\in[n]}|\mathbf{x}(v)|^2}$ 
will converge to a complex Gaussian 
as the input vector is replicated on disjoint supports.
In the limit, the only information remaining in the sketch pertains to the $F_2$-moment. 
Indeed, one in fact needs the number of i.i.d.~registers to 
grow \emph{polynomially} in the support-size for it to be able to estimate $F_k$~\cite{ganguly2004estimating}.

\medskip 

By the generalized central limit theorem,  if for some sequences $(a_w)_{w\in\N}$ and $(b_w)_{w\in\N}$, $(M_w-b_w)/a_w$ converges to some non-degenerate random variable $Y$ as $w\to \infty$, then $Y$ has to be $\alpha$-stable for $\alpha\in(0,2]$. If in addition, $\var Q<\infty$, then $Y$ is Gaussian. 
Non-Gaussian stable distributions 
will be discussed shortly, 
in \cref{sec:stable_previous}.

\subsection{Multi-Level Subsampling and the Poisson Limit Theorem}\label{sec:multi_level}
Another important sketching technique is \emph{subsampling}. The idea is to devise a sketch that works for $\Theta(m)$ elements and then solve the generic case by subsampling the stream 
at rates $2^{-k}$ for $k\in\N$, one of which reduces it
down to $\Theta(m)$ elements. Without loss of generality, suppose now we have $m$ non-zero elements with distinct values $\mathbf{x}(1),\ldots,\mathbf{x}(m)$. Once again,
let $\mathbf{x}_w \in \R^{nw}$
be $\mathbf{x}$ repeated $w$ times on disjoint supports.
To obtain a $\Theta(m)$-size set  of subsamples, one needs to subsample $\mathbf{x}_w$ with rate $1/w$. Let  $Y_{w,j}$ be the indicator that the $j$th copy of $\mathbf{x}(1)$ is sampled where $\E Y_{w,j} = 1/w$. The number of elements with value $\mathbf{x}(1)$ in the subsampled set is
\begin{align*}
    \sum_{j=1}^w Y_{w,j} \to \Poisson(1) \text{ in distribution as $w\to\infty$}.
\end{align*}
Here we have invoked the Poisson limit theorem (see~\cite[Theorem 3.6.1]{durrett2019probability}) 
which can be applied since
$\E \sum_{j=1}^w Y_{w,j} =1$ and $\max_{j=1}^w \pr(Y_{w,j}\neq 0) = \frac{1}{w}\to 0$ as $w\to \infty$. Similarly, the number of elements with value $\mathbf{x}(j)$ is also $\Poisson(1)$ and the occurrences of different values are independent. It is well known that such a 
limiting distribution can be simulated algorithmically by duplicating 
every element 
a $\Poisson(1)$ number of times (see \cite{flajolet1985probabilistic,flajolet2007hyperloglog,pettie2021information,wang2023better}). 
On the other hand, such limit distributions can be equivalently simulated by hashing each element to a Poisson process and then sampling at different times. Thus, all sketches based on subsampling can be simulated by the \LevyTower{} 
with the corresponding (compound) Poisson processes.

\subsection{Stable Random Variables and Stable Processes}\label{sec:stable_previous}

Whenever $Q$ has a finite variance, $(Q_1+Q_2+\cdots + Q_w)/\sqrt{w}$ goes to Gaussian as $w\to\infty$. Nevertheless, with infinite variance, $Q$ can lie in the domain of attraction of an $\alpha$-stable distribution for any $\alpha\in (0,2)$. Indeed, the use of stable random variables is another sketching technique that directly corresponds to \Levy{} processes, namely $\alpha$-stable processes. Similar to the Gaussian case (\cref{sec:single-level-Gaussian}), $\alpha$-stable processes are self-similar so there is no need to store the whole tower with exponentially spaced sample times $2^{-k}$. Rather, it suffices to sample $m$ independent processes at time $t=1$.

Indyk \cite{indyk2006stable} uses {one-}dimensional $\alpha$-stable random variables for $\alpha\in(0,2]$ to estimate the $F_\alpha$-moment,
and indeed, the characteristic exponent of the $\alpha$-stable process is $f(x)=|x|^\alpha$. 
Ganguly, Bansal, and Dube \cite{ganguly2012estimating} consider the higher dimensional case, where each element has $d$ attributes. We now show how to reconstruct \cite{ganguly2012estimating}'s sketch from the perspective of \Levy{} processes. The target function is
\begin{align*}
    f:\R^d \to \R,\quad f(x) &= \left(\sum_{j=1}^d |x_j|^p\right)^q,
\end{align*}
where $p\in (0,2], q\in (0,1]$.
Such $f$-moments are called $F_{p,q}$ hybrid moments in \cite{ganguly2012estimating}.
By the subordination theorem (\cref{thm:subordination}; see~{\cite[page~197]{ken1999levy}}), 
this function is the characteristic exponent of a vector of $d$ independent $p$-stable processes that are subordinated by a common $q$-stable subordinator. 
That is, the \Levy{} process is precisely $Y_t \in \R^d$,
\begin{align*}
    Y_t = (X_{Z_t}^{[1]},\ldots,X_{Z_t}^{[d]}),
\end{align*}
where $Z$ is a $q$-stable subordinator and the $(X^{[j]})$ are i.i.d.~$p$-stable processes. Note that since $X^{[j]}$ is $p$-stable, we have $X_{Z_t}^{[j]}\sim Z_{t}^{1/p} X_{1}^{[j]}$. 
Thus, sampling $Y$ at time 1 yields
\begin{align*}
     Y_1 \sim Z_1^{1/p}\left(X_{1}^{[1]},\ldots,X_{1}^{[d]}\right).
\end{align*}
Take the inner product and we have, for $x\in \R^d$,
\begin{align*}
    \inner*{x,Y_1} = Z_1^{1/p}\sum_{j=1}^d x_j X_{1}^{[j]},
\end{align*}
which is exactly the random projection defined by the update algorithm of Ganguly, Bansal, and Dube~\cite{ganguly2012estimating}.

\begin{remark}
    One benefit from understanding the connection between moment estimation and \Levy{} processes is \emph{debugging} incorrect claims.  
    The conference version of Ganguly, Bansal, and Dube~\cite{GangulyBD08} stated that the $F_{p,q}$ hybrid moment could be estimated in $\polylog(n)$ space for any $p,q\in (0,2]$.  
    This claim was shown to be incorrect by Jayram and Woodruff~\cite{JayramW09}.
    Using the connection between \Levy{} processes and sketches, this claim should seem fishy, as the $F_{p,q}$-hybrid moment 
    corresponds to a \Levy{} process consisting of a vector of $p$-stable $X^{[1]},\ldots,X^{[d]}$ subordinated by a common $q$-stable subordinator.  But $q$-stable subordinators do not exist for $q\in (1,2]$!
    This error was corrected in the journal version of Ganguly, Bansal, and Dube~\cite{ganguly2012estimating}, 
    where the space is $\polylog(n)$ for $p\in (0,2], q\in (0,1]$ and polynomial for $p\in (0,2], q\in (1,2]$.  (See Jayram and Woodruff~\cite{JayramW09} for $F_{p,q}$-moment estimators for general $p,q$.)
\end{remark}

\subsection{\textsf{HyperLogLog}, \textsf{PCSA}, and Pure Killed Processes}
Cardinality sketches like \textsf{HyperLogLog}~\cite{flajolet2007hyperloglog} and \textsf{PCSA}~\cite{flajolet1985probabilistic} that only allow increments also correspond to L\'evy processes in a surprisingly natural way. We consider the number system $\R\cup \{\infty\}$, where for any $x\in \R$, $x+\infty=\infty+x=\infty$ and $\infty+\infty=\infty$. Such definitions extend to multiplication with natural numbers where $0\cdot\infty=0$ and $k\cdot \infty=\infty$ for any $k\in \Z_+$.
\begin{definition}[Pure killed processes]\label{def:pure_killed}
A \emph{pure killed process} $X=(X_t)_{t\geq 0}$ with \emph{kill rate} $c>0$ is a L\'evy process over $\R\cup \{\infty\}$ which can be simulated as follows.
\begin{itemize}
    \item Sample a \emph{kill time} $Y\sim \Exp(c)$.
    \item $X_t=0$ if $t<Y$ and $X_t=\infty$ otherwise.
\end{itemize}
In particular, we have $\pr(X_t=\infty)=\pr(Y\leq t)=1-e^{-ct}$.
\end{definition}
Assume the insertion stream is $v_1,\ldots,v_T$, where each $v_j$ is an element in the universe $[n]$. Let $X=(X_t)_{t\geq 0}$ be a pure killed process with unit kill rate $(c=1)$.  Similar to the turnstile case in \cref{sec:f-moment-estimation-and-levy-processes}, one stores the sum at time $t$,
\begin{align*}
    C_t &= \sum_{j=1}^T X_t^{(v_j)} = \sum_{v\in[n]}\mathbf{x}(v)\cdot X_t^{(v)} \in \{0,\infty\},
\end{align*}
where $X^{(v)}$s are i.i.d.~copies of $X$. It is straightforward to see that $S_k$ perfectly simulates a Poissonized \textsf{PCSA} cell \cite{flajolet1985probabilistic,pettie2021information}, with $(\{0,\infty\},+)$ mapped to $(\{0,1\},\vee)$.  See \cref{tab:pure}.
\begin{table}[ht]
    \centering
    \begin{tabular}{c|c}
    \textsf{Bit operation in cardinality sketches} & \textsf{Jumps of a pure killed process}\\
    \hline
        $0\lor0 = 0$ & $0+0=0$ \\
        $0\lor1 = 1$ & $0+\infty=\infty$ \\
        $1\lor0 = 1$ & $\infty+0=\infty$ \\
        $1\lor1 = 1$ & $\infty+\infty=\infty$ \\\hline
    \end{tabular}
    \caption{\textsf{HyperLogLog} and \textsf{PCSA} can be considered as \LevyTower{} sketches with pure killed processes where the bit-or operation $(\{0,1\},\lor)$ is simulated with the extended real numbers $(\{0,\infty\},+)$. Note that, of course, such a reinterpretation does not upgrade \textsf{HyperLogLog} and \textsf{PCSA} to work over turnstile streams. The resulting \LevyTower{} is still incremental-only since $\infty-\infty$ is not defined. }
    \label{tab:pure}
\end{table}

\noindent Since $X_t^{(v)}\in\{0,\infty\}$, it only matters whether $x(v)$ is zero or non-zero, and we have
\begin{align}
    \E e^{i C_t} &=\prod_{v\in[n]}\E e^{i \mathbf{x}(v)\cdot X_t^{(v)}} = \left(\E e^{i X_t}\right)^{\norm{\mathbf{x}}_0}\label{eq:11}.
\end{align}
By \cref{def:pure_killed}, $X_t=0$ with probability $e^{-t}$, in which case $e^{i X_t}=1$, and $X_t=\infty$ with probability $1-e^{-t}$, in which case $e^{iX_t}=0$.\footnote{Here we take the Ces\`aro limit $e^{i\infty}=\lim_{T\to\infty}\frac{1}{T}\int_0^T e^{it}\,dt=0$.} Thus we have $\E e^{i X_t} = e^{-t}$. Inserting this back to \cref{eq:11}, we have
\begin{align*}
   \E e^{i C_t} &= e^{-t\norm{\mathbf{x}}_0}.
\end{align*}
This coincides with the observation in 
\cref{sec:f-moment-estimation-and-levy-processes}, where the estimation target $\norm{\mathbf{x}}_0$ lies right in the exponent of 
$\E e^{iC_t}$. 

We note that \textsf{HyperLogLog} and \textsf{PCSA} both correspond to the \emph{same} pure killed process, but 
with different couplings of cells on different levels. In particular, one can view \textsf{HyperLogLog} as storing cells $(C_{1/2},C_{1/4},C_{1/8},\ldots)$ where the $C_{2^{-j}}$s are sampled from 
\emph{a single process} 
(exactly as in the \LevyTower{}) 
while \textsf{PCSA} stores $(C_{1/2},C_{1/4},C_{1/8},\ldots)$
with the $C_{2^{-j}}$s being sampled 
from independent processes.
Since a pure killed process will remain at $\infty$ after its first jump, the cells $(C_{1/2},C_{1/4},C_{1/8},\ldots)$ 
in \textsf{HyperLogLog} always consist 
of a prefix of $\infty$s  followed by all zeros. 
Therefore, \textsf{HyperLogLog} only needs to store the \emph{length} of the prefix of $\infty$s 
as the counter $M$.

\subsection{Uniform Random Projection and the Integral \LevyTower}\label{sect:xmodp}

Indyk's $F_p$-stable sketches \cite{indyk2006stable} are able to estimate the $F_p$-moment for $p\in(0,2]$ but do not handle the $F_0$-moment. 
Cormode, Datar, Indyk, and Muthukrishnan \cite{cormode2003comparing} approximate the $F_0$-moment by the $F_\alpha$-moment for very small $\alpha>0$. 
This scheme needs a bound on $\|\mathbf{x}\|_\infty$, 
as a single unbounded update can raise the estimate to infinity.
Kane, Nelson, and Woodruff \cite{kane2010optimal} take another approach: subsample the stream to a suitable level, randomly project each element within the field $\Z_p$ for 
a random prime $p >  \epsilon^{-1}\log(\|\mathbf{x}\|_{\infty})$, 
and then use an (insertion-only) 
cardinality estimator.
The random projection approach~\cite{kane2010optimal} 
also requires a bound on $\|\mathbf{x}\|_\infty$.
As an $F_0$-moment sketch it has a certain failure probability (e.g., if $p | \mathbf{x}(v)$ for many $v$)
but it can be viewed as estimating a related quantity: the $f_{F_0,p}$-moment for $f_{F_0,p}(x)=\ind{p\nmid x}$.

We show how the \LevyTower{} reconstructs Kane, Nelson, and Woodruff's \cite{kane2010optimal} 
trick mechanically, 
by \emph{pure computation}. The target function $f(x)=\ind{p\nmid x}$ for $x\in \Z$ can be written as the following \LK-representable function; see \cref{rem:LK-1d-real-special-case}.
\begin{align*}
    f(x) &= \frac{1}{p}\sum_{j=0}^{p-1}(1-\cos(2\pi jx/p)).
\end{align*} 
When $p \mid x$ each of the summands is 
$0$ whereas when $p \nmid x$ the cosine terms sum to $0$.
The \Levy{} process with characteristic exponent $f$ is the compound Poisson process with jumps uniformly chosen from $-2\pi (p-1)/p, \ldots, 0,\ldots,2\pi (p-1)/p $. 
See~\cref{fig:L0}.
The corresponding \LevyTower{} reconstructs the Kane-Nelson-Woodruff sketch~\cite{kane2010optimal}: subsample through Poisson processes and uniformly project over $\Z_p$ with the uniform random jumps. It is straightforward to see that the resulting \LevyTower{} is \emph{integral} in the sense that $S_k\in \{2j\pi/p:j=0,\ldots,p-1\}$ so 
$S_k$ can be identified by a $\Z_p$-value.

\begin{figure}[!ht]
    \centering
    \includegraphics[height=4cm]{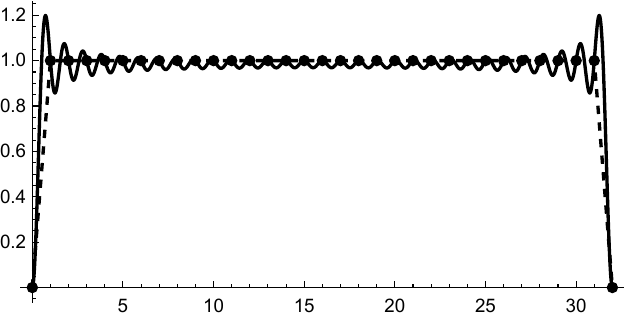}\hfill%
    \includegraphics[height=4cm]{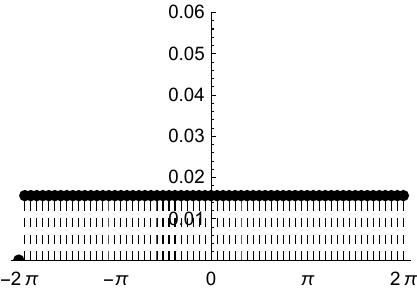}
    \caption{Left: $f_{F_0,32}(x)=\sum_{j=1}^{31}\frac{1}{32} (1-\cos(2\pi j x/32)).$ The black dots mark the values at integer $x$s. Right: The jump distribution of the compound Poisson process $X$ with characteristic exponent $f_{F_0,32}$. The subsampling and uniform random projection tricks used in \cite{kane2010optimal} are recovered from computing the corresponding \Levy{} process.}
    \label{fig:L0}
\end{figure}

\section{Tractability and \Levy{} Processes}\label{sec:tractability}

The most basic complexity-theoretic question in any model of computation is to separate the \emph{tractable} from the \emph{intractable}, e.g.,
to characterize the class 
$\mathbf{P}$ in Turing-machine complexity.
Braverman and Ostrovsky~\cite{braverman2010zero} defined an analogue of $\mathbf{P}$ for $f$-moment estimation.  

\begin{definition}[Tractability]\label{def:tractable}
    A function $f:\Z\to\R$ is \emph{tractable}
    if it is possible to $(1\pm\epsilon)$-approximate
    the $f$-moment of a vector 
    $\mathbf{x}\in\{-n,\ldots,0,\ldots,n\}^n$
    updated in a turnstile stream, with a $\poly(\epsilon^{-1},\log n)$-bit sketch.
\end{definition}

Braverman, Chestnut, Woodruff, and Yang~\cite{braverman2016streaming} relaxed this definition to include $n^{o(1)}\poly(\epsilon^{-1})$-bit sketches, which slightly broadens the class of tractable functions, but does not seem to affect its character very much.  
Braverman and Ostrovsky~\cite{braverman2010zero} 
characterized the set of tractable $f$ that are monotonically non-decreasing on $[0,\infty)$,
and Braverman, Chestnut, Woodruff, and Yang~\cite{braverman2016streaming} 
gave a nearly complete characterization of all 
tractable functions---monotone or not---excluding what they termed ``nearly periodic functions.''
As our results apply to $f$-moments where $f: \R\to\C$ and $f: \R^d\to \C$, we expand the definition of tractability to explicitly include the domain.

\begin{definition}[$M$-Tractability]\label{def:M-tractable}
    Let $M$ be a commutative monoid.
    A function $f: M\to\C$ is \emph{tractable}
    if for any parameter $\epsilon>0$, it is possible to $(1\pm\epsilon)$-approximate
    the $f$-moment of a vector $\mathbf{x}\in M^n$
    updated in an $M$-turnstile stream (\cref{def:M-turnstile}), with a $\poly(\epsilon^{-1},\log n)$-bit sketch.
    (When $M$ is discrete, say $\Z$, we must assume some bound $\|\mathbf{x}\|_\infty < n$.
    When $M$ is continuous let us additionally suppose that there are $\poly(n)$ updates, and each update is presented to $O(\log n)$ bits of precision.)
    Define $\mathcal{T}[M]$ to be the set of all $M$-tractable functions.
\end{definition}

\medskip 

The present work gives us a completely new way to approach the tractability question.  
We have shown that the \LevyTower{} estimates 
every function $f$ that is the characteristic exponent of a \Levy{} process, and such exponents are characterized by the \LK{} representation theorem.  
In this section we investigate the class of 
tractable functions captured by the Braverman et al.~\cite{braverman2016streaming} 
sketch and the \LK-representable functions captured by the \LevyTower{} sketch.
The \LK-class includes many multidimensional functions $f : \R^d\to \C$ outside the scope of Braverman et al.~\cite{braverman2016streaming},
so we shall mainly investigate each classes's coverage of $\mathcal{T}[\Z]$ and 
implications for $\mathcal{T}[\R]$.
In the remainder of this section we present the following observations and results.

\begin{description}
    \item[I.] We show that some easy tests for \emph{in}tractability have a basis in \emph{non}-\LK-representable functions. 

    \item[II.] We prove that a function can be $\Z$-tractable but not $\R$-tractable, that is, $\mathcal{T}[\Z]$ is a strict superset of $\mathcal{T}[\R]$.
    
    \item[III.] We demonstrate that the class of tractable functions implied by the \LevyTower/\LK{} representation theorem includes nearly periodic functions not captured by the $F_2$-heavy hitter framework of \cite{braverman2010zero,braverman2016streaming}.

    \item[IV.] We observe that there exist simple, tractable functions that do not correspond to the characteristic exponent of any \Levy{} process.  Nonetheless, we can prove that such functions are tractable in the \LevyTower/\LK-framework via what we call the \emph{Fourier-Hahn-\Levy} method.  We conjecture that this method is powerful enough to capture \emph{all} $M$-tractable moments, for $M\in \{\Z,\R,\Z^d,\R^d\}$.
\end{description}

\subsection{I: Intractability and Non-\LK-representability}

We give a comparison between one-dimensional real characteristic exponents and existing sketching lower bounds based on 
communication complexity lower bounds. 
Such exponents can be written as $f(x) = Ax^2 +\int_0^\infty(1-\cos(xs))\,\nu(ds)$, since $f$ being real implies that $\nu$ is symmetric; see~\cref{rem:LK-1d-real-special-case}.

\begin{lemma}\label{lem:levy_char}
Let $f(x)=Ax^2 +\int_0^\infty(1-\cos(xs))\,\nu(ds)$. The following statements are true.
\begin{itemize}
\item $f(x)\geq 0$ for any $x\in \R$. [Commentary: functions with both positive and negative values require $\Omega(\poly(n))$-size sketches for constant factor approximations \cite{chestnut2015stream}.]
\item For any $z\in\Z$ and $x\in\R$, $f(zx)\leq z^2 f(x)$. [Commentary: functions increasing faster than quadratic require $\Omega(\poly(n))$-size sketches for constant factor approxiamtions~\cite{alon1996space,bar2004information}.]
\item For any $z\in \R_+$, $f(z)=0$ if and only if $f(x+z)=f(x),\forall x\in \R$, i.e., $f$ is periodic with period $z$. [Commentary: functions that have zeros other than the origin require $\Omega(\poly(n))$-size sketches, unless the function 
is periodic with period $\min\{z>0:f(z)=0\}$ \cite{chestnut2015stream}.]
\end{itemize}
\end{lemma}
\begin{proof}
The first statement is obvious. For the second statement, it suffices to prove that
with $y=xs$, 
\[
(1-\cos(zy))\leq z^2(1-\cos(y)),
\]
for any $y\in \R, z\in \Z$.
When $z=0$, the statement trivially holds. Without loss of generality, assume $z\in\Z_+$ and let $g(y)=z^2 (1-\cos(y))-(1-\cos(zy))$.  Check that $g'(y)=z^2\sin(y)-z\sin(zy)$ and $g''(y)=z^2\cos(y)-z^2\cos(zy)$. Clearly for $y\in[0,\pi/z]$, we have $\cos(zy)\leq \cos(y)$ and thus $g''(y)\geq 0$. Note $g'(0)=0$ and therefore $g(y)\geq 0$ for $y\in[0,\pi/z]$. Now consider $y\in[\pi/z,\pi]$, $z^2(1-\cos(y))\geq z^2(1-\cos(\pi/z)) \geq 2 \geq 1-\cos(zy)$, which implies $g(y)\geq 0$ on $[0,\pi]$. Thus we have $g(y)\geq 0$ for $[\pi,2\pi]$ too by symmetry. The statement is thus proved since $g(y)$ is periodic in $2\pi$.

For the third statement, 
$f(z) = 0$ necessarily implies that $\sigma=0$ and the measure $\nu$ concentrates on $2\pi j/z$ for $j\in \Z$. Thus $f(x)=\sum_{j\in \Z}(1-\cos(2\pi j x/z))\nu(\{2\pi j /z\})$ is periodic with period $z$.
\end{proof}

\subsection{II: Tractability over Integers vs. Reals}\label{sect:tractability-integers-vs-reals}

It is known that $F_p$ is intractable for $p>2$~\cite{ChakrabartiKS03,Ganguly12,LiW13,WoodruffZ12}.
Nonetheless, there are tractable $f$ that grow slightly 
super-quadratically, such as $f(x)=x^2\log^c(1+|x|)$, for $c=O(1)$; see~\cite{braverman2010zero,braverman2016streaming}.
To be more specific, such $f$ are tractable over the integers, 
\emph{but not the reals}, as we now prove.

\begin{theorem}\label{thm:TZ-TR}
    $\mathcal{T}[\Z]\neq \mathcal{T}[\R]$.  Specifically, $f(x) = x^2\log(1+|x|) \in \mathcal{T}[\Z]-\mathcal{T}[\R]$.
\end{theorem}

\begin{proof}
    When $\mathbf{x}\in \{-n,\ldots,n\}^n$ 
    it is known that 
    $(1\pm\epsilon)$-approximating $F_3(\mathbf{x}) = \sum_j |\mathbf{x}(j)|^3$ 
    requires $\tilde{\Omega}_\epsilon(n^{1/3})$ 
    space~\cite{ChakrabartiKS03,Ganguly12,LiW13,WoodruffZ12}.
    Map $\mathbf{x}$ to $\mathbf{y} = \mathbf{x}/n^2$.
    Then 
    \begin{align*}
f(\mathbf{y})   = \sum_j f(\mathbf{y}(j)) 
= \sum_j \mathbf{y}(j)^2\log(1+|\mathbf{y}(j)|)
&= \sum_j \frac{\mathbf{x}(j)^2}{n^4}\log\left(1+\frac{|\mathbf{x}(j)|}{n^2}\right)\\
                &= \sum_j \left(\frac{|\mathbf{x}(j)|^3}{n^6} - \frac{1}{2}\frac{|\mathbf{x}(j)|^4}{n^8}+O\left(\frac{|\mathbf{x}(j)|^5}{n^{10}}\right)\right)\\
                &\in [(1-1/n)n^{-6}F_3(\mathbf{x}), n^{-6}F_3(\mathbf{x})].
    \end{align*}
    Thus any $(1\pm \epsilon)$-approximation of $f(\mathbf{y})$ yields a $(1\pm(\epsilon+1/n))$ approximation of $F_3(\mathbf{x})$.
\end{proof}

\subsection{III: Nearly Periodic Functions}\label{sect:nearlyperiodic}

Braverman, Chestnut, Woodruff, and Yang~\cite[\S 5]{braverman2016streaming}
gave the following example of a 
nearly periodic function that cannot be 
solved in their $F_2$-heavy hitter-based framework, but can be estimated using \emph{ad hoc} algorithmic tricks, 
in $O(\epsilon^{-8}\log^{15} n)$ space. 

\begin{definition}[$g_{np}$, a nearly periodic function {\cite[\S 5]{braverman2016streaming}}]
For $x\in \N$,     $g_{np}(x)=2^{-\tau(x)}$, where $\tau(x)=\max\{j\in\N:2^j|x\}$, i.e., the position of the first ``1'' in the binary representation of $x$.
\end{definition}

The reason that $g_{np}$ cannot be tracked by finding $F_2$-heavy hitters is that $g_{np}(x)$ can occasionally become polynomially small in $x$. In particular, for $k\in \N$, $g_{np}(2^k)=2^{-k}$.
We now demonstrate how the function $g_{np}$ can be estimated \emph{directly} using the \LevyTower.

\medskip 

To prove that $g_{np}\in\mathcal{T}[\Z]$ 
we may assume that $\mathbf{x}\in\{-n,\ldots,0,\ldots,n\}^n$, 
so it suffices to consider a $2^w$-periodic version of $g_{np}$ with $w>\log n$.
\begin{definition}[$2^w$-periodic $g_{np}$]
For $x\in \N$,
\begin{align*}
    g_{np,w}(x) &= g_{np}(x\operatorname{mod} 2^w)
\end{align*}    
\end{definition} 
We show that, indeed, $g_{np}$ corresponds to a \Levy{} process,
specifically a compound Poisson process.

\begin{figure}[!ht]
    \centering
    \includegraphics[height=4cm]{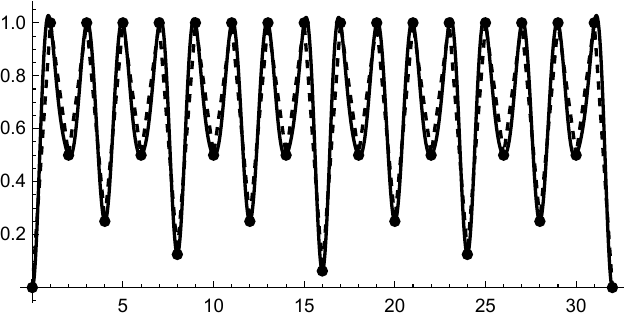}\hfill%
    \includegraphics[height=4cm]{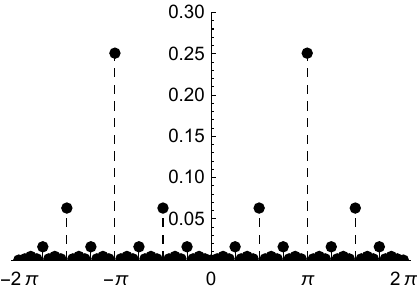}
    \caption{Left: $g_{np,5}(x)=\sum_{j=1}^{31}\frac{2^{2\tau(j)+1}+1}{1536} (1-\cos(2\pi j x/32)).$ The black dots mark the values at integer $x$s. Right: The jump distribution of the compound Poisson process $X$ with characteristic exponent $g_{np,5}$. Such nearly periodic functions do not fit in the $F_2$-heavy-hitter based framework in \cite{braverman2016streaming}. Nevertheless, one may compute the corresponding \Levy{} process and apply the \LevyTower.}
    \label{fig:braver}
\end{figure}
\begin{lemma}\label{lem:gnp}
    Fix $w\in \N$. For any $x\in \N$,
    \begin{align*}
        g_{np,w}(x) &= \sum_{j=1}^{2^w-1}\frac{2^{2\tau(j)+1}+1}{3\cdot 2^{2w-1}} (1-\cos(2\pi j x/2^w)).
    \end{align*}
\end{lemma}

\begin{proof}
    This expression was derived by taking the Fourier transform of $g_{np,w}$
    over $\Z/2^w\Z$ but we shall prove the lemma without going through Fourier transforms.
    When $x=0$, $g_{np,w}(x)=0$ as all the cosines are 1.
    When $x\in(0,2^w)$ write it as $x=y2^r$ where $y$ is odd, 
    so our goal is to show that
    $g_{np,w}(x)=2^{-r}$.
\begin{align*}
    g_{np,w}(x) &= \sum_{j=1}^{2^w-1}\frac{2^{2\tau(j)+1}+1}{3\cdot 2^{2w-1}} (1-\cos(2\pi j x/2^w))\\
                &= \frac{1}{3\cdot 2^{2w-1}}\sum_{s=0}^{w-1}(2^{2s+1}+1)\sum_{\substack{j=z2^s\\z \text{ odd}}} (1-\cos(2\pi yz/2^{w-r-s}))\\
\intertext{Now when $s\geq w-r$ the second sum vanishes as all cosines are 1.  Moreover, $z\in\{1,3,\ldots,2^{w-s}-1\}$ takes on $2^{w-s-1}$ values, 
and since $y$ is coprime to every power of 2, the second sum runs through 
the $2^{w-r-s-1}$ odd residues of $2^{w-r-s}$, $2^r$ times each. Continuing, we have}
    &= \frac{1}{3\cdot 2^{2w-1}}\sum_{s=0}^{w-r-1}(2^{2s+1}+1)\cdot 2^r\sum_{z\in \{1,3,\ldots,2^{w-r-s}-1\}} (1-\cos(2\pi z/2^{w-r-s}))
\intertext{When $w-r-s>1$ the cosines of the odd residues sum to 0, whereas when $w-r-s=1$ there is only one odd residue and $1-\cos(2\pi/2)=2$.}
    &= \frac{1}{3\cdot 2^{2w-1}}\sum_{s=0}^{w-r-1}(2^{2s+1}+1)\cdot 2^r \cdot 
    \left\{\begin{array}{ll}
    2^{w-r-s-1} & \text{ when $w-r-s>1$,}\\
    2           & \text{ when $w-r-s=1$.}
    \end{array}\right.\\
    &=\frac{1}{3\cdot 2^{2w-1}}\left(
    \sum_{s=0}^{w-r-2} (2^{2s+1}+1)2^{w-s-1} + (2^{2(w-r-1)+1}+1)2^{r+1}\right)\\
    &= \frac{1}{3\cdot 2^{2w-1}}\left(
    2^w\sum_{s=0}^{w-r-2} (2^{s}+2^{-s-1}) + (2^{2w-r}+2^{r+1})\right)\\
    &= \frac{1}{3\cdot 2^{2w-1}}\left(
    2^w(2^{w-r-1}-2^{-(w-r-1)}) + (2^{2w-r}+2^{r+1})\right)\\
    &= \frac{3\cdot 2^{2w-r-1}}{3\cdot 2^{2w-1}} 
    = 2^{-r}.
\end{align*}
\end{proof}

By \LK{}, $g_{np,w}$ is the characteristic function of a compound Poisson process 
$(X_t)$ with jump rate
$$\sum_{j=1}^{2^{w}-1}\frac{2^{2\tau(j)+1}+1}{3\cdot 2^{2w-1}} =\frac{2^{2w}-1}{3\cdot 2^{2w-1}},$$
where the jump distribution is symmetric, with jump $J$ distributed as
\begin{align*}
    \pr(J=2\pi j/2^w)=\pr(J=-2\pi j/2^w)&= \frac{1}{2} \frac{2^{2\tau(j)+1}+1}{2^{2w}-1}, \quad \text{for $j=1,\ldots 2^{w}-1.$}
\end{align*}

 Moreover, the parameter $w$ for the random jump induced by $g_{np,w}$ has a neat algorithmic interpretation as the max recursion depth if one implements each random jump by a recursive algorithm\footnote{(Halt with prob.~1/2; go left and repeat with prob.~1/4; and go right and repeat with prob.~1/4; see \cref{fig:braver}).}. 
 The case where $w\to\infty$ corresponds to the algorithm without depth limit. 
 Let $\mathbb{B} \subset [0,1)$ be the set of all real fractions with finite binary representations. Note that the jump distribution converges pointwisely to
\begin{align*}
    \pr(J=2\pi x)=\pr(J=-2\pi x)&=\frac{1}{2} 2^{1-2\tau_*(x)} = 2^{-\tau_*(x)},
\end{align*}
for any $x\in \mathbb{B}$, where $\tau_*(x)$
is the length of the representation of $x$, e.g.,
$\tau_*(3/8) = \tau_*((0.011)_2)=3$.  
This jump distribution can be simulated in unbounded input streams as follows.
Let $A=\sum_{j=1}^\infty A[j]2^{-j}$
be uniform in $[0,1]$,
$T \sim \Geometric(1/2)$, 
and $\xi\in\{-1,1\}$ be Rademacher.
Then the jump $J$ is distributed as:
\begin{align*}
    J \sim 2\pi\, \xi \left(\sum_{j=1}^{T-1}A[j] 2^{-j} + 2^{-T}\right).
\end{align*}

\subsection{IV: The Fourier-Hahn-\Levy{} Method} \label{sec:fourier-hahn-levy}

The example of $g_{np}$ shows that it is possible to \emph{systematically} sketch nearly periodic functions by sketching periodic functions and letting the period go to infinity. (See Wang~\cite{Wang25} for other examples.)
Given the success of the \LevyTower{} in estimating both
``standard'' $f$-moments and exotic ones like the $g_{np}$-moment, it is natural to conjecture that the \LK{} representation theorem actually defines
the set of tractable functions $f$.
Unfortunately, reality is not quite this clean, and indeed, this simplistic conjecture is \emph{absolutely false}.  
We begin by giving a simple function $f$ that is obviously tractable, and obviously not \LK-representable.  
Nonetheless, the conjecture can be salvaged by encoding $f$ as the difference between 
two \LK-representable functions, 
via what we call the 
\emph{Fourier-Hahn-\Levy} method.

\begin{example}[The 0-1-5 Problem]
Consider the function $f : \Z\to\R_+$, where 
\[
f(x) = \left\{\begin{array}{ll}
0 & \text{ if $x=0$,}\\
1 & \text{ if $|x|=1$,}\\
5 & \text{ if $|x|>1$.}
\end{array}\right.
\]
The $f$-moment is easy to estimate to within a $1\pm \epsilon$ factor in $O(\epsilon^{-2}\polylog(n))$ space, 
using $F_0$-sampling~\cite{CormodeF14}. 
However, no \Levy{} process has its characteristic function equal to $f$, for otherwise we would have $f(2\cdot 1)\leq 2^2 f(1)$ by \cref{lem:levy_char}, 
which implies $5\leq 4$, a contradiction.
\end{example}

\subsubsection{The Periodic Fourier-Hahn-\Levy{} Method}

Let us now extend our method to handle 
the \emph{0-1-5 problem} and others.
Suppose that $\|\mathbf{x}\|_\infty < n$.
It suffices to consider a periodic version
$f_p$ of $f$ with period $p\geq 2n$ 
that agrees with $f$ on the range $\{-n,\ldots,n\}$.  For example, 
in the case of the 0-1-5 problem we would 
define $f_p$ as
\[
f_p(x) = \left\{\begin{array}{ll}
0 & \text{ if $|x|_p = 0$,}\\
1 & \text{ if $|x|_p = 1$,}\\
5 & \text{ if $|x|_p > 1$,}
\end{array}\right.
\]
where $|x|_p = \min\{x \operatorname{mod} p, p - x \operatorname{mod} p\}$.
The three-step \emph{Fourier-Hahn-\Levy{} method} 
is executed as follows.
\begin{description}
    \item[Fourier Transform.] Decompose $f(x)-f(0)=\sum_{j=1}^{p-1}(1-\cos(2\pi  jx /p)) \hat{f}(j)$ where the coefficient $\hat{f}(j)$ can be computed by taking the Fourier transform of $f$.\footnote{We have $f(x) - f(0) = \sum_{j=1}^{p-1}\tilde{f}(j)(e^{2\pi i jx/p}-1)$.
    Since $f$ is real and symmetric $\tilde{f}$ is real and $\tilde{f}(j)=\tilde{f}(p-j)$, 
    so we can pair up the summands $j$ and $p-j$ as follows:
    $\tilde{f}(j)(e^{2\pi i jx/p}-1) + \tilde{f}(p-j)(e^{-2\pi i jx/p}-1) = -2\tilde{f}(j)(1-\cos(2\pi jx/p)) = -2\tilde{f}(p-j)(1-\cos(2\pi jx/p))$.  
    Setting $\hat{f}(j) = -\tilde{f}(j)$ we have $f(x) - f(0) = \sum_{j=1}^{p-1} (1-\cos(2\pi jx/p))\hat{f}(j)$.}
    Note that $\hat{f}$ is real since we have assumed $f$ to be real and symmetric.
    \item[Hahn Decomposition.] Decompose $\hat{f}=\hat{f}_+-\hat{f}_-$ where $\hat{f}_+$ and $\hat{f}_-$ are non-negative real functions.
    \item[\Levy{} Process Simulation.] Define functions $f_{+},f_{-}$ as the inverse Fourier transform of $\hat{f}_+,\hat{f}_-$ respectively. 
    Since the Fourier transform is linear, we have $f=f_+ -f_-$. By construction, $f_+,f_-$ are each characteristic exponents of some \Levy{} process. Estimate those function moments separately with \LevyTower{} sketches parameterized by $m$, 
    and take their difference to yield an estimate of the $f$-moment.
\end{description}

\begin{lemma}\label{lem:FHL}
Given any symmetric function 
$f:\Z \to \R_+$ with $f(0)=0$ and any stream 
$\mathbf{x} \in \Z^n$, 
let $\widehat{f(\mathbf{x})}$ be the estimator derived from 
the Fourier-Hahn-\Levy{} method with 
parameter $p>2\|\mathbf{x}\|_\infty$.
Then $\widehat{f(\mathbf{x})}$ has error
\begin{align*}
\left|\widehat{f(\mathbf{x})}-f(\mathbf{x})\right| = O\left(\frac{|f_{+}(\mathbf{x})|+|f_{-}(\mathbf{x})|}{\sqrt{m}}\right),
\end{align*}
with probability at least 98/100.
\end{lemma}

\begin{figure}
    \centering
    \includegraphics[scale=0.6]{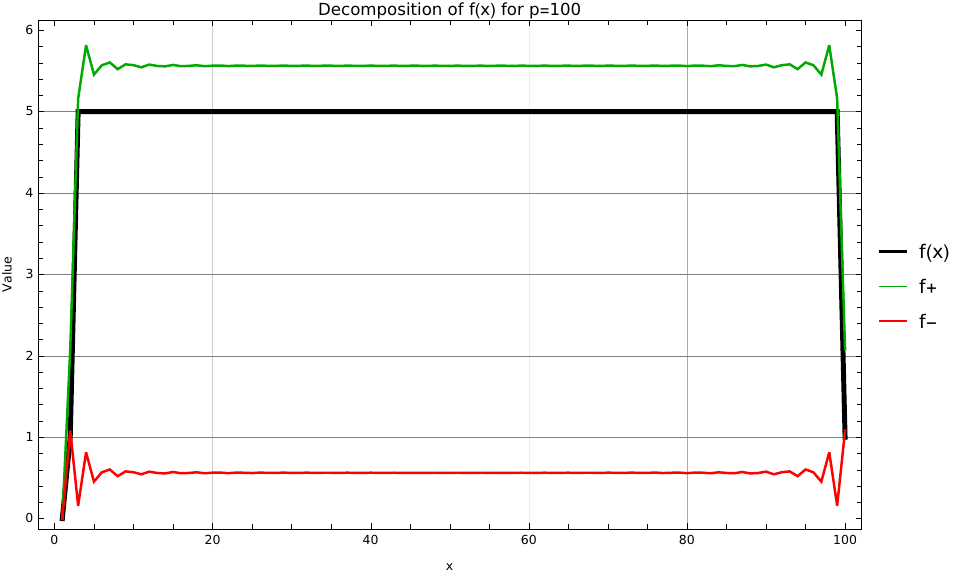}
    \caption{Here $f$ is the symmetric, periodic version of the 0-1-5 
    function over $\Z/p\Z$, namely 
    $f(0)=0, f(1)=f(99)=1, f(x)=5$ for $x\in [2,98]$.
    It is decomposed into $f = f_{+}-f_{-}$, 
    where in this case $f_{+}+f_{-}=\Theta(f)$.}
    \label{fig:FHL-015}
\end{figure}

If $f$ is the characteristic exponent of a \Levy{} process, 
then $f=f_{+}$ and $f_{-}=0$. 
This case was already considered, where we get an asymptotically 
unbiased estimate 
with multiplicative error. 
Clearly, whenever $|f_{+}(\mathbf{x})|+|f_{-}(\mathbf{x})| = O(|f(\mathbf{x})|)$, the Fourier-Hahn-\Levy{} method \emph{still} estimates the $f$-moment 
to within a $(1\pm\epsilon)$-factor, 
when $m=\Theta(\epsilon^{-2})$.
This handles the \emph{0-1-5 problem}, among others; 
see \cref{fig:FHL-015}.

\medskip 

\begin{figure}
    \centering
    % Top Row (NW and NE)
    \begin{subfigure}[b]{0.47\textwidth}
        \includegraphics[width=\textwidth]{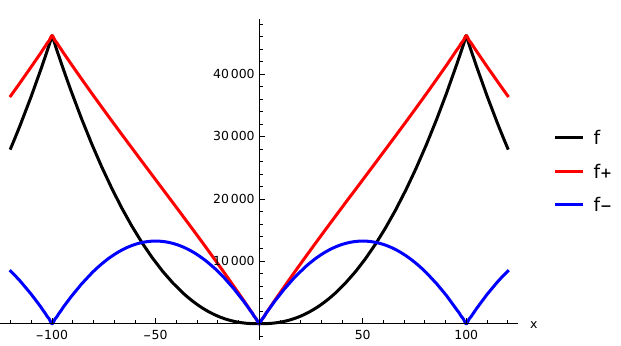}
        \caption{Decomposition of $f=f_+ - f_-$.}
        \label{fig:nw}
    \end{subfigure}
    \hfill
    \begin{subfigure}[b]{0.47\textwidth}
        \includegraphics[width=\textwidth]{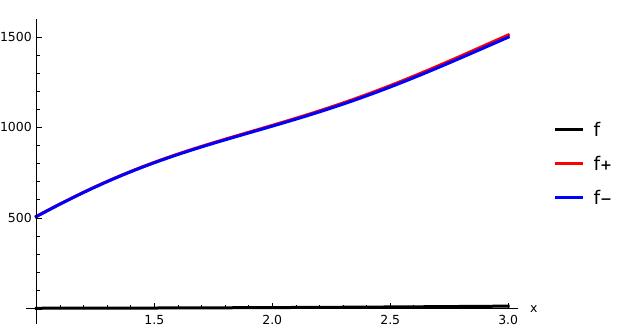}
        \caption{Decomposition of $f = f_+ - f_-$ over $[1,3]$.}
        \label{fig:ne}
    \end{subfigure}

    \vspace{1cm} % Vertical spacing between rows

    % Bottom Row (SW and SE)
    \begin{subfigure}[b]{0.47\textwidth}
        \includegraphics[width=.86\textwidth]{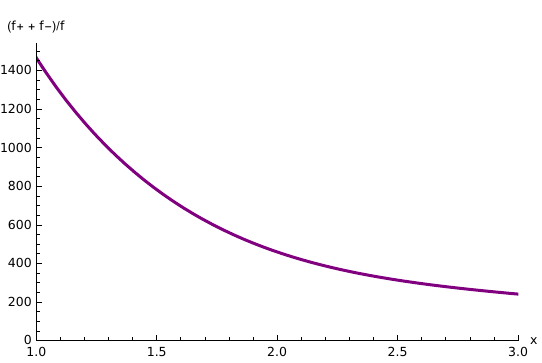}
        \caption{Ratio $(f_+ + f_-)f$ over $[1,3]$.}
        \label{fig:sw}
    \end{subfigure}
    \hfill
    \begin{minipage}[b]{0.45\textwidth}
        \centering
        \caption{The (periodic) Fourier-Hahn-\Levy{} method applied to $f(x) = x^2\log(1+|x|)$ over the domain $[-100,100]$.  Ther periodic version of $f$ is over $\Z/200\Z$.}
        \label{fig:FHL-f}
        \vspace{1.5cm} % Adjust this to align the caption vertically with the SW image
    \end{minipage}

\end{figure}

Let us explore how the Fourier-Hahn-\Levy{} method performs on functions 
that grow slightly faster than quadratic, say $f(x)=x^2\log(1+|x|)$.
\cref{fig:FHL-f} illustrates the decomposition of 
the periodic form $f_p(x) = |x|_p^2\log(1+|x|_p)$, 
where $n=100, p=200$, which agrees with $f$ on $\{-n,\ldots,n\}$. This decomposition has an impractically large ratio 
$\max_{x \in [1,n]} (f_+(x) + f_-(x))/f(x)$, 
which is over $1{,}400$ when $x=1$.

\medskip 

A better approach uses the fact that real \LK-representable functions are of the form
$Ax^2 + \int_0^{\infty} (1 - \cos(xs))\nu(ds)$.  
When $|x|\leq n, p=2n$, 
we may write $f(x) = x^2\log(1+n) - g_p(x)$, 
where $g_p(x) = |x|_p^2\log(\frac{1+n}{1+|x|_p})$ is non-negative over the relevant range $|x|\leq n$,
and $x^2\log(1+n)$ is already \LK-representable, $\log(1+n)$ being constant. 
Applying the periodic Fourier-Hahn-\Levy{} method to $g=g_p$ over 
$\Z/p\Z$, yields $g = g_+ - g_-$, where
$g_+,g_-$ are \LK-representable. 
Since \LK-representable functions are closed under 
addition we can decompose $f=f_+-f_-$ into \LK-representable functions such that
whenever $|x|\leq n$,
\[
f(x) = x^2\log(1+n) - g(x) 
     = (x^2\log(1+n) + g_-(x)) - g_+(x)
     = f_+(x) - f_-(x),
\]
where $f_+(x) = x^2\log(1+n) + g_-(x)$ and $f_-(x) = g_+(x)$.  
Note that $g_-$ and $g_+ = f_-$ are periodic, 
whereas $f_+$ is not. 
As one can see in \cref{fig:FHLplus}
this mixed-periodic decomposition 
has a substantially better ratio
$(f_+(x)+f_-(x))/f(x)$, but it is still impractical.

\begin{figure}
    \centering
    % Top Row (NW and NE)
    \begin{subfigure}[b]{0.48\textwidth}
        \includegraphics[width=\textwidth]{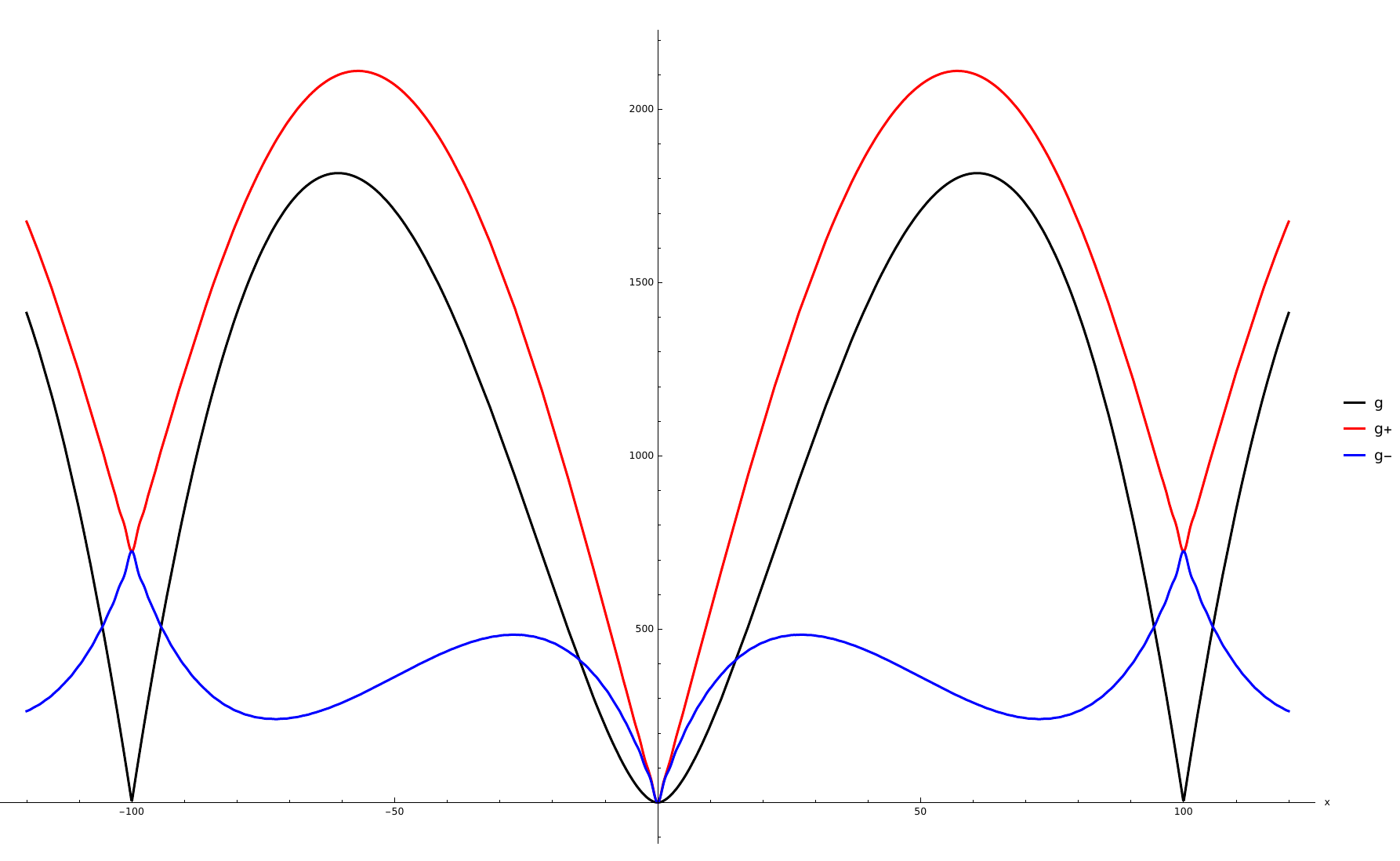}
        \caption{Decomposition of $g=g_+ - g_-$.}
        \label{fig:nw}
    \end{subfigure}
    \hfill
    \begin{subfigure}[b]{0.48\textwidth}
        \includegraphics[width=\textwidth]{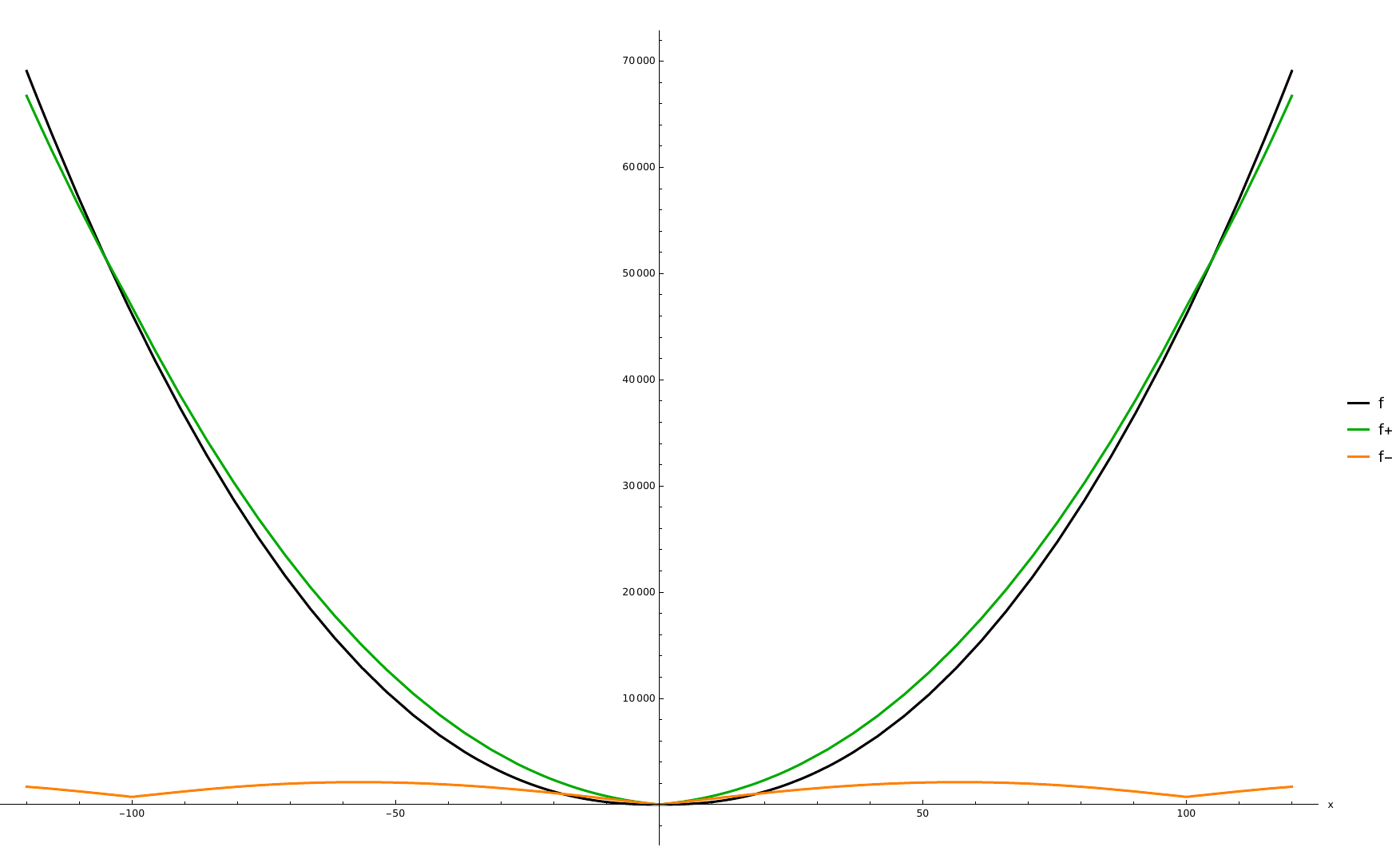}
        \caption{Decomposition of $f = f_+ - f_-$.}
        \label{fig:ne}
    \end{subfigure}

    \vspace{1cm} % Vertical spacing between rows

    % Bottom Row (SW and SE)
    \begin{subfigure}[b]{0.48\textwidth}
        \includegraphics[width=.86\textwidth]{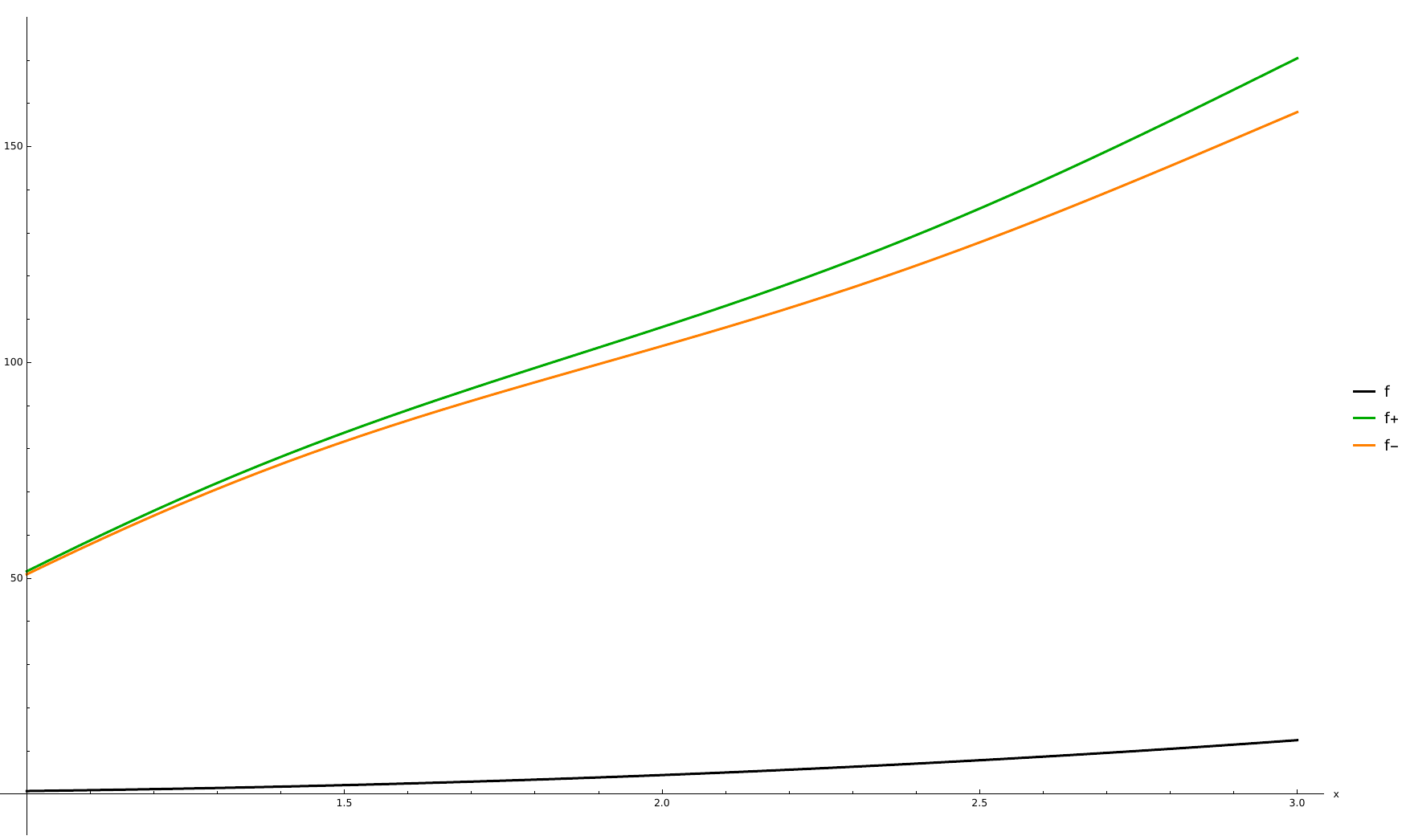}
        \caption{Decomposition of $f=f_+ - f_-$ over $[1,3]$.}
        \label{fig:sw}
    \end{subfigure}
    \hfill
    \begin{subfigure}[b]{0.48\textwidth}
        \includegraphics[width=.86\textwidth]{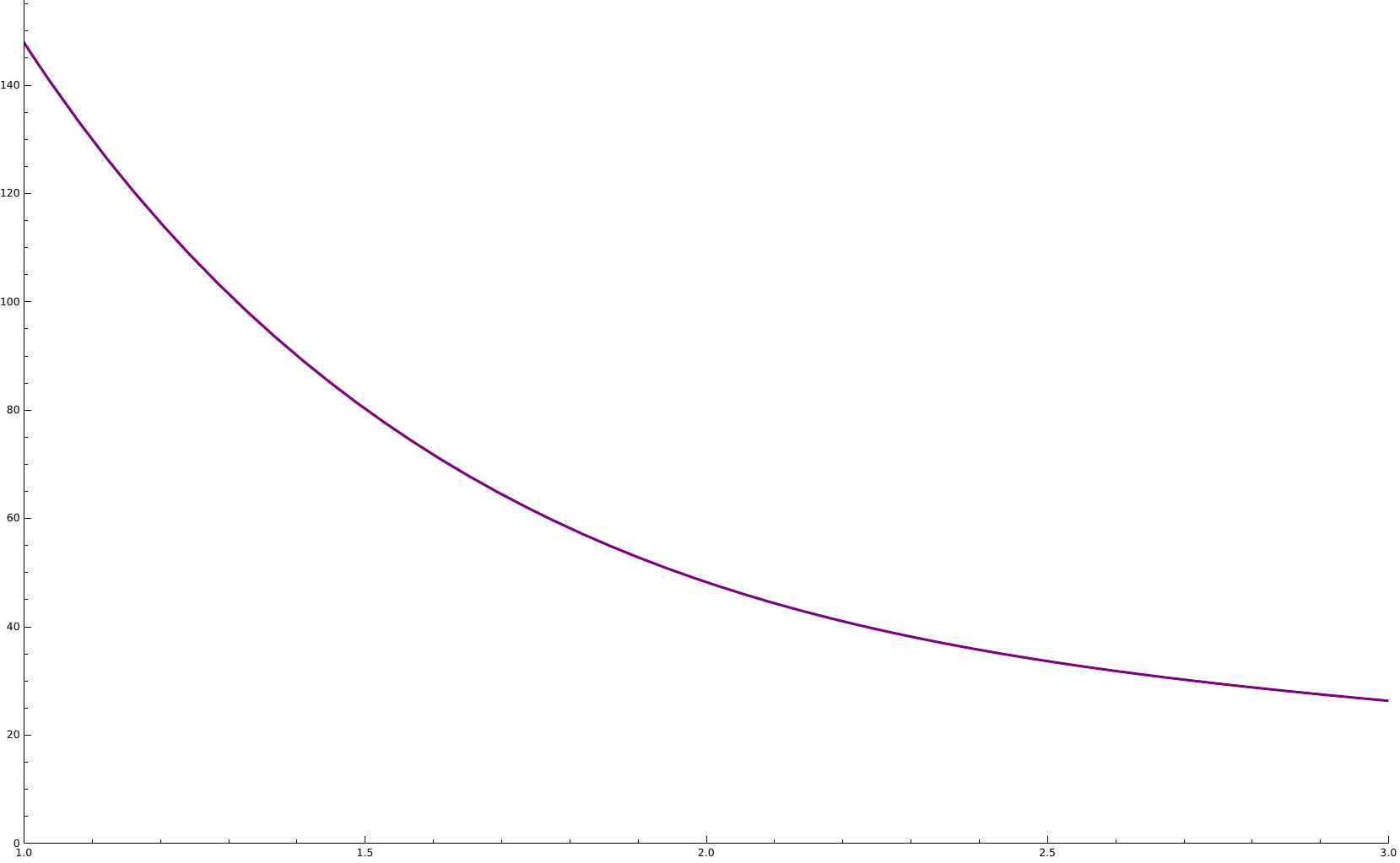}
        \caption{Ratio $(f_+ + f_-)f$.}
        \label{fig:se}
    \end{subfigure}
    
    \caption{An improved Fourier-Hahn decomposition of $f(x)=x^2\log(1+|x|)$ over $|x|\leq n$, $n=100$. 
    We write $f(x)=x^2\log(1+n) - g(x)$, and apply the Fourier-Hahn-\Levy{} method to decompose $g$ over $\Z/(2n)\Z$.}
        \label{fig:FHLplus}
\end{figure}

\subsubsection{The Approximate Fourier-Hahn-\Levy{} Method}

At this point we realize that it is counterproductive to consider only \emph{exact} 
decompositions of $f = f_+ - f_-$.  In order to $(1\pm \epsilon)$-approximate 
the $f$-moment 
$f(\mathbf{x})$ 
with a $\poly(\epsilon^{-1},\log n)$-bit sketch,
it suffices to find a function 
$g = g_+ - g_-$ 
such that $g_+,g_-$ are 
\LK-representable 
and for all relevant $|x|<n$, 
$g(x) \in [(1-\epsilon)f(x),(1+\epsilon)f(x)]$ 
is a good approximation of $f$
and
$g_+(x) + g_-(x) \leq \poly(\epsilon^{-1},\log n)f(x)$
is not \emph{too} much larger than $f$.
(It is not necessary that $g$ be periodic, nor that it have any prescribed behavior outside of $|x|<n$.)
Rather than do this analytically, we can generate $g,g_+,g_-$ by solving 
a linear program encoding the relevant constraints.  
Let $(\omega_k)$ be a sufficiently dense set
of frequencies.  
For approximating $f$, $\omega_k = (1+\epsilon)^{-k}$ for $k\in \{-\Omega(\epsilon^{-1}),\ldots,0,\ldots,\log_{1+\epsilon} n+O(1)\}$ 
suffices.  
Let $u_k,v_k$ be the non-negative coefficients of $g_+,g_-$ corresponding to $\omega_k$, and $u_*,v_*$ be the coefficients of $x^2$.  In other words, $g_+,g_-$ are
\[
g_+(x) = u_* x^2 + \sum_k u_k(1-\cos(\omega_k x))
\qquad \text{and}
\qquad
g_-(x) = v_* x^2 + \sum_k v_k(1-\cos(\omega_k x)).
\]
The linear program is then:
\begin{align*}
\text{Minimize $\rho$ subject to:} \qquad \qquad \qquad \qquad \qquad \qquad \qquad \qquad \quad&\\
\forall x\in \{1,\ldots,n\}.\quad (u_*-v_*)x^2 + \sum_{k} (u_k - v_k)(1-\cos(\omega_k x)) &\in [(1-\epsilon)f(x),(1+\epsilon)f(x)],\\
\forall x\in \{1,\ldots,n\}.\quad  (u_* + v_*)x^2 + \sum_{k} (u_k + v_k)(1-\cos(\omega_k x)) &\leq \rho f(x),\\
u_k, v_k &\geq 0.
\end{align*}

\begin{figure}
    \centering
    \begin{tabular}{ll}
    \includegraphics[height=0.22\linewidth]{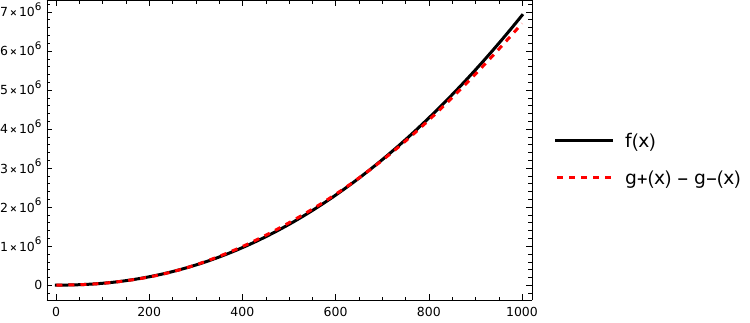} 
&   \includegraphics[height=0.22\linewidth]{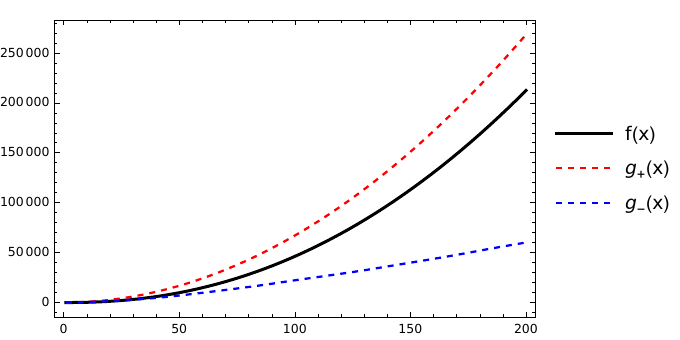}\\
\multicolumn{1}{c}{\small (a)}
& \multicolumn{1}{c}{\small (b)}\\
&\\
    \hcm{.3}\includegraphics[height=0.232\linewidth]{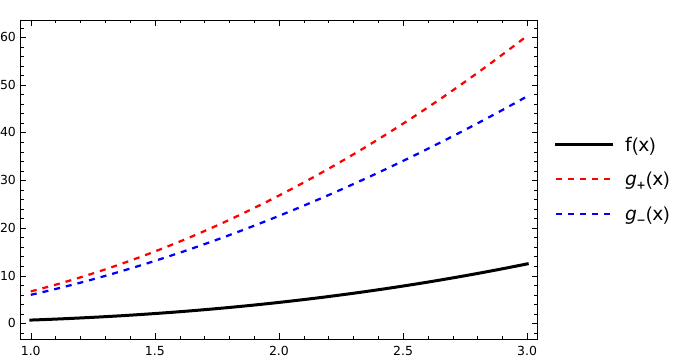}
&   \hcm{.3}\includegraphics[height=0.22\linewidth]{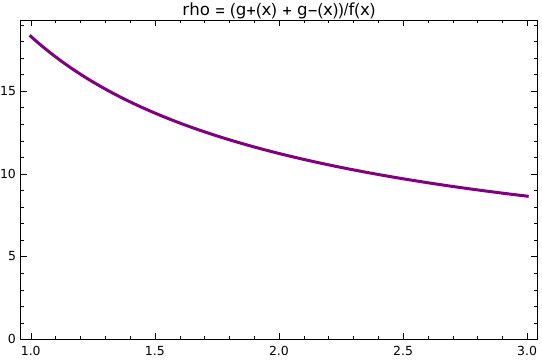}\\ 
\multicolumn{1}{c}{\small (c)}
& \multicolumn{1}{c}{\small (d)}\\
&\\
    \includegraphics[height=0.22\linewidth]{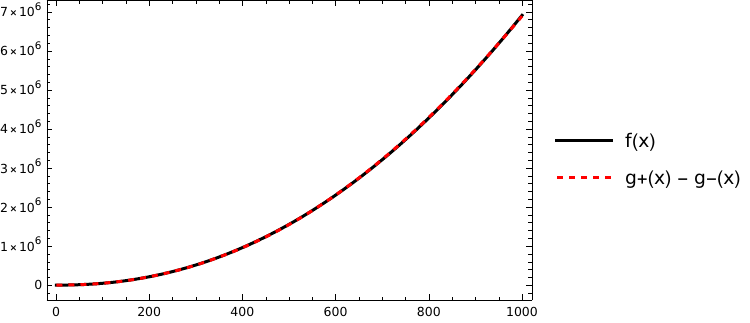} 
&   \includegraphics[height=0.22\linewidth]{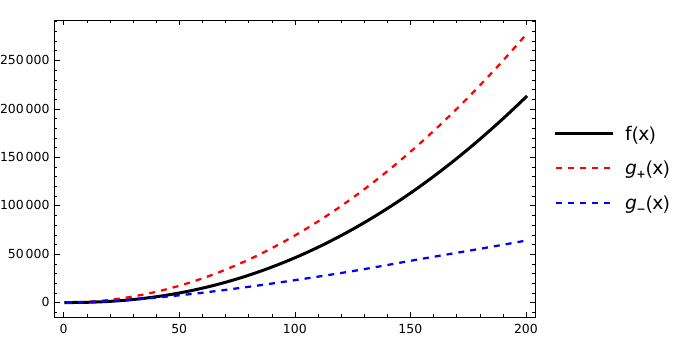}\\
\multicolumn{1}{c}{\small (e)}
& \multicolumn{1}{c}{\small (f)}\\
&\\
    \hcm{.3}\includegraphics[height=0.232\linewidth]{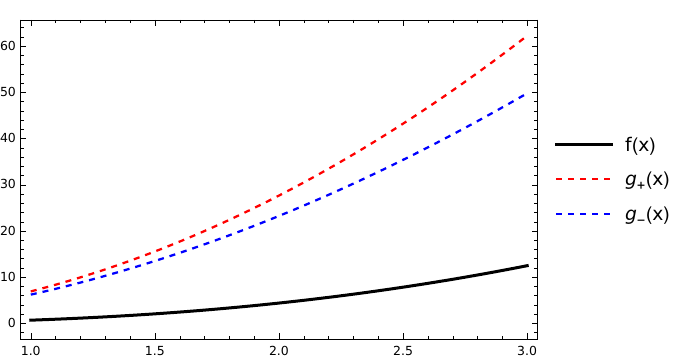}
&   \hcm{.3}\includegraphics[height=0.22\linewidth]{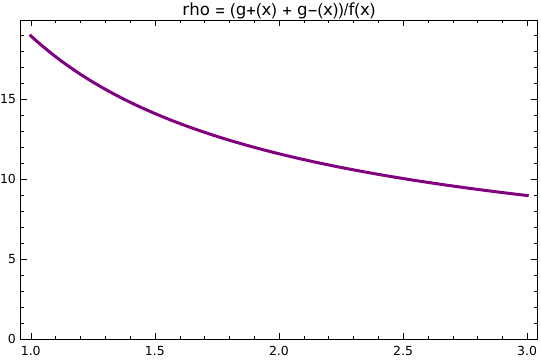}\\
\multicolumn{1}{c}{\small (g)}
& \multicolumn{1}{c}{\small (h)}\\
    \end{tabular}
    \caption{Plots (a)--(d) present the results of the LP construction of $g_+,g_-$ with $n=1000,\epsilon=0.03$ while plots (e)--(h)
    are for $n=1000,\epsilon=0.003$.  Plots (d,h) show that the \emph{cost} of the decomposition $(g_+(x)+g_-(x))/f(x) \leq \rho$ does not 
    depend heavily on $\epsilon$, for $\epsilon$ in this range.}
    \label{fig:LP-results}
\end{figure}

Figure~\ref{fig:LP-results} illustrates 
the results of this LP when $n=1000,\epsilon=0.03$,
plots (a)--(d), 
and $n=1000, \epsilon=0.003$, plots (e)--(h).
As one can see, the ratio $\rho = \max_{x \in [1,n]} (g_+(x)+g_-(x))/f(x)$ of these solutions (plots (d,h)) 
are much more reasonable than the (exact)
applications of the periodic or mixed-periodic 
Fourier-Hahn-\Levy{} method.

\section{Conclusion}\label{sec:conclusion}

In this paper we introduced a new way to study the \emph{tractability question} in data sketches, 
which has historically been studied alongside 
the \emph{universality question}~\cite{braverman2010zero,braverman2016streaming,BravermanC15,Wang25}.

\begin{description}
    \item[Tractability Question.] What is the class $\mathcal{T}$ of $f$-moments that can be $(1\pm\epsilon)$-estimated in $\poly(\epsilon^{-1},\log n)$ space?

    \item[Universality Question.] A sketch is 
    \emph{$\mathcal{C}$-universal} if it can $(1\pm \epsilon)$-estimate
    the $f$-moment, for any $f\in\mathcal{C}$.  Is there a $\mathcal{T}$-universal sketch?
\end{description}

We have demonstrated that the \LevyTower{} parameterized by a suitable \Levy{} process $(X_t)$ can estimate the $f_X$-moment, 
where $f_X$ is the characteristic exponent of $X$.  Moreover,
many existing sketches can be reinterpreted as sampling from \Levy{} processes.  At first glance this approach seems suited to exploring the tractability question, but \emph{ill-suited} to the universality question, for how can the \LevyTower$(X)$ be $\mathcal{T}$-universal if it is only built for estimating a \emph{single} $f_X$-moment?  
An unexpected outcome of this work is 
the revelation that \emph{some} 
\Levy{} processes $(X_t)$ ``leak'' lots of information about $f$-moments unrelated to $f_X$.  
In particular, when a characteristic exponent $f : \mathbb{R}\to\mathbb{R}^+$ is real, it can be expressed as
$f(x) = Ax^2 + \int_0^\infty (1-\cos(xs)\nu(ds))$ for some real $A$ and measure $\nu$; see \cref{rem:LK-1d-real-special-case}, and \cref{sect:xmodp,sect:nearlyperiodic}
for two examples.  In a companion paper by 
the second author, Wang~\cite{Wang25} 
demonstrates that a data structure similar to \LevyTower{} parameterized by a symmetric Poisson process (\textsf{SymmetricPoissonTower}) 
can estimate $F_2$-moments, 
all the \emph{basis moments} $(f_s)_{s>0}$, 
where $f_s(x)=1-\cos(sx)$,
and any linear mixture of the basis moments.
This sketch is universal for all the 
usual tractable real moments, and ``natively'' handles nearly periodic $f$-moments not amenable to $F_0$-sampling or $F_2$-heavy hitter sampling~\cite{braverman2010zero,braverman2016streaming,BravermanC15}.

\subsection{Conjectures}

Is the \LK{} representation theorem the final answer to the fundamental tractability~\cite{braverman2010zero,braverman2016streaming,BravermanC15} question? 
In \cref{sect:tractability-integers-vs-reals} we muddied the waters a bit by showing that
$\mathcal{T}[\Z]$ and $\mathcal{T}[\R]$ are different classes, i.e., tractability with respect to $\mathbf{x}\in \Z^n$ vs. $\mathbf{x}\in \R^n$ are two different concepts.
In \cref{sec:fourier-hahn-levy} we showed there exist simple, tractable functions with no \LK-representation, then demonstrated that such functions can still be handled in the \LevyTower/\LK{} framework via the \emph{periodic Fourier-Hahn-\Levy} method or the more general \emph{approximate Fourier-Hahn-\Levy} method. 
We believe the Fourier-Hahn-\Levy{} methods
are sufficiently powerful to capture all tractability classes $\mathcal{T}[M]$, where $\mathbf{x}\in M^n$ 
and $M$ is any of the domains of interest, 
$M\in \{\Z,\R,\Z^d,\R^d\}$.  We state this conjecture formally for $M=\Z$ or $\R$.

\begin{conjecture}[Tractability]\label{conjecture:FourierHahnLevy}
Suppose that $f\in\mathcal{T}[\Z]$ (respectively, $f\in\mathcal{T}[\R]$).
For any $\epsilon>0$ and parameter $n$ 
controlling the length of $\mathbf{x}\in \Z^n$ (respectively, $\R^n$) and $\|\mathbf{x}\|_\infty < n$, 
there exists a $g = g_+ - g_-$ satisfying the following criteria:
\begin{itemize}
\item $g_+$ and $g_-$ are \LK-representable, 
\item $g(x)\in [(1-\epsilon)f(x),(1+\epsilon)f(x)]$
for all $x\in \{-n,\ldots,0,\ldots,n\}$
(respectively, $x\in [-n,n]$), and
\item $(g_+(x) + g_-(x))/f(x) = O(\poly(\epsilon^{-1},\log n))$.
\end{itemize}
\end{conjecture}

It is possible that the connection between space-efficient $G$-samplers and the Laplace exponents of subordinators\footnote{(non-negative, one-dimensional \Levy{} processes)} is also natural, but here we have to be much more stringent in our space requirements to make a plausible conjecture.

\begin{conjecture}[Sampling]\label{conj:G-sampling}
    Suppose there is a $O(\log n)$-bit data structure for 
    $G$-sampling from a vector $\mathbf{x}\in \N^n$ subject to incremental updates.  Then $G$ is the Laplace exponent of a non-negative \Levy{} process on $\R_+ \cup \{\infty\}$.
\end{conjecture}

The key criterion of \cref{conj:G-sampling} is the $O(\log n)$ space bound, which permits techniques like storing the minimum hash value and index, but not more sophisticated techniques~\cite{monemizadeh20101,AndoniKO11,JowhariST11}.
The conjecture would be false if we loosened the space bound to $\polylog(n)$ bits, as $F_p$ sampling is possible within this bound for $p\in [0,2]$~\cite{jayaram2021perfect} 
but not in correspondence with a subordinator when 
$p\in (1,2]$~\cite{ken1999levy}.

%\bibliographystyle{plain}
%\bibliography{refs}

\appendix

\section{Sampling Without Replacement}\label{sec:wor}

We can take $k$ independent copies of the \LevyMinSampler{} or \ParetoSampler{} sketches to sample $k$ indices from the $(G(\mathbf{x}(v))/G(\mathbf{x}))_{v\in[n]}$ distribution \emph{with} replacement.  A small change to these algorithms will 
sample $k$ indices \emph{without} replacement.  
See Cohen, Pagh, and Woodruff~\cite{CohenPW20}
for an extensive discussion of why WOR (without replacement)
samplers are often more desirable in practice.
The algorithm $(G,k)$-\SamplerWOR (\cref{alg:G-sampler-WOR})
samples $k$ (distinct) indices without replacement.

\RestyleAlgo{ruled}
\begin{algorithm}[htbp]
  \SetAlgoLined\DontPrintSemicolon
  \SetKwFunction{algo}{algo}\SetKwFunction{proc}{proc}
  \SetKwFunction{activate}{Activate}
  \SetKwProg{update}{Update}{}{}
  \SetKwProg{sample}{Sample}{}{}
  \SetKwInOut{sketch}{Sketch}
  \SetKwInOut{hash}{Hash function}
  \sketch{$S\subset [n]\times \R_+$, initialized as $\emptyset$}
  \hash{$H:[n]\to \Uniform(0,1)$}
  \KwResult{Sample $k$ elements $u$ with prob.~$G(\mathbf{x}(u))/\sum_{v\in[n]}G(\mathbf{x}(v))$ \emph{without replacement}}
  
  \tcp{upon update $\mathbf{x}(v)\gets \mathbf{x}(v)+\Delta$}
  \update{$(v,\Delta)$}{
   $ S \gets \KMIN(S \cup \{(v,\ell_G(\Exp(\Delta),H(v)))\})$ \tcp*{$\Exp(\Delta)$ is freshly sampled}}
  \tcp{upon sample}
  \sample{$(\,)$}{
  \Return{$\{v : (v,\cdot)\in S \}$}
  }
  \caption{\LevyMinSampler{} without replacement. The function $\KMIN(L)$ takes a list $L\subset [n]\times \R_+$,
discards any $(v,h)\in L$ if there is a $(v,h')\in L$ with $h'<h$, then returns the $k$ elements with the smallest second coordinate.}\label{alg:G-sampler-WOR}
\end{algorithm}

In a similar fashion, one can define a \ParetoSampler{} without replacement
that maintains the minimum $k$-Pareto frontier,
defined by discarding any tuple $(a,b,v)$ if there is 
another $(a',b,v)$ with $a'<a$, then retaining
only those tuples that are dominated by at most $k-1$ other tuples.

\begin{theorem}\label{thm:WOR}
    Consider a stream of $\poly(n)$ incremental updates to a vector $\mathbf{x}\in\R_+^n$.
    The $(G,k)$-\SamplerWOR{} occupies $2k$ words of memory,
    and can report an ordered tuple $(v_*^{1},\ldots,v_*^{k}) \in [n]^k$
    such that 
\begin{align}
\pr((v_*^{1},\ldots,v_*^{k}) = (v^{1},\ldots,v^k))
= \prod_{i=1}^{k} \frac{G(\mathbf{x}(v^i))}{G(\mathbf{x}) - \sum_{j=1}^{i-1} G(\mathbf{x}(v^j))}.\label{eqn:WOR}
\end{align}
    The $k$-\ParetoSampler{} occupies $O(k\log n)$ words w.h.p.~and for any $G\in\mathcal{G}$ at query time,
    can report a tuple $(v_*^{1},\ldots,v_*^{k}) \in [n]^k$ distributed according to \cref{eqn:WOR}.
\end{theorem}

\begin{proof}
    The proof of Theorem \ref{thm:levy-min-sampler} shows that
    $h_v \sim \Exp(G(\mathbf{x}(v)))$ and if 
    $v_*^1$ minimizes $h_{v_*^1}$, that
    $h_{v_*^1} \sim \Exp(G(\mathbf{x}))$. 
    It follows that $\pr(v_*^1=v)=G(\mathbf{x}(v))/G(\mathbf{x})$.
    By the memoryless property of the exponential distribution,
    for any $v\neq v_*^1,$ 
    $h_v - h_{v_*^1} \sim \Exp(G(\mathbf{x}(v)))$,
    hence $h_{v_*^2} - h_{v_*^1} \sim \Exp(G(\mathbf{x})-G(\mathbf{x}(v_*^1)))$ and 
    $\pr(v_*^2 = v \mid v_*^1, v\neq v_*^1) = G(\mathbf{x}(v))/(G(\mathbf{x})-G(\mathbf{x}(v_*^1)))$.  The distribution of $v_*^3,\ldots,v_*^k$ is analyzed in the same way.

    By the 2D-monotonicity property, 
    the $k$-Pareto frontier contains all the 
    points that would be returned by $(G,k)$-\SamplerWOR,
    hence the output distribution of $k$-\ParetoSamplerWOR{}
    is identical.  The analysis of the space bound follows
    the same lines, except that $X_i$ is the indicator
    for the event that $h_{v_i}$ is among the $k$-smallest
    elements of $\{h_{v_1},\ldots,h_{v_i}\}$, so 
    $\E(X_i)=\min\{k/i,1\}$, $\E(|S|) < kH_n$, 
    and by a Chernoff bound, $|S|=O(k\log n)$ 
    with high probability.
\end{proof}

\end{document}